%% file: main.tex
\documentclass[12pt, vi, upcase, twoside]{mitthesis}
\usepackage{lgrind}
\usepackage{cmap}
\usepackage[T1]{fontenc}
\usepackage[english]{babel}
\usepackage[latin1]{inputenc}
\usepackage{textcomp}

\usepackage[T1]{fontenc}
\usepackage{lmodern}
\usepackage{xspace,xcolor,enumerate}
\usepackage{amssymb,amsmath}

\usepackage{pgf}
\usepackage{tikz}
\usepackage{array}
\usetikzlibrary{snakes}
\usepackage{pgfplots}

\usepackage{bm}
\usepackage{bbm}
\usepackage{amsmath}
\usepackage{hyperref}
\usepackage{mathpazo}
\usepackage{amsthm}

\DeclareMathOperator{\Tr}{Tr}

\usepackage{scalerel,stackengine}
\stackMath
\newcommand\reallywidehat[1]{%
	\savestack{\tmpbox}{\stretchto{%
			\scaleto{%
				\scalerel*[\widthof{\ensuremath{#1}}]{\kern-.6pt\bigwedge\kern-.6pt}%
				{\rule[-\textheight/2]{1ex}{\textheight}}%WIDTH-LIMITED BIG WEDGE
			}{\textheight}% 
		}{0.5ex}}%
	\stackon[1pt]{#1}{\tmpbox}%
}
\newcommand{\norm}[1]{\ensuremath{\lVert #1 \rVert}}

\allowdisplaybreaks

\newtheorem{theorem}{Theorem}[chapter]
\newtheorem{corollary}[theorem]{Corollary}
\newtheorem{lemma}[theorem]{Lemma}
\newtheorem*{claim}{Claim}
\newtheorem{definition}[theorem]{Definition}

\begin{document}

\include{cover}
% Some departments (e.g. 5) require an additional signature page.  See
% signature.tex for more information and uncomment the following line if
% applicable.
%\include{signature}
\pagestyle{plain}
\include{contents}

\include{chap1}
\include{chap2}
\include{chap3}

\include{chap4}
%\appendix
%\include{appa}
%\include{appb}
%\nocite{*}
\include{biblio}
\end{document}

%% file: cover.tex
% -*-latex-*-
%
% For questions, comments, concerns or complaints:
% thesis@mit.edu
%
%
% $Log: cover.tex,v $
% Revision 1.8  2008/05/13 15:02:15  jdreed
% Degree month is June, not May.  Added note about prevdegrees.
% Arthur Smith's title updated
%
% Revision 1.7  2001/02/08 18:53:16  boojum
% changed some \newpages to \cleardoublepages
%
% Revision 1.6  1999/10/21 14:49:31  boojum
% changed comment referring to documentstyle
%
% Revision 1.5  1999/10/21 14:39:04  boojum
% *** empty log message ***
%
% Revision 1.4  1997/04/18  17:54:10  othomas
% added page numbers on abstract and cover, and made 1 abstract
% page the default rather than 2.  (anne hunter tells me this
% is the new institute standard.)
%
% Revision 1.4  1997/04/18  17:54:10  othomas
% added page numbers on abstract and cover, and made 1 abstract
% page the default rather than 2.  (anne hunter tells me this
% is the new institute standard.)
%
% Revision 1.3  93/05/17  17:06:29  starflt
% Added acknowledgements section (suggested by tompalka)
%
% Revision 1.2  92/04/22  13:13:13  epeisach
% Fixes for 1991 course 6 requirements
% Phrase "and to grant others the right to do so" has been added to
% permission clause
% Second copy of abstract is not counted as separate pages so numbering works
% out
%
% Revision 1.1  92/04/22  13:08:20  epeisach

% NOTE:
% These templates make an effort to conform to the MIT Thesis specifications,
% however the specifications can change.  We recommend that you verify the
% layout of your title page with your thesis advisor and/or the MIT
% Libraries before printing your final copy.
\title{Combinatorial Optimization via the Sum of Squares hierarchy}

\author{Goutham Rajendran}
% If you wish to list your previous degrees on the cover page, use the
% previous degrees command:
%       \prevdegrees{A.A., Harvard University (1985)}
% You can use the \\ command to list multiple previous degrees
%       \prevdegrees{B.S., University of California (1978) \\
%                    S.M., Massachusetts Institute of Technology (1981)}
\department{Department of Computer Science}

% If the thesis is for two degrees simultaneously, list them both
% separated by \and like this:
% \degree{Doctor of Philosophy \and Master of Science}
\degree{Master's in Computer Science}

% As of the 2007-08 academic year, valid degree months are September,
% February, or June.  The default is June.
\degreemonth{May}
\degreeyear{2018}
\thesisdate{May 25, 2015}

%% By default, the thesis will be copyrighted to MIT.  If you need to copyright
%% the thesis to yourself, just specify the `vi' documentclass option.  If for
%% some reason you want to exactly specify the copyright notice text, you can
%% use the \copyrightnoticetext command.
%\copyrightnoticetext{\copyright IBM, 1990.  Do not open till Xmas.}
\copyrightnoticetext{\copyright 2018 by Goutham Rajendran. All rights reserved}

% If there is more than one supervisor, use the \supervisor command
% once for each.
\thesisadvisor{Madhur Tulsiani}{Assistant Professor}
\coadvisor{Janos Simon}{Professor}

% This is the department committee chairman, not the thesis committee
% chairman.  You should replace this with your Department's Committee
% Chairman.
\chairman{Arthur C. Smith}{Chairman, Department Committee on Graduate Theses}

% Make the titlepage based on the above information.  If you need
% something special and can't use the standard form, you can specify
% the exact text of the titlepage yourself.  Put it in a titlepage
% environment and leave blank lines where you want vertical space.
% The spaces will be adjusted to fill the entire page.  The dotted
% lines for the signatures are made with the \signature command.
\maketitle

% The abstractpage environment sets up everything on the page except
% the text itself.  The title and other header material are put at the
% top of the page, and the supervisors are listed at the bottom.  A
% new page is begun both before and after.  Of course, an abstract may
% be more than one page itself.  If you need more control over the
% format of the page, you can use the abstract environment, which puts
% the word "Abstract" at the beginning and single spaces its text.

\begin{titlepage}
	\vfil
	Copyright \textcircled{c} 2018 by Goutham Rajendran\\ All rights reserved
\end{titlepage}

%% You can either \input (*not* \include) your abstract file, or you can put
%% the text of the abstract directly between the \begin{abstractpage} and
%% \end{abstractpage} commands.

% First copy: start a new page, and save the page number.
\clearpage
% Uncomment the next line if you do NOT want a page number on your
% abstract and acknowledgments pages.
% \pagestyle{empty}
%\setcounter{savepage}{\thepage}
\begin{abstractpage}
\input{abstract}
\end{abstractpage}

% Additional copy: start a new page, and reset the page number.  This way,
% the second copy of the abstract is not counted as separate pages.
% Uncomment the next 6 lines if you need two copies of the abstract
% page.
% \setcounter{page}{\thesavepage}
% \begin{abstractpage}
% \input{abstract}
% \end{abstractpage}

\cleardoublepage

\section*{Acknowledgments}

I thank Prof. Madhur Tulsiani for introducing me to this beautiful subject and many interesting concepts, and for giving me a lot of useful ideas and suggestions regarding this work.

I wish to express my gratitude to Prof. L\'{a}szl\'{o} Babai for imparting his wisdom and giving me countless advice; and Prof. Janos Simon for being a constant source of support.

I am grateful to Prof. Aravindan Vijayaraghavan for helpful discussions on Densest $k$-subgraph and Prof. Pravesh Kothari for communicating, through Madhur, the idea of how to use pseudocalibration to obtain hardness results for Max $K$-CSP.

This work would not have been possible without the support of my friends and family, especially my parents.

%%%%%%%%%%%%%%%%%%%%%%%%%%%%%%%%%%%%%%%%%%%%%%%%%%%%%%%%%%%%%%%%%%%%%%
% -*-latex-*-

%% file: abstract.tex
% $Log: abstract.tex,v $
% Revision 1.1  93/05/14  14:56:25  starflt
% Initial revision
% 
% Revision 1.1  90/05/04  10:41:01  lwvanels
% Initial revision
% 
%
%% The text of your abstract and nothing else (other than comments) goes here.
%% It will be single-spaced and the rest of the text that is supposed to go on
%% the abstract page will be generated by the abstractpage environment.  This
%% file should be \input (not \include 'd) from cover.tex.

We study the Sum of Squares (SoS) Hierarchy with a view towards combinatorial optimization. We survey the use of the SoS hierarchy to obtain approximation algorithms on graphs using their spectral properties. We present a simplified proof of the result of Feige and Krauthgamer on the performance of the hierarchy for the Maximum Clique problem on random graphs. We also present a result of Guruswami and Sinop that shows how to obtain approximation algorithms for the Minimum Bisection problem on low threshold-rank graphs.

We study inapproximability results for the SoS hierarchy for general constraint satisfaction problems and problems involving graph densities such as the Densest $k$-subgraph problem. We improve the existing inapproximability results for general constraint satisfaction problems in the case of large arity, using stronger probabilistic analyses of expansion of random instances. We examine connections between constraint satisfaction problems and density problems on graphs. Using them, we obtain new inapproximability results for the hierarchy for the Densest $k$-subhypergraph problem and the Minimum $p$-Union problem, which are proven via reductions. 

We also illustrate the relatively new idea of pseudocalibration to construct integrality gaps for the SoS hierarchy for Maximum Clique and Max $K$-CSP. The application to Max $K$-CSP that we present is known in the community but has not been presented before in the literature, to the best of our knowledge.

%% file: contents.tex
  % -*- Mode:TeX -*-
%% This file simply contains the commands that actually generate the table of
%% contents and lists of figures and tables.  You can omit any or all of
%% these files by simply taking out the appropriate command.  For more
%% information on these files, see appendix C.3.3 of the LaTeX manual. 
\tableofcontents
%\newpage
%\listoffigures
%\newpage
%\listoftables

%% file: chap1.tex
%% This is an example first chapter.  You should put chapter/appendix that you
%% write into a separate file, and add a line \include{yourfilename} to
%% main.tex, where `yourfilename.tex' is the name of the chapter/appendix file.
%% You can process specific files by typing their names in at the 
%% \files=
%% prompt when you run the file main.tex through LaTeX.
\chapter{Introduction}

The famous Cook-Levin theorem showed the existence of at least one NP-hard problem, namely the Boolean satisfiability problem. Using reductions, many natural problems that are interesting have been found to be NP-hard, which means that an efficient algorithm to these problems would essentially prove P = NP. So, assuming that P $\neq$ NP, the focus has been on trying to find efficient algorithms, which could possibly be randomized, that give good approximation guarantees.

We study optimization problems where we are given an instance $I$ and we would like to compute the optimum value of the objective function, denoted $OPT$, over feasible solutions. The optimization could be either to maximize or minimize the value. An $\alpha$-approximation algorithm for $\alpha \le 1$ for a maximization (resp. minimization) problem is an efficient algorithm that finds a solution for any instance $I$ with value at least $\alpha\cdot OPT$ (resp. at most $\frac{1}{\alpha}\cdot OPT$).  Even when $\alpha > 1$, we use the term $\alpha$-approximation algorithm for a maximization (resp. minimization) problem to mean an efficient algorithm that finds a solution for any instance $I$ with value at least $\frac{1}{\alpha}\cdot OPT$ (resp. at most $\alpha\cdot OPT$). Note that this double definition is essentially to avoid the convention that the approximation factor is either at most $1$ or at least $1$ and instead use them interchangeably. Here, efficient algorithm means that it's running time is polynomial in the size of the instance $I$ and note that the algorithm could be randomized.

A plethora of techniques have been introduced towards this objective and two of the crucial techniques are Linear programming and the related Semidefinite programming, which are powerful because they can be applied to a variety of problems, with a single framework.

\section{Linear Programming and Semidefinite programming}

A linear program is an optimization problem of the following form:
\begin{align*}
\text{Maximize}\qquad&\bm{c}^T\bm{x}&\\
\text{subject to}\qquad&A\bm{x} \le \bm{b}&\\
&\bm{x} \in \mathbb{R}^n&
\end{align*}
%\begin{eqnarray*}
%\text{Maximize }&\bm{c}^T\bm{x}\\
%\text{subject to }&A\bm{x} \le \bm{b}\\
%&\bm{x} \in \mathbb{R}^n
%\end{eqnarray*}

Here, $A \in \mathbb{R}^{m \times n}, \bm{b} \in \mathbb{R}^m$. Linear programs can be solved in polynomial time using the ellipsoid method or the interior point method. When the condition $\bm{x}\in \mathbb{R}^n$ is replaced by $\bm{x}\in \mathbb{Z}^n$, we call it an integer program. Integer programming is NP-hard. Many approximation algorithms start by considering an integer program to a given problem, relaxing it to a linear program, solving it and then rounding the solutions to integers and finally proving that this rounding achieves good approximation guarantees.

To explain semidefinite programming, we need to define positive semidefinite matrices.

\begin{definition}
A symmetric matrix $A \in \mathbb{R}^{n \times n}$ is said to be positive semidefinite, denoted $A \succeq 0$, if any of the following equivalent conditions hold.
\begin{itemize}
	\item $A = X^TX$ for some $X \in \mathbb{R}^{d \times n}, d\le n$
	\item All eigenvalues of $A$ are nonnegative
	\item $\bm{x}^TA\bm{x} \ge 0$ for all $\bm{x} \in \mathbb{R}^n$
\end{itemize}
\end{definition}

A Semidefinite program (SDP) has $n^2$ variables $y_{1, 1}, y_{1, 2}, \ldots, y_{n, n}$ which can be thought of to form a matrix $Y \in \mathbb{R}^{n \times n}$. Then, the objective is of the following form:
\begin{align*}
\text{Maximize}\qquad&C\bullet Y&\\
\text{subject to}\qquad&A_i\bullet Y \le b_i& \forall i = 1, 2, \ldots, m\\
&Y \succeq 0&\\
&Y \in \mathbb{R}^{n\times n}&
\end{align*}
%\begin{eqnarray*}
%\text{Maximize }&C\bullet X&\\
%\text{subject to }&A_i\bullet X \le b_i&\\
%&X \succeq 0&\\
%&X \in \mathbb{R}^{n\times n}&
%\end{eqnarray*}

Here, $C, A_1 \ldots, A_m \in \mathbb{R}^{n \times n}, b_1, \ldots, b_m \in \mathbb{R}^n$. Also, note that "$\bullet$" denotes entrywise dot product, that is $C \bullet Y = \displaystyle\sum_{i = 1}^n\displaystyle\sum_{j = 1}^n C_{i, j}Y_{i, j}$. So, it is a linear program in the entries of $Y$ with the additonal constraint that $Y$ is positive semidefinite. Note that since $Y$ has to be positive semidefinite, it also has to be symmetric and so, there are essentially only $n(n+1)/2$ variables.

It is a famous result of Gr\"{o}tschel, Lov\'{a}sz and Schrijver\cite{oracle} that SDPs can be solved in polynomial time, under some mild assumptions. We remark that, by solved, we mean that for any constant $\epsilon > 0$, we can get an additive $\epsilon$-approximation in polynomial time. It may not be possible to find the exact solution because the exact solution may be irrational.

To show an example of how SDPs can be useful, consider the Maximum Cut problem. In this problem, we are given an undirected unweighted graph $G = (V, E)$ and we would like to find a partition $(S, V - S)$ of the vertex set so that the number of edges with exactly one endpoint in $S$, is maximized. This problem is NP-hard. The best known approximation algorithm for this problem due to Goemans and Williamson\cite{gw} uses semidefinite programming.

Consider the following program over integers for Max-Cut. For each vertex $u \in V$, introduce the variable $x_u$ which takes the value $1$ when $u \in S$ and $-1$ when $u \not\in S$. The constraint $x_u^2 = 1$ enforces $x_u = \pm 1$ and for each edge $(u, v)$, observe that the expression $\left(\frac{1}{2} - \frac{1}{2}x_ux_v\right)$ indicates whether that edge is in the cut. So, Max-Cut is equivalent to the following optimization problem.
\begin{align*}
\text{Maximize}\qquad&\displaystyle\sum_{(u, v)\in E}\left(\frac{1}{2} - \frac{1}{2}x_ux_v\right)&\\
\text{subject to}\qquad&\qquad x_u^2 = 1&\\
&\qquad x_u \in \mathbb{R}&
\end{align*}
%\begin{eqnarray*}
%\text{Maximize }&\displaystyle\sum_{(u, v)\in E}&\left(\frac{1}{2} - \frac{1}{2}x_ux_v\right)\\
%\text{subject to }&&x_u^2 = 1\\
%&&x_u \in \mathbb{R}
%\end{eqnarray*}

This is an instance of a quadratic program. Unfortunately, quadratic programs are NP-hard. Indeed, the above is a reduction from Max-Cut to quadratic programs. Goemans and Williamson relaxed the above program to a semidefinite program which can be efficiently solved and they showed a rounding algorithm which achieves a good approximation. The relaxation is to replace the real numbers $x_u$ by vectors $\bm{V}_u$ of arbitrary dimension. That is, we relax $x_u \in \mathbb{R}$ to $\bm{V}_u \in \mathbb{R}^d$ for some positive integer $d$. Then, replace all products $x_ux_v$ with the standard inner product $\langle\bm{V}_u, \bm{V}_v\rangle$.

The new program for Max-Cut is as follows.
\begin{align*}
\text{Maximize}\qquad&
\displaystyle\sum_{(u, v)\in E}\left(\frac{1}{2} - \frac{1}{2}\langle\bm{V}_u, \bm{V}_v\rangle\right)&\\
\text{subject to}\qquad&\qquad\langle\bm{V}_u, \bm{V}_u\rangle = 1&\\
&\qquad\quad\bm{V}_u \in \mathbb{R}^d&
\end{align*}
%\begin{eqnarray*}
%\text{Maximize}&
%\displaystyle\sum_{(u, v)\in E}&\left(\frac{1}{2} - \frac{1}{2}\langle\bm{V}_u, \bm{V}_v\rangle\right)\\
%\text{subject to} &&\langle\bm{V}_u, \bm{V}_u\rangle = 1\\
%&&\bm{V}_u \in \mathbb{R}^d
%\end{eqnarray*}

The program in the form above is called a vector program. Note that we just need to ensure that $d$ exists, but don't specify its value beforehand. To solve this, we introduce $n^2$ variables $y_{u, v}$ for all vertices $u, v$ and replace all $\langle\bm{V}_u, \bm{V}_v\rangle$ with $y_{u, v}$. Then, observe that the above program can be written as a linear program in $y_{u, v}$. The only catch is that, the solution to this program $y_{u, v}$ should be such that there exist vectors $\bm{V}_u$ in $\mathbb{R}^d$ for some $d$ such that $y_{u, v} = \langle\bm{V}_u, \bm{V}_v\rangle$. This is precisely the condition that $Y = (y_{u, v})$ is positive semidefinite. If we add this constraint to the program, we have a semidefinite program in $Y$ that we can solve.

Once we find $Y$, we can efficiently find the actual vectors $\bm{V}_u \in \mathbb{R}^d$ (known as the Cholesky decomposition) and the final rounding algorithm is as follows: Sample a random unit vector $\bm{g}$ in $\mathbb{R}^d$. The rounding sets $x_u = sgn(\langle \bm{g}, \bm{V}_u\rangle)$ where $sgn(x)$ is $1$ if $x \ge 0$ and $-1$ if $x < 0$. The partition corresponding to these $x_u$s is precisely the partition that we output, that is, we output $S = \{u \in V \; | \; x_u = 1\}$.

Goemans and Williamson\cite{gw} proved that this randomized rounding achieves $\alpha_{GW} \approx 0.87856$ approximation. Feige and Schechtman\cite{feige} proved that the above analysis is optimal for this SDP. Moreover, Khot et al.\cite{kkmo} proved that this is the best approximation algorithm possible for this problem assuming the Unique Games Conjecture.

\section{Hierarchies}

Hiearchies are sequences of progressively stronger relaxations of linear or semidefinite programs which are obtained by adding more consistency constraints that an actual solution would satisfy. These are not problem specific and in general, could be done for most problems where the program variables take values in $\{0, 1\}$. In the hierarchies we study, the relaxed variables encode the probability of a subset of original variables being assigned $1$ in the optimum solution. Although we lose in running time by adding more constraints, we will still have polynomial running time if we add only polynomially many constraints.

Linear programming hierarchies were studied by Lov\'{a}sz and Schrijver\cite{ls}; and Sherali and Adams\cite{sa}. The semidefinite programming hierarchies were studied by Shor\cite{shor}, Nesterov\cite{nesterov}, Parrillo\cite{parr} and Lasserre\cite{las}. It is known as the Sum of Squares(SoS) hierarchy, which will be the focus of our thesis. Although it is defined for generic polynomial optimization, we will study mainly the SDP formulation also known as the Lasserre hierarchy.

It is known that the SoS hierarchy is at least as powerful as the Lov\'{a}sz-Schrijver or Sherali-Adams hierarchies. We generally try to prove approximation guarantees by considering the weakest possible hierarchy that will ensure that guarantee; and we prove hardness results for the strongest possible hierarchies. There are other intermediate hierarchies that have been studied, but we will not consider them here.
%In general, the hierarchies can be thought of as lift-and-project methods or as proof systems.

The performance of a program can be quantified by its integrality gap. Suppose the actual optimum to an instance $I$ of a maximization problem is $OPT$ and the program returns optimum value $FRAC \ge OPT$, then the integrality gap for this instance is defined to be $\frac{FRAC}{OPT}$. The maximum value of this quantity over all instances of a fixed size is the integrality gap of the program and measures how good the program performs in the worst case. This can similarly be defined for minimization problems. An integrality gap of $1$ means the program exactly solves the given problem. We can prove large integrality gaps for these hierarchies for some natural problems, providing evidence of their intrinsic hardness.

\section{Thesis Organization}

In this thesis, we provide a short exposition of the Sum of Squares hierarchy as well as obtain new results, mainly for combinatorial problems. 

In Chapter $2$, we define the hierarchy and give a flavor of the algorithmic results that can be obtained. We study the performance of the SoS hierarchy for the Maximum Clique problem on random graphs. In particular, we present the relevant result of Feige and Krauthgamer\cite{feige} and present a variant of their proof using the stronger SoS hierarchy (see Section \ref{sosproof})), instead of their original proof which uses the weaker Lov\'{a}sz-Schrijver hierarchy. We then give an exposition of Guruswami and Sinop's\cite{gs} approximation algorithm, via the SoS hierarchy, for the Minimum Bisection problem when the instance is a low threshold-rank graph.

In Chapter 3, we present SoS hierarchy lower bounds for general Constraint Satisfaction problems (Max $K$-CSP) due to Kothari et al.\cite{kmow} and show how they can be used to obtain lower bounds for Densest $k$-subgraph and it's variants. Then, we present an alternate view of the SoS hierarchy using pseudoexpectation operators and formally show the equivalence to this alternate view. 

Finally, we illustrate the powerful idea of pseudocalibration to construct lower bounds for the SoS hierarchy for Maximum Clique and Max $K$-CSP. The idea was introduced and applied to Maximum Clique by Barak et al.\cite{pseudo} but we present a slightly different explanation from the one in their paper. We also show that we can alternatively use pseudocalibration to arrive at the integrality gap construction of Kothari et al.\cite{kmow} for Max $K$-CSPs (see Section \ref{pcalibration}), as opposed to their purely combinatorial approach. This application is fairly well-known in the community but has not been presented anywhere in the literature, to the best of our knowledge.

We also exhibit new results. We improve the existing SoS hardness results for the Max $K$-CSP problem in the case when $K$ grows as a function of the instance size (see Corollary \ref{csp_hardness}). We obtain new hardness results for Densest $k$-subhypergraph (see Theorem \ref{hyper}) and Minimum $p$-Union (see Corollary \ref{union}). The former is a reduction from SoS hardness of Densest $k$-subgraph and the latter is a reduction from SoS hardness of Max $K$-CSPs. To the best of our knowledge, no prior SoS hardness results were known for either of these problems.

%% file: chap2.tex
%% This is an example first chapter.  You should put chapter/appendix that you
%% write into a separate file, and add a line \include{yourfilename} to
%% main.tex, where `yourfilename.tex' is the name of the chapter/appendix file.
%% You can process specific files by typing their names in at the 
%% \files=
%% prompt when you run the file main.tex through LaTeX.
\chapter{The Sum of Squares Hierarchy}

We will first provide a rough outline of the hierarchy. Suppose we have a program with variables $\{x_i\}_{i \le n}$ where we need to optimize over $x_i \in \{0, 1\}$. We would like to write a new program with possibly more variables whose optimal solution, restricted to the relevant variables, is a convex combination of optimal integer solutions $\{x_i\}_{i \le n}$ to the original program. For this purpose, we can consider vectors $\bm{V_{i}}$ that essentially capture these variables so that $\norm{\bm{V_{i}}}^2$ will be expected value of $x_i$ under such a distribution. More generally, for some predetermined integer $r$, for every subset $S$ of the variable indices of size at most $r$, we introduce vectors $\bm{V}_S$ which can be thought of as to capture the event that every variable with index in $S$ is in the optimum solution, that is, it represents $\displaystyle\prod_{i \in S} x_i$. So, $\norm{\bm{V}_S}^2$ is intended to be the probability that every variable with subscript in $S$ is $1$ in the final solution. These $\bm{V}_S$ are known as local variables. We set $\norm{\bm{V}_{\phi}}^2 = 1$ because the empty event should ideally have probability $1$.

Now, for all $i, j$, terms such as $x_ix_j$ would be replaced by $\langle \bm{V}_{\{i\}}, \bm{V}_{\{j\}}\rangle$. But notice that we could also have replaced it by $\langle \bm{V}_{\{i, j\}}, \bm{V}_{\phi}\rangle$. To rectify this situation, we would add the constraint $\langle \bm{V}_{\{i\}}, \bm{V}_{\{j\}}\rangle = \langle \bm{V}_{\{i, j\}}, \bm{V}_{\phi}\rangle$. More generally, we add the constraint $\langle \bm{V}_{S_1}, \bm{V}_{S_2}\rangle = \langle \bm{V}_{S_3}, \bm{V}_{S_4}\rangle$ for all sets $S_1, S_2, S_3, S_4$ of size at most $r$ such that $S_1 \cup S_2 = S_3 \cup S_4$. These are known as local consistency constraints. In a sense, they ensure that the vectors $\bm{V}_S$ mimic an actual probability distribution. Once we have these constraints, terms like $x_1x_3x_4$ could be replaced by any of $\langle \bm{V}_{\{1, 3\}}, \bm{V}_{\{4\}}\rangle$ or $\langle \bm{V}_{\phi}, \bm{V}_{\{1, 3, 4\}}\rangle$ or $\langle \bm{V}_{\{1, 4\}}, \bm{V}_{\{3\}}\rangle$. We also add the constraints $\langle \bm{V}_{S_1}, \bm{V}_{S_2}\rangle \ge 0$ for all sets $S_1, S_2$ of size at most $r$, as would be satisfied by actual distributions. We remark that if $|S_1 \cup S_2| \le r$, then these follow from the previous constraints since $\langle \bm{V}_{S_1}, \bm{V}_{S_2}\rangle = \langle \bm{V}_{S_1 \cup S_2}, \bm{V}_{S_1 \cup S_2}\rangle = \norm{\bm{V}_{S_1 \cup S_2}}^2 \ge 0$. Finally, any constraint for the given problem is replaced by many extra constraints on these new variables that conform to our interpretation. For example $x_1x_3 + x_5 \le 10$ would be replaced by the constraints $\langle \bm{V}_S, \bm{V}_{\{1, 3\}}\rangle + \langle \bm{V}_S, \bm{V}_{\{5\}}\rangle \le 10\langle \bm{V}_S, \bm{V}_{\phi}\rangle$ for all sets $S$ with $|S| \le r$. Here, we assume $r \ge 2$ since the variable $\bm{V}_{\{1, 3\}}$ doesn't exist otherwise.

\section{The SoS relaxation for boolean programs}

Now we will describe the relaxation in a general setting, following the above intuition. This is slightly restricted but should suffice for most applications. Suppose we have an program over the variables $x_1, \ldots, x_n$ of the form below:
\begin{align*}
	\mbox{Maximize}\qquad&p(x_1, \ldots, x_n)&\\
	\text{subject to}\qquad&q_i(x_1, \ldots, x_n) \ge 0\quad\qquad i = 1, 2, \ldots, m&\\
	& x_i \in \{0, 1\}&
\end{align*}
%\begin{eqnarray*}
%	&\text{Maximize }p(x_1, \ldots, x_n)&\\
%	\text{subject to }&q_i(x_1, \ldots, x_n) \ge 0& i = 1, 2, \ldots, m\\
%	& x_i \in \{0, 1\}&
%\end{eqnarray*}

Here, $p$ and $q_1, \ldots, q_m$ are polynomials. Since $x_i \in \{0, 1\}$, we have that $x_i^2 = x_i$ and so, we can assume without loss of generality that $p, q_1, \ldots, q_m$ are multilinear. Let $r$ be any integer which is at least the degree of $p$ and at least the degree of $q_i$ for all $i \le m$. For $T \subseteq [n]$ denote by $x_T$ the product $\displaystyle\prod_{i \in T} x_i$. Also, define $[n]_{\le r} = \{T \subseteq [n] \;|\; |T| \le r\}$ to be the set of subsets with at most $r$ elements. Then, we can write $p(x_1, \ldots, x_n) = \displaystyle\sum_{T \in [n]_{\le r}} p_Tx_T, q_i(x_1, \ldots, x_n) = \displaystyle\sum_{T \in [n]_{\le r}} (q_i)_Tx_T$ uniquely.
\pagebreak[4]

\begin{definition}
We define a level $r$ SoS relaxation to be the following vector program with variables $\bm{V}_S$ for $S \in [n]_{\le r}$:
%For each $i \le m$, define $M_i = \{T \in [n]_{\le r}|(q_i)_T \neq 0\}$ to be the set of subsets $T \subseteq [n]$ for which $x_T$ appears in $q_i$ with a nonzero coefficient. And define $N_i = \{S \subseteq[n]| S \cup T \in [n]_{\le r}\forall T \in M_i\}$ to be the set of subsets $S \in [n]_{\le r}$ for which $|S \cup T| \le r$ for all $T$ such that $x_T$ appears in $q_i$ with a nonzero coefficient.
\begin{align*}
	\text{Maximize}\qquad&\displaystyle\sum_{T \in [n]_{\le r}}p_T\norm{\bm{V}_T}^2&&\\
	\text{subject to}\qquad&\displaystyle\sum_{T \in [n]_{\le r}}(q_i)_T\langle \bm{V}_T, \bm{V}_S\rangle \ge 0 & \forall& S \in [n]_{\le r}, i = 1, \ldots, m\\
	&\langle \bm{V}_{S_1}, \bm{V}_{S_2}\rangle = \langle \bm{V}_{S_3}, \bm{V}_{S_4}\rangle &\forall& S_1 \cup S_2 = S_3 \cup S_4\text{ and }S_i \in [n]_{\le r}\\
	&\langle \bm{V}_{S_1}, \bm{V}_{S_2}\rangle \ge 0 &\forall& S_1, S_2 \in [n]_{\le r}\\
	&\norm{\bm{V}_{\phi}}^2 = 1&&
\end{align*}
%\begin{eqnarray*}
%	&\text{Maximize }\displaystyle\sum_{T \in [n]_{\le r}}p_T\norm{\bm{V}_T}^2&\\
%	\text{subject to }&\displaystyle\sum_{T \in [n]_{\le r}}(q_i)_T\langle \bm{V}_T, \bm{V}_S\rangle \ge 0 & \forall S \in [n]_{\le r}, i = 1, \ldots, m\\
%	&\langle \bm{V}_{S_1}, \bm{V}_{S_2}\rangle = \langle \bm{V}_{S_3}, \bm{V}_{S_4}\rangle &\forall S_1 \cup S_2 = S_3 \cup S_4\text{ and }S_i \in [n]_{\le r}\\
%	&\langle \bm{V}_{S_1}, \bm{V}_{S_2}\rangle \ge 0 &\forall S_1, S_2 \in [n]_{\le r}\\
%	& \norm{\bm{V}_{\phi}}^2 = 1&\\
%\end{eqnarray*}
\end{definition}

First, note that this is indeed a relaxation because if the optimum solution to the original program was $x_i = b_i \in \{0, 1\}$, then the $1$-dimensional solution $\bm{V}_{T} = \displaystyle\prod_{i \in T} b_i$ satisfies the constraints and gives the same objective value.

Observe that we have $mn^{O(r)}$ constraints. This problem can be reformulated as a semidefinite program as follows: Introduce real variables $y_{S_1, S_2}$ to mean $\langle \bm{V}_{S_1}, \bm{V}_{S_2}\rangle$. So, we get a linear program in $y_{S_1, S_2}$ and moreover, the existence of vectors $\bm{V}_S$ for a given collection of $y_{S_1, S_2}$ is equivalent to saying that $Y = (y_{S_1, S_2})$ (which is an $n^{O(r)}\times n^{O(r)}$ matrix) is positive semidefinite. So, this program can be solved in time polynomial in the number of constraints. In most cases, we have $m$ to be constant. In that case, this program can be solved in $n^{O(r)}$ time.

Here, $r$ is called the number of levels of the program. It is known that if $r$ is as large as $n$, then we get actual probability distributions and hence, we would have solved the problem exactly. In general, we can study the tradeoff between the approximation guarantee and running time as $r$ grows.

\section{Examples}

We now give SoS relaxations for the natural integer program for some problems. We will describe the intended meaning of the basic linear program's variables $\{x_i\}_{i \in [n]}$ but the SoS relaxation will only contain the variables $\bm{V}_S$ for $|S| \le r$, where $r$ is the number of levels. $r$ can be arbitrary but in most cases, for notational simplicity, we just consider it to be at least the minimum size of a set $S$ that is present in the objective or one of the constraints. So, for example, in Densest $k$-subgraph, we assume $r \ge 2$ because the objective contains $\bm{V}_{\{u, v\}}$ for edges $(u, v)$. We can work with $r = 1$ but need to precisely explain what relaxation we are working with. We will show an example of this in the next section.

\subsection{Maximum Independent Set}

An instance of Maximum Independent Set is a graph $G = (V, E)$. The objective is to find a subset of vertices $S$ such that no edge $(u, v)$ has $u, v \in S$ and the subset $S$ is as large as possible. In the basic program, we have variables $x_u$ which indicate whether the vertex $u$ is in the final independent set. So, we need to maximize $\displaystyle\sum_{u \in V} x_u$ subject to $x_ux_v = 0$ for all edges $(u, v)$. Note that this condition ensures that the resulting set has no edges within. Assume $V = [n]$. The level $r$ SoS relaxation is as follows.
\begin{align*}
\text{Maximize}\qquad&\displaystyle\sum_{u\in V}\norm{\bm{V}_{\{u\}}}^2&&\\
\text{subject to}\qquad&\langle \bm{V}_{\{u, v\}}, \bm{V}_S \rangle = 0 &\forall& (u, v) \in E, S \in [n]_{\le r}\\
&\langle \bm{V}_{S_1}, \bm{V}_{S_2}\rangle = \langle \bm{V}_{S_3}, \bm{V}_{S_4}\rangle &\forall& S_1 \cup S_2 = S_3 \cup S_4\text{ and }S_i \in [n]_{\le r}\\
&\langle \bm{V}_{S_1}, \bm{V}_{S_2}\rangle \ge 0 &\forall& S_1, S_2 \in [n]_{\le r}\\
&\norm{\bm{V}_{\phi}}^2 = 1&&
\end{align*}
%\begin{eqnarray*}
%	&\text{Maximize }\displaystyle\sum_{u\in V}\norm{\bm{V}_{\{u\}}}^2&\\
%	\text{subject to }&\langle \bm{V}_{\{u, v\}}, \bm{V}_S \rangle = 0 &\forall (u, v) \in E, S \in [n]_{\le r}\\
%	&\langle \bm{V}_{S_1}, \bm{V}_{S_2}\rangle = \langle \bm{V}_{S_3}, \bm{V}_{S_4}\rangle &\forall S_1 \cup S_2 = S_3 \cup S_4\text{ and }S_i \in [n]_{\le r}\\
%	&\langle \bm{V}_{S_1}, \bm{V}_{S_2}\rangle \ge 0 &\forall S_1, S_2 \in [n]_{\le r}\\
%	& \norm{\bm{V}_{\phi}}^2 = 1&
%\end{eqnarray*}

\subsection{Max K-CSP}

An instance of Max $K$-CSP over alphabet $[q]$ contains $m$ constraints $C_1, \ldots, C_m$ over $n$ variables $x_1, \ldots, x_n$. Each constraint $C_i$ is a boolean predicate on a subset of $K$ distinct variables. That is, if $T_i$ is the subset of $K$ distinct variables for the $i$th constraint, then $C_i$ is a function from $[q]^{T_i}$ to $\{0, 1\}$. An assignment is a mapping of the variables to $[q]$. We say that an assignment satisfies $C_i$ if the evaluation of $C_i$ on the assignment restricted to $T_i$ is $1$. The objective is to assign values from $[q]$ to the variables $x_1, \ldots, x_n$ such that maximum number of constraints are satisfied.

For each $i \le m$ and $\alpha \in [q]^{T_i}$, let $C_i(\alpha)$ indicate whether the assignment $\alpha$ satisfies $C_i$. In the basic program, we have variables $y_{(j, \alpha_j)}$ which indicate whether the assignment to $x_j$ is $\alpha_j$. So, two immediate constraints are $\displaystyle\sum_{\alpha_j \in [q]} y_{(j, \alpha_j)} = 1$ and for $\alpha_j \neq \alpha_j'$ we have $y_{(j, \alpha_j)}y_{(j, \alpha_j')} = 0$. For $\alpha \in [q]^{T_i}$, denote by $y_{(T_i, \alpha)}$ the product $\displaystyle\prod_{j \in T_i} y_{(j, \alpha_j)}$ which indicates whether the final assignment to the variables restricted to $T_i$ is $\alpha$. So, we need to maximize the number indices $i$ with $\displaystyle\sum_{\alpha \in [q]^{T_i}} C_i(\alpha)y_{(T_i, \alpha)} = 1$ because this is equivalent to $C_i$ being satisfied. In the level $r$ SoS relaxation, we have variables $\bm{V}_{(S, \alpha)}$ for all subsets $S \in [n]_{\le r}$ for all assignments $\alpha \in [q]^{S}$. The level $r$ SoS relaxation is as follows:
\begin{align*}
\text{Maximize}~~&\displaystyle\sum_{i = 1}^m\displaystyle\sum_{\alpha \in [q]^{T_i}} C_i(\alpha)\norm{\bm{V}_{(T_i, \alpha)}}^2&&\\
\text{subject to}~~&\langle \bm{V}_{(S_1, \alpha_1)}, \bm{V}_{(S_2, \alpha_2)} \rangle = 0 &\forall& \alpha_1(S_1 \cap S_2) \neq \alpha_2(S_1 \cap S_2)\text, S_1, S_2 \in [n]_{\le r}\\
&\langle \bm{V}_{(S_1, \alpha_1)}, \bm{V}_{(S_2, \alpha_2)}\rangle = \langle \bm{V}_{(S_3, \alpha_3)}, \bm{V}_{(S_4, \alpha_4)}\rangle &\forall& S_1 \cup S_2 = S_3 \cup S_4, \alpha_1 \circ \alpha_2 = \alpha_3 \circ \alpha_4, S_i \in [n]_{\le r}\\
&\displaystyle\sum_{\alpha\in [q]} \langle\bm{V}_{\{j\}, [j \to \alpha]}, \bm{V}_S\rangle = \norm{\bm{V}_S}^2&\forall& S \in [n]_{\le r}, j \in [n]\\
&\langle \bm{V}_{S_1}, \bm{V}_{S_2}\rangle \ge 0 &\forall& S_1, S_2 \in [n]_{\le r}\\
& \norm{\bm{V}_{\phi}}^2 = 1&&
\end{align*}
%\begin{eqnarray*}
%	&\text{Maximize }\displaystyle\sum_{i = 1}^m\displaystyle\sum_{\alpha \in [q]^{T_i}} C_i(\alpha)\norm{\bm{V}_{\{T_i, \alpha\}}}^2&\\
%	\text{subject to }&\langle \bm{V}_{(S_1, \alpha_1)}, \bm{V}_{(S_2, \alpha_2)} \rangle = 0 &\forall \alpha_1(S_1 \cap S_2) \neq \alpha_2(S_1 \cap S_2)\text, S_1, S_2 \in [n]_{\le r}\\
%	&\langle \bm{V}_{(S_1, \alpha_1)}, \bm{V}_{(S_2, \alpha_2)}\rangle = \langle \bm{V}_{(S_3, \alpha_3)}, \bm{V}_{(S_4, \alpha_4)}\rangle &\forall S_1 \cup S_2 = S_3 \cup S_4, \alpha_1 \circ \alpha_2 = \alpha_3 \circ \alpha_4, S_i \in [n]_{\le r}\\
%	&\displaystyle\sum_{\alpha\in [q]} \norm{\bm{V}_{\{i\}, [i \to \alpha]}}^2 = 1& \forall i \in [n]\\
%	&\langle \bm{V}_{S_1}, \bm{V}_{S_2}\rangle \ge 0 &\forall S_1, S_2 \in [n]_{\le r}\\
%	& \norm{\bm{V}_{\phi}}^2 = 1&\\
%\end{eqnarray*}

Here, $\alpha(S_1 \cap S_2)$ is the assignment $\alpha$ restricted to $S_1 \cap S_2$, the first condition ensures that there are no contradictions in partial assignments for two sets. If $\alpha_1 \in [q]^{S_1}, \alpha_2 \in [q]^{S_2}$ which do not contradict each other, then $\alpha_1 \circ \alpha_2$ is the assignment on $S_1 \cup S_2$ that is the union of the two assignments. The second condition is a simple consistency constraint for the union of two partial assignments. The third constraint enforces that each variable is assigned some letter from $[q]$.

\subsection{Densest k-Subgraph}

An instance of Densest $k$-Subgraph is an undirected unweighted graph $G = (V, E)$ and a positive integer $k$. The objective is to find a subset $W$ of $V$ with exactly $k$ vertices such that the number of edges with both endpoints in $W$, is maximized.

The variable $x_u$ indicates whether the vertex $u$ is in the final solution. So, we need to have $\displaystyle\sum_{u \in V} x_u = k$ and the number of edges is $\displaystyle\sum_{(u, v) \in E} x_ux_v$, which we need to maximize. Assume $V = [n]$. The level $r$ SoS relaxation is as follows.
\begin{align*}
\text{Maximize}\qquad&\displaystyle\sum_{(u, v)\in E}\norm{\bm{V}_{\{u, v\}}}^2&&\\
\text{subject to}\qquad&\displaystyle\sum_{v \in V}\langle \bm{V}_{\{v\}}, \bm{V}_S \rangle = k\norm{\bm{V}_S}^2 &\forall& S \in [n]_{\le r}\\
&\langle \bm{V}_{S_1}, \bm{V}_{S_2}\rangle = \langle \bm{V}_{S_3}, \bm{V}_{S_4}\rangle &\forall& S_1 \cup S_2 = S_3 \cup S_4\text{ and }S_i \in [n]_{\le r}\\
&\langle \bm{V}_{S_1}, \bm{V}_{S_2}\rangle \ge 0 &\forall& S_1, S_2 \in [n]_{\le r}\\
&\norm{\bm{V}_{\phi}}^2 = 1&&
\end{align*}
%\begin{eqnarray*}
%	&\text{Maximize }\displaystyle\sum_{(u, v)\in E}\norm{\bm{V}_{\{u, v\}}}^2&\\
%	\text{subject to }&\displaystyle\sum_{v \in V}\langle \bm{V}_{\{v\}}, \bm{V}_S \rangle = k\norm{\bm{V}_S}^2 &\forall S \in [n]_{\le r}\\
%	&\langle \bm{V}_{S_1}, \bm{V}_{S_2}\rangle = \langle \bm{V}_{S_3}, \bm{V}_{S_4}\rangle &\forall S_1 \cup S_2 = S_3 \cup S_4\text{ and } S_i \in [n]_{\le r}\\
%	&\langle \bm{V}_{S_1}, \bm{V}_{S_2}\rangle \ge 0 &\forall S_1, S_2 \in [n]_{\le r}\\
%	& \norm{\bm{V}_{\phi}}^2 = 1&
%\end{eqnarray*}

\section{Maximum Clique on random graphs}\label{sosproof}

An instance of Maximum Clique is a graph $G = (V, E)$ and the objective is to find the size of the largest complete graph that is a subgraph, known as a clique, of $G$.

Through a series of work, in particular \cite{hastad} followed by \cite{ponnuswami}, it is known that maximum clique is hard to approximate within a factor of $n / 2^{(\log n)^{3/4 + \epsilon}}$ for any $\epsilon > 0$ where $n$ is the number of vertices, assuming $NP \not\subseteq BPTIME(2^{(\log n)^{O(1)}})$. But it is still interesting to understand how well we can do on average case instances, that is, when the graph is randomly picked from a predetermined distribution.

In particular, we consider Erd\"{o}s-R\'{e}nyi random graphs $G \sim G(n, 1/2)$ which is a graph $G = (V, E)$ on $n$ vertices where for each $u \neq v$, the edge $(u, v)$ is present in $E$ with probability $1/2$. By standard probabilistic arguments, it can be shown that $G \sim G(n, 1/2)$ has no cliques of size more than $2 \log n$ with high probability.

It is natural to consider the SoS relaxation of the standard integer program and study how it performs on a graph $G$ sampled from $G(n, 1/2)$. Feige and Krauthgamer\cite{probable} proved that a weaker hierarchy known as the {L}ov\'asz-{S}chrijver hierarchy for $r$ levels returns an optimum value of $\Theta(\sqrt{n / 2^r})$, with high probability. We will prove the upper bound for the SoS hierarchy as is studied here.

The basic program has boolean variables $x_u$ for $u \in V$ where $x_u$ indicates whether $u$ is in the largest clique. So, we need to maximize $\displaystyle\sum_{u \in V} x_u$ subject to $x_ux_v = 0$ for all pairs $(u, v)$, with $u \neq v$, that are not edges. The constraint ensures that two chosen vertices are always connected by an edge. The level $r$ SoS relaxation $\mathcal{P}_r$ for maximum clique is as follows.
\begin{align*}
\text{Maximize}\;&\displaystyle\sum_{u\in V}\norm{\bm{V}_{\{u\}}}^2&&\\
\text{subject to}\;&\langle \bm{V}_{S_1}, \bm{V}_{S_2} \rangle = 0 &\forall& S_1, S_2 \in [n]_{\le r}\text{ such that }\exists u\neq v \in S_1 \cup S_2, (u, v) \not\in E\\%\wedge u \neq v\wedge \{u, v\} \subseteq S_1 \cup S_2\wedge S_1, S_2 \in [n]_{\le r}\\
&\langle \bm{V}_{S_1}, \bm{V}_{S_2}\rangle = \langle \bm{V}_{S_3}, \bm{V}_{S_4}\rangle &\forall& S_1 \cup S_2 = S_3 \cup S_4\text{ and }S_i \in [n]_{\le r}\\
&\langle \bm{V}_{S_1}, \bm{V}_{S_2}\rangle \ge 0 &\forall& S_1, S_2 \in [n]_{\le r}\\
&\norm{\bm{V}_{\phi}}^2 = 1&&
\end{align*}
%\begin{eqnarray*}
%	&\text{Maximize }\displaystyle\sum_{u\in V}\norm{\bm{V}_{\{u\}}}^2&\\
%	\text{subject to }&\langle \bm{V}_{\{u, v\}}, \bm{V}_S \rangle = 0 &\forall (u, v) \not\in E, S \in [n]_{\le r}\\
%	&\langle \bm{V}_{S_1}, \bm{V}_{S_2}\rangle = \langle \bm{V}_{S_3}, \bm{V}_{S_4}\rangle &\forall S_1 \cup S_2 = S_3 \cup S_4\text{ and }S_i \in [n]_{\le r}\\
%	&\langle \bm{V}_{S_1}, \bm{V}_{S_2}\rangle \ge 0 &\forall S_1, S_2 \in [n]_{\le r}\\
%	& \norm{\bm{V}_{\phi}}^2 = 1&
%\end{eqnarray*}

Here, for $r \ge 2$, the first constraint is equivalent to $\langle \bm{V}_{\{u, v\}}, \bm{V}_S \rangle = 0$ for all $(u, v) \not\in E, u \neq v, S \in [n]_{\le r}$. The reason we write it in a different manner above is to incorporate the case $r = 1$. When $r = 1$, the constraint is precisely $\langle \bm{V}_{\{u\}}, \bm{V}_{\{v\}} \rangle = 0$ for all $(u, v) \not\in E, u \neq v$.
%In particular, the level $1$ relaxation $\mathcal{P}_1$ is as follows:
%\begin{align*}
%\text{Maximize}\qquad&\displaystyle\sum_{u\in V}\norm{\bm{V}_{\{u\}}}^2&&\qquad\qquad\qquad\qquad\qquad\qquad\qquad\\
%\text{subject to}\qquad&\langle \bm{V}_{\{u\}}, \bm{V}_{\{v\}} \rangle = 0 &\forall& (u, v) \not\in E, u \neq v\\
%&\langle \bm{V}_{\{u\}}, \bm{V}_{\phi}\rangle = \norm{\bm{V}_{\{u\}}}^2 &\forall& u\in V\\
%&\langle \bm{V}_{\{u\}}, \bm{V}_{\{v\}}\rangle \ge 0 &\forall& u, v \in V\\
%& \norm{\bm{V}_{\phi}}^2 = 1&&
%\end{align*}
%\begin{eqnarray*}
%	&\text{Maximize }\displaystyle\sum_{u\in V}\norm{\bm{V}_{\{u\}}}^2&\\
%	\text{subject to }&\langle \bm{V}_{\{u\}}, \bm{V}_{\{v\}} \rangle = 0 &\forall (u, v) \not\in E, u \neq v\\
%	&\langle \bm{V}_{\{u\}}, \bm{V}_{\phi}\rangle = \norm{\bm{V}_{\{u\}}}^2 &\forall u\in V\\
%	&\langle \bm{V}_{\{u\}}, \bm{V}_{\{v\}}\rangle \ge 0 &\forall u, v \in V\\
%	& \norm{\bm{V}_{\phi}}^2 = 1&
%\end{eqnarray*}

We will analyze this SDP by relating it to a function on graphs known as the Lov\'{a}sz $\vartheta$ function. Lov\'{a}sz\cite{lov} introduced a function $\vartheta(G)$ that can be computed efficiently which gives an upper bound on $\alpha(G)$, the size of the maximum independent set in $G$. The function is usually defined using orthonormal representations of graphs but it can be shown to be equivalent (see for instance, \cite[Section 9.2]{geomrep}) to the following definition.

\begin{definition}[Lov\'{a}sz $\vartheta$ function]
$\vartheta(G)$ is the optimum value of the following SDP on variables $\bm{W}_u$ for $u \in V$:
\begin{align*}
\text{Maximize}\qquad&\displaystyle\sum_{u, v\in V} \langle \bm{W}_u, \bm{W}_v\rangle&\\
\text{subject to}\qquad\quad&\quad\langle \bm{W}_u, \bm{W}_v \rangle = 0 &\forall (u, v) \in E\\
&\displaystyle\sum_{u\in V} \norm{\bm{W}_u}^2 = 1&
\end{align*}
%\begin{eqnarray*}
%	&\text{Maximize }\displaystyle\sum_{u, v\in V} \langle \bm{y}_u, \bm{y}_v\rangle&\\
%	\text{subject to }&\langle \bm{y}_u, \bm{y}_v \rangle = 0 &\forall (u, v) \in E\\
%	&\displaystyle\sum_{u\in V} \norm{\bm{y}_u}^2 = 1
%	\end{eqnarray*}
\end{definition}

Let $\overline{G}$ be the complement graph of $G$, that is, $\overline{G} = (V, E')$ where $(u, v) \in E'$ if and only if $u \neq v$ and $(u, v) \not\in E$. The following lemma relates the optimum value of $P_1$, the level $1$ SoS relaxation, to the value of $\vartheta$ of the complement graph.

\begin{lemma}\label{base}
The optimum value of $\mathcal{P}_1$ for $G$ is at most $\vartheta(\overline{G})$
\end{lemma}

\begin{proof}
Consider the optimal solution $\{\bm{V}_S\}_{S \in [n]_{\le 1}}$ for $\mathcal{P}_1$ with $\displaystyle\sum_{u \in V}\norm{\bm{V}_{\{u\}}}^2 = FRAC$, the optimum value of $\mathcal{P}_1$. Consider the SDP formulation for $\vartheta(\overline{G})$ with variables $\bm{W}_u$ and set $\bm{W}_u = \bm{V}_{\{u\}} / \sqrt{FRAC}$. We have $\displaystyle\sum_{u \in V}\norm{\bm{W}_{u}}^2 = 1$. For each edge $(u, v) \in E'$, we have $\langle \bm{W}_u, \bm{W}_v\rangle = \langle \bm{V}_{\{u\}}, \bm{V}_{\{v\}} \rangle / FRAC = 0$ since $(u, v) \not\in E$. Finally,
\begin{align*}
FRAC\times\vartheta(\overline{G}) &\ge 
FRAC \times \displaystyle\sum_{u, v\in V} \langle \bm{W}_u, \bm{W}_v\rangle\\
&= \displaystyle\sum_{u, v\in V} \langle \bm{V}_{\{u\}}, \bm{V}_{\{v\}}\rangle\\
&= \langle \displaystyle\sum_{u\in V} \bm{V}_{\{u\}}, \displaystyle\sum_{u\in V} \bm{V}_{\{u\}}\rangle\\
&= \norm{\displaystyle\sum_{u\in V} \bm{V}_{\{u\}}}^2\norm{\bm{V}_{\phi}}^2\\
&\ge \left(\langle \displaystyle\sum_{u\in V} \bm{V}_{\{u\}}, \bm{V}_{\phi}\rangle\right)^2\\
&= \left(\displaystyle\sum_{u\in V} \langle \bm{V}_{\{u\}}, \bm{V}_{\phi}\rangle\right)^2 = \left(\displaystyle\sum_{u\in V} \norm{\bm{V}_{\{u\}}}^2\right)^2 = FRAC^2
\end{align*}
where the second inequality follows by Cauchy-Schwarz inequality. This proves that $FRAC \le \vartheta(\overline{G})$ as required.
\end{proof}

The following theorem was shown by \cite{probable} for the Lov\'{a}sz-Schrijver hierarchy for the Maximum Independent Set problem. We modify it slightly by showing it for the SoS hierarchy for the Maximum Clique problem, which is equivalent to Maximum Independent Set problem on the complement graph. Using a stronger hierarchy makes their proof much simpler and this simpler version is presented below.

For a subset $S \subseteq V$ of a graph $G = (V, E)$, define $\Gamma_G(S) = \{u \in V \;|\; \exists v \in S, (u, v) \in E\}$ to be the set of neighbors of $S$ in $G$ and define $G - S$ to be the graph obtained from $G$ by deleting the vertices in $S$ along with their edges.

\begin{theorem}[\cite{probable}]\label{fk03}
Fix $0 < \epsilon < 1$. Let $G = (V, E)$ be a graph on $n$ vertices and let $r \ge 1$ be an integer such that for all subsets $S \subseteq V$ of size at most $r$, the graph $G' = G - S - \Gamma_{\overline{G}}(S)$ on $n'$ vertices satisfies the following assumptions:
\begin{itemize}
	\item $\vartheta(\overline{G'}) \le 2(1 + \epsilon)\sqrt{n'}$
	\item Each vertex in $G'$ has degree between $\frac{n'}{2(1 + \epsilon)}$ and $\frac{(1 + \epsilon)n'}{2}$.
\end{itemize}
Let $d = (1 - \epsilon)\sqrt{2}$. If $d^{r + 1} \le \epsilon^2\sqrt{n}$, then $\mathcal{P}_r$ has an optimum value of at most $4\sqrt{n}/d^{r + 1}$.
\end{theorem}

\begin{proof}
Let us denote the optimum value of $\mathcal{P}_r$ by $FRAC$. We induct on $r$. When $r = 1$, using lemma \ref{base} and the first assumption for $S = \phi$, we get $FRAC \le \vartheta(\overline{G}) \le 2(1 + \epsilon)\sqrt{n} < 4\sqrt{n}/d^2$. Now assume that the result holds for $r$ levels and consider $r + 1$ levels for a graph $G$ satisfying the given conditions for all subsets $S$ of size at most $r + 1$. Let the optimal SoS vectors for $\mathcal{P}_{r + 1}$ be $\{\bm{V}_S\}_{S \in [n]_{\le r + 1}}$. We wish to prove that $FRAC = \displaystyle\sum_{u \in V} \norm{\bm{V}_{\{u\}}}^2 \le 4\sqrt{n}/d^{r + 2}$.

For each $u \in V$, define the graph $G_u = G - \{u\} - \Gamma_{\overline{G}}(\{u\})$ and let it have vertex set $V_u$ with $n_u$ vertices. Observe that $G_u$ satisfies the conditions given in the theorem for all subsets $S$ of size at most $r$. Indeed, if we consider any subset $S$ of $G_u$ of size at most $r$, then $G_u - S - \Gamma_{\overline{G_u}}(S) = G - S' - \Gamma_{\overline{G}}(S')$ where $S' = S \cup \{u\}$ is of size at most $r + 1$, which proves that the two assumptions hold. So, by the induction hypothesis, since $G_u$ satisfies the assumptions for sets of size at most $r$, the relaxation $\mathcal{P}_r$ for $G_u$ has an optimum value of at most $4\sqrt{n_u} / {d^{r + 1}}$.

Let $R = \{u \in V \;| \; \norm{\bm{V}_{\{u\}}} > 0\}$ be the set of vertices with nonzero SoS vectors. Fix $w \in R$. Define the vectors $\bm{U}_{S} = \bm{V}_{\{w\} \cup S} / \norm{\bm{V}_{\{w\}}}$. Informally, this can be thought of to capture the event that $S$ is a subset of the maximum clique conditioned on the event that $w$ is already chosen in the clique. We claim that $\bm{U}_{S}$ is a feasible solution for $\mathcal{P}_r$ for $G_w$. Note that $\bm{U}_S$ for $|S| \le r$ is well defined since $|\{w\} \cup S| \le r + 1$. For any $(u, v) \not\in E, u \neq v$ and $S_1, S_2$ of size at most $r$ such that $u, v \in S_1 \cup S_2$, we have $\langle \bm{U}_{S_1}, \bm{U}_{S_2} \rangle = \langle \bm{V}_{\{w\}\cup S_1}, \bm{V}_{\{w\} \cup S_2} \rangle / \norm{\bm{V}_{\{w\}}}^2 = 0$ since $u, v \in (\{w\}\cup S_1)\cup (\{w\}\cup S_2)$. For $S_1, S_2, S_3, S_4$ of size at most $r$ such that $S_1 \cup S_2 = S_3 \cup S_4$, we have that $(\{w\} \cup S_1) \cup (\{w\} \cup S_2) = (\{w\} \cup S_3) \cup (\{w\} \cup S_4)$ and hence, $\langle \bm{U}_{S_1}, \bm{U}_{S_2}\rangle = \langle \bm{V}_{\{w\} \cup S_1}, \bm{V}_{\{w\} \cup S_2}\rangle / \norm{\bm{V}_{\{w\}}}^2 = \langle \bm{V}_{\{w\} \cup S_3}, \bm{V}_{\{w\} \cup S_4}\rangle / \norm{\bm{V}_{\{w\}}}^2 = \langle \bm{U}_{S_3}, \bm{U}_{S_3}\rangle$ and $\langle \bm{U}_{S_1}, \bm{U}_{S_2}\rangle = \langle \bm{V}_{\{w\} \cup S_1}, \bm{V}_{\{w\} \cup S_2}\rangle / \norm{\bm{V}_{\{w\}}}^2 \ge 0$. Finally, $\norm{\bm{U}_{\phi}}^2 = \norm{\bm{V}_{\{w\}}}^2 / \norm{\bm{V}_{\{w\}}}^2 = 1$.

By the induction hypothesis, we get that $\displaystyle\sum_{u \in V_w} \norm{\bm{U}_{\{u\}}}^2 \le 4\sqrt{n_u}/{d^{r + 1}}$ which implies $\displaystyle\sum_{u \in V_w} \langle\bm{V}_{\{u\}}, \bm{V}_{\{w\}}\rangle \le (4\sqrt{n_u}/{d^{r + 1}})\norm{\bm{V}_{\{w\}}}^2$. By taking $S = \phi$ in the assumptions, we get that $w$ has degree at most $\frac{(1 + \epsilon)n}{2}$ and so, $n_u \le \frac{(1 + \epsilon)n}{2}$. Using this and the assumption that $d^{r + 1} \le \epsilon^2\sqrt{n}$, we get $4\sqrt{n_u}/{d^{r + 1}} \le 4(1 - \epsilon)\sqrt{1 + \epsilon}\sqrt{n}/{d^{r + 2}} < 4\sqrt{n}/d^{r + 2} - 1$ and therefore, $\displaystyle\sum_{u \in V_w} \langle\bm{V}_{\{u\}}, \bm{V}_{\{w\}}\rangle \le 4\sqrt{n}/d^{r + 2} - 1$.

We have $FRAC = \displaystyle\sum_{u \in V} \norm{\bm{V}_{\{u\}}}^2 = \displaystyle\sum_{u \in V} \langle \bm{V}_{\{u\}}, \bm{V}_{\phi}\rangle = \langle\displaystyle\sum_{u \in V} \bm{V}_{\{u\}}, \bm{V}_{\phi}\rangle$. By Cauchy-Schwarz, this is at most $\norm{\displaystyle\sum_{u \in V} \bm{V}_{\{u\}}}\cdot\norm{\bm{V}_{\phi}} = \norm{\displaystyle\sum_{u \in V} \bm{V}_{\{u\}}}$. When $(u, w) \not\in E$, we have $\langle \bm{V}_{\{u\}}, \bm{V}_{\{w\}}\rangle = 0$. And when $w \not\in R$, we have $\bm{V}_{\{w\}} = 0$. Using these, we get
\begin{align*}
FRAC^2 &\le \langle\displaystyle\sum_{u \in V} \bm{V}_{\{u\}},\displaystyle\sum_{u \in V} \bm{V}_{\{u\}}\rangle\\
&= \displaystyle\sum_{u \in V, w \in V} \langle\bm{V}_{\{u\}}, \bm{V}_{\{w\}}\rangle\\
&= \displaystyle\sum_{u \in V, w \in R} \langle\bm{V}_{\{u\}}, \bm{V}_{\{w\}}\rangle \\
&= \displaystyle\sum_{w \in R}\left(\norm{\bm{V}_{\{w\}}}^2 + \displaystyle\sum_{u \in V_w} \langle\bm{V}_{\{u\}}, \bm{V}_{\{w\}}\rangle\right)\\
&\le \displaystyle\sum_{w \in R}\left(\norm{\bm{V}_{\{w\}}}^2 + (4\sqrt{n}/d^{r + 2} - 1)\norm{\bm{V}_{\{w\}}}^2\right)\\
&= (4\sqrt{n}/d^{r + 2})\displaystyle\sum_{w \in R}\norm{\bm{V}_{\{w\}}}^2\\
&\le (4\sqrt{n}/d^{r + 2})\displaystyle\sum_{w \in V}\norm{\bm{V}_{\{w\}}}^2\\
&= (4\sqrt{n}/d^{r + 2})FRAC
\end{align*}
This completes the induction.
\end{proof}

We finally argue that that $G \sim G(n, 1/2)$ satisfies the assumptions in Theorem \ref{fk03} with high probability when $r = O(\log n)$. Juh{\'a}sz\cite{juhasz} showed a concentration result on the value of $\vartheta(G)$ for $G \sim G(n, 1/2)$ using eigenvalue concentration bounds of random matrices\cite{furedi} but by using stronger concentration bounds\cite{kriv}, Feige and Krauthgamer\cite{probable} were able to show the following result.

\begin{theorem}[\cite{probable}]\label{conc}
For any $\epsilon > 0$, there exists $\epsilon' > 0$ such that for any $r \le \epsilon' \log n$, $G = (V, E) \sim G_{n, 1/2}$ satisfies the following condition with high probability: for all subsets $S \subseteq V$ of size at most $r$, the graph $G' = G - S - \Gamma_{G}(S)$ on $n'$ vertices satisfies the following assumptions:
\begin{itemize}
	\item $\vartheta(G') \le 2(1 + \epsilon)\sqrt{n'}$
	\item Each vertex in $\overline{G'}$ has degree between $\frac{n'}{2(1 + \epsilon)}$ and $\frac{(1 + \epsilon)n'}{2}$.
\end{itemize}
\end{theorem}

Observe that when $G \sim G(n, 1/2)$, the graph $\overline{G}$ is also distributed as $G(n, 1/2)$. So, for any $\epsilon > 0$, there exists $\epsilon' > 0$ such that for any $r \le \epsilon' \log n$, with high probability, for all subsets $S \subseteq V$ of size at most $r$, the graph $G' = \overline{G} - S - \Gamma_{\overline{G}}(S)$ on $n'$ vertices satisfies the two assumptions in Theorem \ref{conc}. But note that $\overline{G'} = G - S - \Gamma_{\overline{G}}(S)$ which proves that $G$ satisfies the conditions of Theorem \ref{fk03} with high probability for $r = O(\log n)$.

We get that $\mathcal{P}_r$ for $G \sim G(n, 1/2)$ has an optimum value of at most $4\sqrt{n} / ((1 - \epsilon)\sqrt{2})^{r  + 1}$ with high probability. This in particular gives an algorithm for the the Planted Clique problem. An instance of Planted Clique is a graph $G = (V, E)$ drawn from one of the following distributions equally likely:
\begin{itemize}
	\item $G(n, 1/2)$ - The Erdos-Renyi graph on $n$ vertices where each edge $(u, v)$ is present with probability $1/2$ for all $u \neq v$.
	\item $G(n, 1/2, k)$ - The graph is first sampled from $G(n, 1/2)$ and then $k$ vertices are chosen uniformly at random and clique is planted on these $k$ vertices. That is, if $W$ is the chosen $k$ vertices, then for all $u, v \in W$ with $u \neq v$, the edge $(u, v)$ is added if not already present. The resulting graph is returned.
\end{itemize}
The objective is to determine which distribution the graph $G$ is drawn from, with probability of being correct at least some constant $p > 1/2$.

If $k \gg 4\sqrt{n}/((1 - \epsilon)\sqrt{2})^{r  + 1}$, then we get that SoS for $r$ levels distinguishes the two distributions with high probability because the optimum value of the relaxation is at most $4\sqrt{n}/((1 - \epsilon)\sqrt{2})^{r  + 1}$ for $G(n, 1/2)$ and is at least $k$ for $G(n, 1/2, k)$. So, we can solve the Planted Clique problem for $k \gg \sqrt{n / 2^r}$ in $n^{O(r)}$ time.

We will later study SoS lower bounds for Maximum Clique on random graphs, where we show that this upper bound is almost tight, and this Planted Clique view will be very useful for constructing integrality gaps.

\section{Approximation algorithms for low threshold-rank graphs}

%\begin{definition}
%The $\epsilon$-threshold rank of a real symmetric matrix $X$ is the number of eigenvalues that are more than $\epsilon$.
%\end{definition}

For a graph $G = (V, E)$ on $n$ vertices, consider the normalized adjacency matrix $A$. The graph $G$ is informally called a low threshold-rank graph if $A$ has very few eigenvalues more than a positive constant $\epsilon$. These kind of graphs satisfy many interesting properties. For instance, if there is only one eigenvalue more than $0.5$, then that means that the second eigenvalue is at most $0.5$ and by Cheeger's inequality, this graph is an expander. More generally, Gharan and Trevisan\cite{gt} proved that low threshold rank graphs roughly look like a union of expanders, in the sense that few edges of the graph can be deleted so that each remaining component is an expander.

Guruswami and Sinop\cite{gs} obtained approximation algorithms to many standard graph problems including Unique Games, by rounding the solutions to the SoS hierarchy via an idea known as propagation. For a positive integer $r$ and constant $\epsilon > 0$, by using $O(r/\epsilon^2)$ levels of the SoS hierarchy, they were able to obtain approximation algorithms with approximation guarantees depending inversely on $\lambda_r(L)$, the $r$th smallest eigenvalue of the normalized Laplacian $L$ of the graph. In particular, for low threshold-rank graphs (where $\lambda_r(L)$ is large for small $r$), we get good approximation algorithms which are efficient.

Similar results were obtained by Barak, Raghavendra and Steurer\cite{brs} by rounding SoS solutions via an idea known as local to global correlation.

For the sake of exposition, we will describe the result and rounding algorithm of $\cite{gs}$ for Minimum bisection. An instance of Minimum Bisection is a graph $G = (V, E)$ and an integer $k$. The objective is to find a subset $S \subseteq V$ with exactly $k$ vertices such that the number of edges with exactly one endpoint in $S$ is minimized.

\begin{theorem}[\cite{gs}]\label{guruswamisinop}
Consider any instance of Minimum Bisection $(G, k)$ and for any subset $S$ of the vertices $V$, let $\Gamma(S)$ denote the number of edges with exactly one endpoint in $S$. For a positive integer $r$ and constant $\epsilon > 0$, in time $n^{O(r/\epsilon^2)}$, we can find a set $R' \subseteq V$ such that $k(1 - o(1)) \le |R'| \le (1 + o(1))$ and $\Gamma(R') \le \frac{1 + \epsilon}{\min(1, \lambda_{r + 1}(L))}\Gamma(R)$, where $R$ is the optimal solution, namely $R = \text{argmin}_{|S| = k}\Gamma(S)$.
\end{theorem}

We will first describe two ingredients that we will need for our proof. For the rest of this section, for any square matrix $A$, let $\lambda_t(A)$ denotes the $t$-th smallest eigenvalue of $A$ and let $\norm{A}_F$ denote the Frobenius norm of $A$.

Consider any matrix $X \in \mathbb{R}^{n' \times m'}$. Let the singular values be $\sigma_1 \ge \sigma_2 \ge \ldots \ge \sigma_{m'} \ge 0$ and let the columns of the matrix be $\bm{v}_1, \ldots, \bm{v}_{m'}$. For any $r \le m'$, we know that among all projection maps $\Pi^{\perp}$ on $\mathbb{R}^{n'}$ into the orthogonal complement of subspaces of dimension $r$, the minimum value of $\displaystyle\sum_{i = 1}^{m'} \norm{\Pi^{\perp}\bm{v}_i}^2$ is $\displaystyle\sum_{i = r + 1}^{m'} \sigma_i^2$. We would like to analyze what happens if we restrict our projection to be in a subspace spanned by a subset of the $\bm{v}_i$s. The following lemma claims that we can still achieve a good guarantee.

\begin{lemma}[\cite{gspre}]\label{gscol}
	For all positive integers $r' \ge r$, there exist $r'$ columns of $X$ such that if $\Pi^{\perp}$ is the projection map on $\mathbb{R}^{n}$ into the orthogonal complement of the subspace spanned by these columns, then $\displaystyle\sum_{i = 1}^{m'} \norm{\Pi^{\perp} \bm{v}_i}^2 \le \frac{r' + 1}{r' - r + 1}\left(\displaystyle\sum_{i = r + 1}^{m'} \sigma_i^2\right)$. In particular, for any $\epsilon > 0$, if $r' \ge r / \epsilon$, then the right hand side is at most $\frac{1}{1 - \epsilon} \left(\displaystyle\sum_{i = r + 1}^{m'} \sigma_i^2\right)$.
\end{lemma}

We will also need an inequality on the Frobenius norm of the difference of two real symmetric matrices.

\begin{lemma}[Hoffman-Wielandt inequality]\label{hoff}
	Let $A, B$ be $n \times n$ normal matrices with eigenvalues $\lambda_1(A), \ldots, \lambda_n(A)$ and $\lambda_1(B), \ldots, \lambda_n(B)$ respectively, then
	\[\norm{A - B}_F^2 \ge \min_{\sigma\in S_n} \sum_{i = 1}^n |\lambda_i(A) - \lambda_{\sigma(i)}(B)|^2\]
\end{lemma}

For a proof, see for example \cite[Theorem VI.4.1, page 165]{bhatia}. In particular, this inequality holds if $A, B$ are symmetric real matrices.

In the following proof, assume $G$ is $d$-regular, but Guruswami and Sinop's results work for general graphs.

\begin{proof}[Proof outline of Theorem \ref{guruswamisinop}.] We will show a slightly weaker approximation guarantee of $(1 + \frac{1}{(1 - \epsilon)\lambda_{r + 1}(L)})\Gamma(R)$. This will illustrate the main idea behind the rounding algorithm, and getting the improved guarantee requires only a bit more work. Let the vertex set be $[n]$. In the basic program, we have variables $x_u$ which indicate whether $u \in R$ and so, we have the constraint $\displaystyle\sum_{u \in V} x_u = k$. Note that the expression $(x_u - x_v)^2$ indicates whether the edge $(u, v)$ is cut. So, the objective is $\displaystyle\sum_{(u, v) \in E}(x_u - x_v)^2$. For $r' = \Omega(r / \epsilon^2)$, we consider the SoS relaxation for $r' + 1$ levels:
\begin{align*}
	\text{Minimize}\qquad&\displaystyle\sum_{(u, v)\in E}\norm{\bm{V}_{\{u\}} - \bm{V}_{\{v\}}}^2&&\\
	\text{subject to}\qquad&\displaystyle\sum_{v \in V}\langle \bm{V}_{\{v\}}, \bm{V}_S \rangle = k\norm{\bm{V}_S}^2 &\forall& S \in [n]_{\le r'}\\
	&\langle \bm{V}_{S_1}, \bm{V}_{S_2}\rangle = \langle \bm{V}_{S_3}, \bm{V}_{S_4}\rangle &\forall& S_1 \cup S_2 = S_3 \cup S_4 \in [n]_{\le r'}\\
	&\langle \bm{V}_{S_1}, \bm{V}_{S_2}\rangle \ge 0 &\forall& S_1, S_2 \in [n]_{\le r'}\\
	& \norm{\bm{V}_{\phi}}^2 = 1&&\\
\end{align*}
%\begin{eqnarray*}
%	&\text{Maximize }\displaystyle\sum_{(u, v)\in E}\norm{\bm{V}_{\{u\}} - \bm{V}_{\{v\}}}^2&\\
%	\text{subject to }&\displaystyle\sum_{v \in V}\langle \bm{V}_{\{v\}}, \bm{V}_S \rangle = k\norm{\bm{V}_S}^2 &\forall S \in [n]_{\le r'}\\
%	&\langle \bm{V}_{S_1}, \bm{V}_{S_2}\rangle = \langle \bm{V}_{S_3}, \bm{V}_{S_4}\rangle &\forall S_1 \cup S_2 = S_3 \cup S_4 \in [n]_{\le r'}\\
%	&\langle \bm{V}_{S_1}, \bm{V}_{S_2}\rangle \ge 0 &\forall S_1, S_2 \in [n]_{\le r'}\\
%	& \norm{\bm{V}_{\phi}}^2 = 1&\\
%\end{eqnarray*}

Suppose $\bm{V}_S \in \mathbb{R}^{\gamma}$ is our optimal SoS solution. For all nonempty $S \subseteq [n]_{\le r'}, \alpha \in \{0, 1\}^S$, suppose $\alpha$ maps all of $S' \subseteq S$ to $1$ and all of $S - S'$ to $0$ for some $S' \subseteq S$, define $\bm{U}_{S, \alpha} = \displaystyle\sum_{S' \subseteq T \subseteq S} (-1)^{|T - S'|}\bm{V}_{T'}$, a vector intended to capture the event that $\alpha$ correctly indicates whether $u \in S$ is in $R$. This definition can be thought of to an application of the inclusion-exclusion principle. We also define $\bm{U}_{\phi, \phi} = \bm{V}_{\phi}$. In the rest of the proof, when $S = \phi$, there is no $\alpha \in \{0, 1\}^S$, but we instead assume by convention that there is a unique element $\phi \in \{0, 1\}^S$ with $\bm{U}_{S, \alpha} =\bm{U}_{\phi, \phi}$. Observe the following facts about $\bm{U}_{S, \alpha}$, which are verified by straightforward computations:
\begin{itemize}
	\item $\bm{U}_{S, \mathbbm{1}_S} = \bm{V}_S$ for all $S\in [n]_{\le r'}$ where $\mathbbm{1}_S$ maps all of $S$ to $1$ and by convention, $\mathbbm{1}_{\phi} = \phi$.
	\item $\displaystyle\sum_{\beta\in \{0, 1\}^{\{u\}}} \bm{U}_{S, \alpha \circ \beta} = \bm{U}_{S - \{u\}, \alpha}$ for all $u \in S \in [n]_{\le r'}, \alpha \in \{0, 1\}^{S - \{u\}}$. Here, $\alpha \circ \beta: \{0, 1\}^S$ sends $v \in S$ to $\alpha(v)$ if $v\neq u$ and $\beta(v)$ otherwise.
	\item For all $S, T \in [n]_{\le r'}$, if $\alpha \in \{0, 1\}^S, \beta \in \{0, 1\}^T$ are such that there exists $u \in S \cap T$ with $\alpha(u) \neq \beta(u)$, then $\langle \bm{U}_{S, \alpha}, \bm{U}_{T, \beta}\rangle = 0$.
	\item For all $S \in [n]_{\le r'}$, we have $\displaystyle\sum_{\alpha \in \{0, 1\}^S} \bm{U}_{S, \alpha} = V_{\phi}$ and $\displaystyle\sum_{\alpha \in \{0, 1\}^S} \norm{\bm{U}_{S, \alpha}}^2 = 1$. In particular, $\norm{\bm{U}_{\phi, \phi}}^2 = \norm{\bm{V}_{\phi}}^2 = 1$.
	\item For all $S, T, S', T' \in [n]_{\le r'}$ such that $S \cup T = S' \cup T'$ and all $\alpha \in \{0, 1\}^S, \beta \in \{0, 1\}^T, \alpha' \in \{0, 1\}^{S'}, \beta' \in \{0, 1\}^{T'}$ such that $\alpha(u) = \beta(u)$ for all $u \in S \cap T, \alpha'(u) = \beta'(u)$ for all $u \in S' \cap T'$ and $\alpha \circ \beta = \alpha' \circ \beta'$, we have $\langle \bm{U}_{S, \alpha}, \bm{U}_{T, \beta}\rangle = \langle \bm{U}_{S', \alpha'}, \bm{U}_{T', \beta'}\rangle$. Here, $\alpha\circ \beta:\{0, 1\}^{S \cup T}$ maps $u \in S$ to $\alpha(u)$ and $u \in T$ to $\beta(u)$ (note that this is well defined since the values match on the intersection) and $\alpha'\circ \beta'$ is similarly defined.
\end{itemize}
From the above consistency properties, we can think of $\norm{U_{S, \alpha}}^2$ as the probability that $R \cap S = \{u \in S \;|\; \alpha(u) = 1\}$. The rounding algorithm proceeds by guessing a subset $S \in [n]_{\le r'}$ (indeed, all guesses can be tried in $n^{O(r')}$ time) and choosing an assignment $\alpha \in \{0, 1\}^S$ with probability $\norm{\bm{U}_{S, \alpha}}^2$. Once $\alpha$ is chosen, the rounding algorithm returns the set $R'$ where, for all $u \in V$, $u$ is included in $R'$ with probability $\frac{\langle \bm{U}_{S, \alpha}, \bm{U}_{\{u\}, \mathbbm{1}_{\{u\}}}\rangle}{\langle \bm{U}_{S, \alpha}, \bm{U}_{S, \alpha}\rangle} = \frac{\langle \bm{U}_{S, \alpha}, \bm{V}_{\{u\}}\rangle}{\langle \bm{U}_{S, \alpha}, \bm{U}_{S, \alpha}\rangle}$. We remark that for all $u \in S$, $u$ is included in $R'$ if and only if $\alpha(u) = 1$. By Chernoff bounds, it can be shown that $k(1 - o(1)) \le |R'| \le k(1 + o(1))$ with high probability.

It remains to analyze $\Gamma(R')$ and compare it to $\Gamma(R)$. We will argue that there exists a subset $S$ such that the expectation $\mathbb{E}_{\alpha}[\Gamma(R')]$ over the choice of $\alpha$ satisfies our approximation guarantees. For ease of notation, let $E'$ be the set of directed edges of $G$, where each edge in $E$ occurs twice as two directed edges $(u, v)$ and $(v, u)$. We have $\Gamma(R) \ge \displaystyle\sum_{(u, v)\in E}\norm{\bm{V}_{\{u\}} - \bm{V}_{\{v\}}}^2 = \displaystyle\sum_{(u, v)\in E'}\norm{\bm{V}_{\{u\}}}^2 - \displaystyle\sum_{(u, v)\in E'}\langle\bm{V}_{\{u\}}, \bm{V}_{\{v\}}\rangle$.

Fix an $S \in [n]_{\le r'}$. Let $\Pi_1$ be the projection map on $\mathbb{R}^{\gamma}$ into the subspace $\text{span}\{\bm{U}_{S, \alpha}\}_{\alpha \in \{0, 1\}^S}$ and let $\Pi_1^{\perp}$ be the projection into the orthogonal complement of this subspace. Let $\Pi_2$ be the projection map on $\mathbb{R}^{\gamma}$ into the subspace $\text{span}\{\bm{V}_u\}_{u \in S}$ and let $\Pi_2^{\perp}$ be the projection into the orthogonal complement of this subspace. Observe that $\text{span}\{\bm{V}_u\}_{u \in S}$ is contained in $\text{span}\{\bm{U}_{S, \alpha}\}_{\alpha \in \{0, 1\}^S}$ and so, $\norm{\Pi_2\bm{v}} \le \norm{\Pi_1 \bm{v}} \Longrightarrow \norm{\Pi_2^{\perp}\bm{v}} \ge \norm{\Pi_1^{\perp} \bm{v}}$ for any $\bm{v} \in \mathbb{R}^{\gamma}$. Now, using the fact that $G$ is $d$-regular, we have
\begin{align*}
\mathbb{E}_{\alpha}[\Gamma(R')] &=	\mathbb{E}_{\alpha}[\displaystyle\sum_{(u, v) \in E} Pr[u \in R' \wedge v \not\in R'] + Pr[v \in R' \wedge u \not\in R']]\\
&= \mathbb{E}_{\alpha}[\displaystyle\sum_{(u, v) \in E'} Pr[u \in R' \wedge v \not\in R']]\\
&= \displaystyle\sum_{(u, v) \in E'}\displaystyle\sum_{\alpha \in \{0, 1\}^S}\norm{\bm{U}_{S, \alpha}}^2\frac{\langle \bm{U}_{S, \alpha}, \bm{V}_{\{u\}}\rangle}{\langle \bm{U}_{S, \alpha}, \bm{U}_{S, \alpha}\rangle} \times \left(1 - \frac{\langle \bm{U}_{S, \alpha}, \bm{V}_{\{v\}}\rangle}{\langle \bm{U}_{S, \alpha}, \bm{U}_{S, \alpha}\rangle}\right)\\
&= \displaystyle\sum_{(u, v) \in E'}\displaystyle\sum_{\alpha \in \{0, 1\}^S}\langle \bm{U}_{S, \alpha}, \bm{V}_{\{u\}}\rangle -  \displaystyle\sum_{(u, v) \in E'}\displaystyle\sum_{\alpha \in \{0, 1\}^S}\frac{\langle \bm{U}_{S, \alpha}, \bm{V}_{\{u\}}\rangle\langle \bm{U}_{S, \alpha}, \bm{V}_{\{v\}}\rangle}{\langle \bm{U}_{S, \alpha}, \bm{U}_{S, \alpha}\rangle}\\
&= \displaystyle\sum_{(u, v) \in E'}\langle \bm{V}_{\phi}, \bm{V}_{\{u\}}\rangle -  \displaystyle\sum_{(u, v) \in E'}\displaystyle\sum_{\alpha \in \{0, 1\}^S}\frac{\langle \bm{U}_{S, \alpha}, \bm{V}_{\{u\}}\rangle\langle \bm{U}_{S, \alpha}, \bm{V}_{\{v\}}\rangle}{\langle \bm{U}_{S, \alpha}, \bm{U}_{S, \alpha}\rangle}\\
&= \displaystyle\sum_{(u, v) \in E'}\norm{\bm{V}_{\{u\}}}^2 -  \displaystyle\sum_{(u, v) \in E'}\langle \Pi_1\bm{V}_{\{u\}}, \Pi_1\bm{V}_{\{v\}}\rangle\\
&\le \Gamma(R) + \displaystyle\sum_{(u, v) \in E'}\langle \bm{V}_{\{u\}}, \bm{V}_{\{v\}}\rangle - \displaystyle\sum_{(u, v) \in E'}\langle \Pi_1\bm{V}_{\{u\}}, \Pi_1\bm{V}_{\{v\}}\rangle\\
&= \Gamma(R) + \displaystyle\sum_{(u, v) \in E'}\langle \Pi_1^{\perp}\bm{V}_{\{u\}}, \Pi_1^{\perp}\bm{V}_{\{v\}}\rangle\\
&\le \Gamma(R) + \frac{1}{2}\displaystyle\sum_{(u, v) \in E'}(\norm{\Pi_1^{\perp}\bm{V}_{\{u\}}}^2 + \norm{\Pi_1^{\perp}\bm{V}_{\{v\}}}^2)\\
&= \Gamma(R) + \displaystyle\sum_{(u, v) \in E}(\norm{\Pi_1^{\perp}\bm{V}_{\{u\}}}^2 + \norm{\Pi_1^{\perp}\bm{V}_{\{v\}}}^2)\\
&\le \Gamma(R) + \displaystyle\sum_{(u, v) \in E}(\norm{\Pi_2^{\perp}\bm{V}_{\{u\}}}^2 + \norm{\Pi_2^{\perp}\bm{V}_{\{v\}}}^2)\\
&\le \Gamma(R) + d\displaystyle\sum_{u\in V}\norm{\Pi_2^{\perp}\bm{V}_{\{u\}}}^2
\end{align*}
The final step of the proof is to argue that there exists a subset $S \in [n]_{\le r'}$ such that $d\displaystyle\sum_{u\in V}\norm{\Pi_2^{\perp}\bm{V}_{\{u\}}}^2 \le \frac{1}{(1 - \epsilon)\lambda_{r + 1}(L)}\Gamma(R)$.

Consider the $\gamma \times n$ matrix $X$ with columns $\bm{V}_{\{u\}}$ and singular values $\sigma_1 \ge \sigma_2 \ge \ldots \ge \sigma_n$. From lemma \ref{gscol}, we can obtain a subset $S \subseteq V$ of size $r'$ such that $\displaystyle\sum_{u \in V} \norm{\Pi_2^{\perp} \bm{V}_{\{u\}}}^2 \le \frac{1}{1 - \epsilon} \left(\displaystyle\sum_{i = r + 1}^{m'} \sigma_i^2\right)$. Using $\Gamma(R) \ge \displaystyle\sum_{(u, v) \in E} \norm{\bm{V}_{\{u\}} - \bm{V}_{\{v\}}}^2 = d\Tr(X^TXL)$ (remember that $L$ is normalized) and lemma \ref{hoff}, we get
\begin{align*}
\norm{X^TX + L}_F^2 &\ge \min_{\sigma \in S_n}\displaystyle\sum_{i = 1}^n(\lambda_i(X^TX) + \lambda_{\sigma(i)}(L))^2\\
\Longrightarrow \norm{X^TX}_F^2 + \norm{L}_F^2 + 2\Tr(X^TXL) &\ge \displaystyle\sum_{i = 1}^n(\lambda_i(X^TX))^2 + \displaystyle\sum_{i = 1}^n(\lambda_{\sigma(i)}(L))^2 + 2\displaystyle\sum_{i = 1}^n\lambda_i(X^TX)\lambda_{\sigma(i)}(L)\\
&\ge \norm{X^TX}_F^2 + \norm{L}_F^2 + 2\displaystyle\sum_{i = 1}^n\lambda_i(X^TX)\lambda_{\sigma(i)}(L)\\
\Longrightarrow~\quad \qquad\qquad\qquad\qquad\Tr(X^TXL) &\ge \displaystyle\sum_{i = 1}^n\lambda_i(X^TX)\lambda_{\sigma(i)}(L)\\
\Longrightarrow~~\quad\qquad \qquad\qquad\qquad\qquad\Gamma(R) &\ge d\displaystyle\sum_{i = 1}^n\sigma_i^2\lambda_{\sigma(i)}(L)\\
&\ge d\displaystyle\sum_{i = 1}^n\sigma_i^2\lambda_{n + 1 - i}(L)\text{ (by the rearrangement inequality)}\\
&\ge d\displaystyle\sum_{i = r + 1}^n\sigma_i^2\lambda_{r + 1}(L)\\
&\ge d\lambda_{r + 1}(L)(1 - \epsilon)\displaystyle\sum_{u \in V} \norm{\Pi_2^{\perp} \bm{V}_{\{u\}}}^2
\end{align*}
from which we get $\mathbb{E}_{\alpha}[\Gamma(R')] \le \Gamma(R) + \frac{\Gamma(R)}{(1 - \epsilon)\lambda_{r + 1}L}$ just like we wanted.
\end{proof}

Note that we actually only needed $r/ \epsilon$ levels of the hierarchy but to achieve the improved bound, we need $r / \epsilon^2$ levels.

To illustrate why this is an efficient algorithm for low threshold-rank graphs, suppose the $c$th largest eigenvalue of the normalized adjacency matrix is $\gamma = 0.6$ for some constant $c$. Then, $\lambda_{c + 1}(L) \ge \lambda_c(L) = 0.4$. So, we can get a $2.5(1 + \epsilon)$ approximation in $n^{O(c / \epsilon^2)}$ time, which explains why this algorithm works well on graphs whose spectrum has very few large eigenvalues.

%% file: chap3.tex
%% This is an example first chapter.  You should put chapter/appendix that you
%% write into a separate file, and add a line \include{yourfilename} to
%% main.tex, where `yourfilename.tex' is the name of the chapter/appendix file.
%% You can process specific files by typing their names in at the
%% \files=
%% prompt when you run the file main.tex through LaTeX.
\chapter{Lower bounds for the Sum of Squares Hierarchy}

\section{Max K-CSP}\label{first}

An instance of Max $K$-CSP over alphabet $[q]$ contains $m$ constraints $C_1, \ldots, C_m$ on $n$ variables $x_1, \ldots, x_n$. Each constraint $C_i$ is a boolean predicate on a subset of $K$ distinct variables. That is, if $T_i$ is the subset of $K$ distinct variables for the $i$th constraint, then $C_i$ is a function from $[q]^{T_i}$ to $\{0, 1\}$. An assignment is a mapping of the variables to $[q]$. We say that an assignment satisfies $C_i$ if the evaluation of $C_i$ on the assignment restricted to $T_i$ is $1$. The objective is to assign letters from $[q]$ to the variables $x_1, \ldots, x_n$ such that maximum number of constraints are satisfied. This general framework captures a large class of problems and they have natural SoS relaxations as was shown in Chapter $2$.

Kothari et al.\cite{kmow} gave tight tradeoffs between the density $\Delta = m / n$, the number of rounds of the SoS hierarchy and the optimum value of the relaxation for random CSP instances. They consider a graph naturally associated with the CSP instance and argue that if the graph satisfies a condition called the Plausibility assumption, then SoS vectors exist that exhibit almost perfect completeness, or in other words, the optimum SoS value is very close to $m$. Instances of Max $K$-CSP which are random (for the precise meaning, see definition \ref{rand}), satisfy the Plausibility assumption with high probability, so they serve as integrality gaps.

In our construction, the instance $I$ has $m$ $K$-ary constraints on $n$ variables. We fix a prime power $q$ and a subset $\mathcal{C} \subseteq \mathbb{F}_q^K$ and we consider instances $I$ where the variables are $x_1, \ldots, x_n$, the alphabet is $[q]$ and each constraint $P$ on the appropriate subset of variables $x_C = (x_i)_{i \in C}$ is of the form $P(x) = [\text{Is }x_C - b \in \mathcal{C}?]$ where $b \in \mathbb{F}_q^K$ and $\mathcal{C} \subseteq \mathbb{F}_q^K$. Here, $\mathcal{C}$ is fixed for all predicates but $b$ could be different.

There are $2$ natural graphs that we can associate $I$ with. Let the $m$ constraints be on the subsets $C_1, \ldots, C_m$ of $[n]$. We abuse notation to treat $C_i$ as a boolean function from $[q]^{C_i}$ to $\{0, 1\}$ which evaluates to $1$ if and only if that corresponding assignment is satisfied by the $i$th predicate.
\begin{itemize}
	\item Factor Graph: Consider the bipartite graph $G_I$ defined as follows. The left partition is $\{C_i \;|\; i \in [m]\}$, the set of constraints and the right partition is $\{x_j \;|\; j \in [n]\}$. $G_I$ contains the edge $(C_i, x_j)$ if and only if $C_i$ contains $x_j$. Therefore, $G_I$ has $m + n$ vertices and $mK$ edges since each vertex in the left partition has degree $K$.
	\item The Label Extended Factor graph: Fix a positive integer $\beta$ and consider the bipartite graph $H_{I, \beta} = (L, R, E)$ defined as follows. The left partition $L$ is $\{(C_i, \alpha) \;|\; i \in [m], \alpha \in [q]^K, C_i(\alpha) = 1\}$. The right partition $R$ is $\{(x_i, \alpha_{x_i}, j) \;|\; i \in [n], \alpha_{x_i} \in [q], j \in [\beta]\}$ with cardinality $nq\beta$. And $E$ contains all the edges $((C_i, \alpha), (x, \alpha_x, j))$ if $x \in C_i$ and $\alpha$ assigns $x$ to $\alpha_x$. Since each predicate is a random shift of $C$, we have that there are $|\mathcal{C}|$ possible values of $\alpha$ for each $C_i$, so $|L| = m|\mathcal{C}|$.  Therfore, $H_{I, \beta}$ has $N = m|\mathcal{C}| + nq\beta$ vertices and $m|\mathcal{C}|K\beta$ edges since each vertex in $L$ has degree $K\beta$.
\end{itemize}

\begin{definition}[Random Max $K$-CSP instance]\label{rand}
For a fixed $\mathcal{C} \subseteq \mathbb{F}_q^K$, a random instance of Max $K$-CSP of the form above proceeds by choosing the $m$ constraints independently as follows - For each constraint, we first choose the subset of $K$ variables uniformly at random and then choose $b \in \mathbb{F}_q^K$ uniformly at random.
\end{definition}

For an instance $I$, we define some parameters that will be of interest:

Let $\tau \ge 1$ be any integer such that $\mathcal{C}$ is $(\tau - 1)$-wise uniform. This means that the projection to any $\tau - 1$ coordinates from the $K$ coordinates is the uniform distribution in these coordinates. The minimum such $\tau$ is called the complexity of the predicate.

Let $1 \le \eta \le \frac{1}{2}$ be a parameter such that $\eta n$ is roughly the number of levels of SoS that we are considering. So, we would be interested in optimizing $\eta$.

Let $\zeta$ be any parameter such that $0 < \zeta < 1$ and $K \le \zeta \cdot \eta n$. Note that both $\eta, \zeta$ could depend on $n$.

\begin{definition}[$\tau$-subgraph]
Define a $\tau$-subgraph $H$ to be any edge-induced subgraph of $G_I$ such that each constraint vertex in $H$ has degree at least $\tau$ in $H$.
\end{definition}

Edge-induced essentially means that there are no isolated vertices. Also, note that the empty subgraph is a $\tau$-subgraph.

\begin{definition}[Plausible subgraphs]
Define a $\tau$-subgraph $H$ of $G_I$ with $c$ constraint vertices, $v$ variable vertices and $e$ edges to be \textit{plausible} if $v \ge e - \frac{\tau - \zeta}{2}c$.
\end{definition}

Now, we introduce the condition that we would like our factor graph to satisfy.

\begin{definition}[Plausibility assumption:]
All $\tau$-subgraphs $H$ of $G_I$ with at most $2\eta n$ constraint variables are plausible.
\end{definition}

This assumption roughly says that all small subsets of $L$ have large neighborhoods, that is, $G_I$ has expansion properties. The idea is that random instances satisfy the Plausibility assumption with high probability and instances whose factor graphs satisfy the Plausibility assumption exhibit perfect completeness for the SoS relaxation.

More precisely, fix a Max $K$-CSP instance $I$ and let $G_I$ be the factor graph. The following theorem shows SoS hardness for Max $K$-CSP assuming Plausibility.

\begin{theorem}[\cite{kmow}]\label{cspsos}
If the Plausibility assumption holds, then a degree $O(\zeta \eta n)$ SoS relaxation of the instance will have optimum value $m$.
\end{theorem}

In their paper, a more general version was shown for any $\tau$. The completeness value then depends on the statistical distance of the given predicate from a $\tau$-wise uniform distribution on $\mathbb{F}_q^K$. In fact, using essentially the same techniques, we can obtain a result where the constraints can have varying arity and their corresponding predicates can be arbitrary with possibly varying complexity, see for instance \cite{globaltolocal}. But for our purposes, this particular version will suffice.

We remark that the actual optimum value of $I$ will be concentrated around $m|\mathcal{C}|/q^K$ with high probability by a standard Chernoff bound. This is far from the SDP optimum if $|\mathcal{C}|$ is small compared to $q^K$, so this will be the usual setting in our hardness applications.

So, we would like to find the right value of $\eta$ so that all $\tau$-subgraphs with at most $2\eta n$ constraints are plausible. Such a bound can be obtained by a standard probabilistic argument leading to the following theorem.

\begin{theorem}[\cite{kmow}]\label{plaus}
Assume that $\mathcal{C}$ has complexity  at least $\tau \ge 3$. Fix $0 < \zeta < 0.99(\tau - 2)$ and $0 < \beta < \frac{1}{2}$. Then, with probability at least $1 - \beta$, the factor graph $G_I$ of a random Max $K$-CSP instance $I$ with $n$ variables and $m = \Delta n$ constraints will satisfy the Plausibility assumption with $\eta = \frac{1}{K}\left(\frac{\beta^{1 / (\tau - 2)}}{2^{K / (\tau - 2)}}\right)^{O(1)}\cdot \frac{1}{\Delta^{2 / (\tau - 2 - \zeta)}}$.
\end{theorem}

The following corollary is immediate from Theorem \ref{cspsos} and Theorem \ref{plaus}.

\begin{corollary}
Assume that $\mathcal{C}$ has complexity at least $\tau \ge 3$. Fix $0 < \zeta < 0.99(\tau - 2)$ and $0 < \beta < \frac{1}{2}$. Then, with probability at least $1 - \beta$, for a random Max $K$-CSP instance $I$ with $n$ variables and $m = \Delta n$ constraints, the level $O\left(\frac{1}{K}\left(\frac{\beta^{1 / (\tau - 2)}}{2^{K / (\tau - 2)}}\right)^{O(1)}\cdot \frac{n}{\Delta^{2 / (\tau - 2 - \zeta)}}\right)$ SoS relaxation will have perfect completeness, that is, it will have an optimum value of $m$.
\end{corollary}

We will illustrate some ideas involved in the proof of Theorem \ref{cspsos} when we describe pseudocalibration.

We remark that, over boolean predicates of constant arity, this lower bound is tight upto logarithmic factors, due to the following result on imperfect completeness of the SoS hierarchy on random CSPs.

\begin{theorem}[\cite{aow}, \cite{rrs}]
Let $I$ be a random Max $K$-CSP instance over boolean predicates, that is, $q = 2$. With high probability, the level $\tilde{O}(n / \Delta^{2/(\tau - 2)})$ SoS relaxation has optimum value strictly less than $m$.
\end{theorem}

We believe that their techniques should generalize to arbitary alphabet size as well.

\section{Max K-CSP for superconstant K}

If $\tau$ is a constant, as we have in most applications, note that the parameter $\eta$ as per Theorem \ref{plaus} drops off exponentially in $K$ (for a fixed $\tau$). This is fine if $K$ is constant, but for applications like Densest $k$-subgraph, $K$ is large (polynomial in $n$) and so, we need a different bound.

If we had $\tau = \Omega(K)$ as in k-SAT for example, we can use the existing bound because $\frac{K}{\tau - 2}$ will be at most a constant. But in many reductions, we can obtain good soundness generally when $\tau$ is low compared to $K$,  i.e., the predicate has low complexity. In that aspect, we will prove the following bound.

\begin{theorem}
Assume that $\mathcal{C}$ has complexity at least $\tau \ge 4$. Fix $0 < \zeta < 0.99(\tau - 2)$. If $10 \le K \le \sqrt{n}$ and $n^{\nu - 1} \le 1 / (10^8(\Delta K^{\tau - \zeta + 0.75})^{2 / (\tau - \zeta - 2)})$ for some $\nu > 0$, then the factor graph $G_I$ of a random Max $K$-CSP instance $I$ with $n$ variables and $m = \Delta n$ constraints will satisfy the Plausibility assumption with probability $1 - o(1)$, for $\eta = O(1 / (\Delta K^{\tau - \zeta})^{2 / (\tau - \zeta - 2)})$.
\end{theorem}

Note that exponential dependence on $K$ has been dropped assuming an inequality between $\Delta$ and $K$. To prove this theorem, we will be using the following lemma regarding expansion properties of the factor graph of random CSPs.

\begin{lemma}[Implicitly shown in \cite{bcv}]
If $\delta \ge 1.5, 10 \le K \le \sqrt{n}$ and $n^{\nu - 1} \le 1 / (10^8(\Delta K^{2\delta + 0.75})^{1 / (\delta - 1})$ for some $\nu > 0$, then the factor graph $G_I$ of a random Max $k$-CSP instance $I$ with $n$ variables and $m = \Delta n$ constraints will satisfy the following condition with probability $1 - o(1)$ for $\eta = O(1 / (\Delta K^{2\delta})^{1 / (\delta - 1)})$: Any set of $c$ constraints for $c \le \eta n$ will contain more than $(K - \delta)c$ variables.
\end{lemma}

\begin{proof}[Proof of Theorem \ref{csp_hardness}:]
Set $\delta = (\tau - \zeta) / 2$. Note that the conditions of the lemma are satisfied since $\delta \ge (4 - 1) / 2 = 1.5$ and the others are obvious. So, we get that any set of $c$ constraints for $c \le \eta n$ contain more than $(K - \delta)c$ variables. Now, we will prove the Plausibility assumption. Consider any $\tau$-subgraph $H$ of $G_I$ with $c$ constraint vertices, $v$ variable vertices and $e$ edges. We wish to prove that $v \ge e - \frac{\tau - \zeta}{2}c = e - \delta c$ with high probability. Rewrite this as $\delta c \ge (e - v)$. Note that the left hand side depends only on the number of constraint vertices in $H$. If $d_1, \ldots, d_v$ are the degrees of the variable vertices in $H$, then $d_i \ge 1$ since there are no isolated vertices and $e - v = (\displaystyle\sum_{i = 1}^v d_i) - v = \displaystyle\sum_{i = 1}^v (d_i - 1)$. Observe that for a fixed set of $c$ constraint vertices in $H$, $\displaystyle\sum_{i = 1}^v (d_i - 1)$ is maximized when $H$ contains all the neighbors of these $c$ constraint vertices. So, it suffices to prove the inequality only for such subgraphs $H$. Any such subgraph will have $e = Kc$ since all edges connected to the $c$ constraint vertices are chosen and we get that we have to prove $\delta c \ge Kc - v \Longleftrightarrow v \ge (K - \delta) c$. This is guaranteed by the lemma for $c \le \eta n$.
\end{proof}

So, we have the following corollary.
\begin{corollary}\label{csp_hardness}
Assume that $\mathcal{C}$ has complexity at least $\tau \ge 4$. Fix $0 < \zeta < 0.99(\tau - 2)$. If $10 \le K \le \sqrt{n}$ and $n^{\nu - 1} \le 1 / (10^8(\Delta K^{\tau - \zeta + 0.75})^{2 / (\tau - \zeta - 2)})$ for some $\nu > 0$, with high probability, for a random Max $K$-CSP instance $I$ with $n$ variables and $m = \Delta n$ constraints, the level $O\left(\frac{n}{(\Delta K^{\tau - \zeta})^{2 / (\tau - \zeta - 2)}}\right)$ SoS relaxation will have perfect completeness, that is, it will have an optimum value of $m$.
\end{corollary}

\section{Reductions to other problems}

Once we have shown an integrality gap for the SoS Hierarchy for Max $K$-CSPs, we can reduce this to show integrality gaps for the SoS Hierarchy for other problems directly. Roughly speaking, for a given problem $\Gamma$, using the hard instances $I$ of Max $K$-CSPs, we construct instances $J$ for the SoS relaxation of $\Gamma$ such that the following conditions hold:
\begin{itemize}
	\item Completeness: We produce SoS vectors such that they are feasible for the SoS relaxation for $\Gamma$
	\item Soundness: Our construction has to be robust in the sense that the actual optimum value of the instance is far away from the optimum value of the SoS relaxation, which can be bounded by the objective value of the feasible SoS solution constructed above
\end{itemize}

This idea was exploited by Tulsiani\cite{tul} to construct integrality gaps for Maximum Independent Set, Approximate Graph Coloring, Chromatic Number and Vertex Cover; and by Bhaskara et al.\cite{bcv} for Densest $k$-subgraph.

\subsection{Densest $k$-subgraph}

An instance of Densest $k$-Subgraph is an undirected unweighted graph $G = (V, E)$ and a positive integer $k$. The objective is to find a subset $W$ of $V$ with exactly $k$ vertices such that the number of edges with both endpoints in $W$, is maximized.

The first SoS hardness for the Densest $k$-subgraph problem was shown by Bhaskara et al.\cite{bcv}. The same construction with slightly different parameters and a stronger soudness argument was found to give a better gap by Manurangsi\cite{pasin}.

\begin{theorem}[\cite{bcv}, \cite{pasin}]
Fix a constant $0 < \rho < 1$. For all sufficiently large $n, q$ and integer $3\le D \le 10$, there exists an instance of Densest $k$-subgraph with $N = O(nq^{2D - 2 + \rho})$ vertices that demonstrates an integrality gap of $\Omega(q / \ln q)$ for the level $R = \Omega(\frac{n}{q^{(4D - 2 + 2\rho)/(D - 2) + 1}})$ SoS relaxation.
\end{theorem}

The graphs that exhibit this integrality gap are constructed from random instances of Max $K$-CSP. For a random instance $I$ of Max $K$-CSP, consider an instance $\Gamma$ of Densest $k$-subgraph with the graph being $G = H_{I, \Delta}$ and $k = 2m$.

For a prime number $q$, we set $K = q - 1, \Delta = 100q^{D + \rho}/K,  \eta = 1/(10^8(\Delta K^D)^{2 / (D - 2)}$ and $\mathcal{C}$ is a code (a code is a subspace of $\mathbb{F}_q^K$, treated as a vector space over $\mathbb{F}_q$) with dimension $D - 1$ and is $(D - 1)$-wise uniform. The existence of such a code is shown below.

\begin{lemma}
For an integer $D \ge 3$ and prime number $q \ge D$, there exists a code $\mathcal{C}$ in $\mathbb{F}_q^{q - 1}$ which has dimension $(D - 1)$ and is $(D - 1)$-wise uniform.
\end{lemma}

\begin{proof}
Fix a primitive root $g$ of $\mathbb{F}_q$. Consider the $(q - 1) \times (D - 1)$ matrix $A$ as follows.
\[A =
\begin{bmatrix}
1 & 1 & 1 & \dots  & 1 \\
1 & g & g^2 & \dots  & g^{D - 2} \\
1 & g^2 & g^4 & \dots  & g^{2(D - 2)} \\
\vdots & \vdots & \vdots & \ddots & \vdots \\
1 & g^{q - 1} & g^{2(q - 1)} & \dots  & g^{(D - 2)(q - 1)}
\end{bmatrix}
\]
Here, the $(i, j)$th entry of $A$ is $g^{(i - 1)(j - 1)}$ for $i \le q - 1, j \le D - 1$. Considering $A$ as a linear operator from $\mathbb{F}_q^{D - 1}$ to $\mathbb{F}_q^{q - 1}$, we set $\mathcal{C} = \text{Im}(A)$, the image of $A$. Note that the rank of $A$ is $D - 1$ since there are at $D - 1$ columns and the square matrix formed by the first $D - 1$ rows has determinant $\displaystyle\prod_{0 \le i < j \le D - 2}(g^j - g^i)$ which is nonzero since $g$ is a primitive root and $D - 2 \le q - 2$. Therefore, $\dim \mathcal{C} = D - 1$. To prove that $\mathcal{C}$ is $(D - 1)$ uniform, consider any $D - 1$ indices $r_1 < r_2 < \ldots < r_{D - 1}$ in $[q - 1]$. Suppose we wish to determine the number of elements $\bm{c} = (c_1, \ldots, c_{q - 1}) \in \mathcal{C}$ with fixed values of $\bm{c}_{r_i}$. This condition can be written as $A\bm{b} = \bm{c}$ for some vector $\bm{b} \in \mathbb{F}_q^{D - 1}$. Note that, the $(D - 1) \times (D - 1)$ submatrix of $A$ formed by choosing the rows with indices $r_1, \ldots, r_{D - 1}$ is nonsingular, since the determinant is $\displaystyle\prod_{0 \le i < j \le D - 2}(g^{r_j} - g^{r_i}) \neq 0$. This means that the system of $D - 1$ equations uniquely determine $\bm{b}$ and hence, $\bm{c}$ is also determined, which proves that there is a unique $\bm{c}$ with any choice of predetermined values $\bm{c}_{r_i}$. This also proves that $\mathcal{C}$ is $(D - 1)$ uniform.
\end{proof}

Using the SoS hardness results of Max $K$-CSP, We can show that the level $O(\eta n)$ SoS relaxation for Max $K$-CSP with the above parameters, for a sufficiently small constant $\zeta > 0$ achieves perfect completeness. The following lemma determines a lower bound on the completeness of the graph construction assuming perfect completeness for MAX $K$-CSP.

\begin{lemma}[\cite{bcv}]\label{completeness}
If there exists a perfect solution for $r$ levels of the SoS Hierarchy for $I$, then there exists a solution of value $\Delta mK$ for $r/K$ levels of the SoS hierarchy for $\Gamma$.
\end{lemma}

We describe the construction of the SoS vectors because that will be used in a subsequent application to proving SoS hardness of Minimum $p$-Union. The complete proof is given in \cite{bcv}. Suppose $\bm{W}_{(T, \alpha)}$ are the optimal SoS vectors for the level $r$ relaxation of $I$, for $\alpha \in [q]^T, |T| \le r$, then the level $r / K$ SoS vectors $\bm{V}_S$ for $\Gamma$ are as follows. Let $S$ be any subset of the vertices $V$ of $G$ with $|S| \le r / K$. Then, define $S_1 = \{(C_t, \alpha)\;|\;(C_t, \alpha) \in S\}$ be the left vertices in $S$ and $S_2 = \{(x_s, \alpha_{x_s}, j)\;|\;(x_s, \alpha_{x_s}, j) \in S\}$ be the right vertices in $S$. Say $(x_s, \alpha_{x_s})$ is contained in $S$ if either
\begin{itemize}
	\item $x_s \in C_t, \alpha(x_s) = \alpha_{x_s}$ for some $(C_t, \alpha) \in S_1$ or
	\item $(x_s, \alpha_{x_s}, j) \in S_2$ for some $j \in [\Delta]$
\end{itemize}
Say $S$ is inconsistent if there exists a variable $x_s$ with two distinct assignments in $S$, that is, there exist $\alpha_{x_s} \neq \alpha'_{x_s} \in [q]$ such that both $(x_s, \alpha_{x_s})$ and $(x_s, \alpha'_{x_s})$ are contained in $S$. If $S$ is inconsistent, we set $\bm{V}_S = \bm{0}$. Else, define $T = (\cup_{(C_t, \alpha) \in S_1} C_t) \cup (\cup_{j \in [\Delta]}\cup_{(x_s, \alpha_{x_s}, j) \in S_2}\{x_s\})$. Note that $|T| \le r$. We define $\beta \in [q]^T$ as follows: for every variable $x_s$ in $T$, choose $\alpha_{x_s}$ such that $(x_s, \alpha_{x_s})$ is contained in $S$ which happens for a unique $\alpha_{x_s}$ since $x_s \in T$ and $S$ is not inconsistent, and set $\beta(x_s) = \alpha_{x_s}$. Finally, we set $\bm{V}_{S} = \bm{W}_{(T, \beta)}$.

The improved soundness result is as below.

\begin{lemma}[\cite{pasin}]\label{pas}
Let $0 < \rho < 1$ be a constant. If $q/2 \le K \le q, q \ge 10000 / \rho, |\mathcal{C}| \le q^{10}$ and $\Delta \ge 100q^{1 + \rho}|\mathcal{C}|/K$, then the optimum solution for $\Gamma$ has at most $4000\Delta mK\ln q / (q\rho)$ edges with probability at least $1 - o(1)$.
\end{lemma}

\begin{corollary}
For any $0 < \epsilon < 1/14$, there exists an instance of Densest $k$-subgraph on $N$ vertices that demonstrates an integrality gap of $\Omega(N^{1/14 - \epsilon})$ for the level $N^{\Omega(\epsilon)}$ SoS relaxation.
\end{corollary}

\begin{proof}
The corollary follows from the above theorem by setting $D = 4, q = N^{1/14 - \epsilon / 2}$ and $\rho = \epsilon / 1000$.
\end{proof}

\subsection{Densest $k$-subhypergraph}

This is a natural variant of Densest $k$-subgraph for hypergraphs. An instance of Densest $k$-subhypergraph is an unweighted hypergraph $G = (V, E)$ and a positive integer $k$ and the objective is to find a subset $W$ of $V$ with exactly $k$ vertices such that the number of edges $e \in E$ with $e \subseteq W$, is maximum.

For any constant $\epsilon > 0$, for Densest $k$-subhypergraph on $3$-uniform hypergraphs, Chlamt\'{a}\v{c} et al.\cite{dksh} gave an $O(n^{4(4-\sqrt{3})/13 + \epsilon})$ approximation. Here, we present lower bounds for the natural SoS hierarchy for the general problem.

The SoS relaxation is almost identical to Densest $k$-subgraph but this time, the objective function is $\displaystyle\sum_{F \in E}\prod_{u \in F} x_u$. Assume $V = [n]$. The level $r$ SoS relaxation is as follows.
\begin{align*}
\text{Maximize}\qquad&\displaystyle\sum_{F \in E}\norm{\bm{V}_F}^2&&\\
\text{subject to}\qquad&\displaystyle\sum_{v \in V}\langle \bm{V}_{\{v\}}, \bm{V}_S \rangle = k\norm{\bm{V}_S}^2 &\forall& S \in [n]_{\le r}\\
&\langle \bm{V}_{S_1}, \bm{V}_{S_2}\rangle = \langle \bm{V}_{S_3}, \bm{V}_{S_4}\rangle &\forall& S_1 \cup S_2 = S_3 \cup S_4\text{ and }S_i \in [n]_{\le r}\\
&\langle \bm{V}_{S_1}, \bm{V}_{S_2}\rangle \ge 0 &\forall& S_1, S_2 \in [n]_{\le r}\\
&\norm{\bm{V}_{\phi}}^2 = 1&&
\end{align*}
%\begin{eqnarray*}
%	&\text{Maximize }\displaystyle\sum_{F \in E}\norm{\bm{V}_F}^2&\\
%	\text{subject to }&\displaystyle\sum_{v \in V}\langle \bm{V}_{\{v\}}, \bm{V}_S \rangle = k\norm{\bm{V}_S}^2 &\forall S \in [n]_{\le r}\\
%	&\langle \bm{V}_{S_1}, \bm{V}_{S_2}\rangle = \langle \bm{V}_{S_3}, \bm{V}_{S_4}\rangle &\forall S_1 \cup S_2 = S_3 \cup S_4\text{ and }S_i \in [n]_{\le r}\\
%	&\langle \bm{V}_{S_1}, \bm{V}_{S_2}\rangle \ge 0 &\forall S_1, S_2 \in [n]_{\le r}\\
%	& \norm{\bm{V}_{\phi}}^2 = 1&\\
%\end{eqnarray*}

We reduce integrality gaps for the SoS hierarchy for Densest $k$-subgraph to integrality gaps for the SoS hierarchy for Densest $k$-subhypergraph. The maximum number of vertices in any hyperedge is called the arity of the hypergraph.

\begin{theorem}\label{hyper}
For any positive integer $t$, if the integrality gap of $r \ge 2^t$ levels of the SoS hierarchy for Densest $k$-subgraph is $\alpha(n)$ for instances with $n$ vertices and number of edges that is not bounded as $n$ grows, the integrality gap of $r$ levels of SoS hierarchy for Densest $k$-subhypergraph on $n$ vertices of arity $2^t$ is at least $(\alpha(n) / 2^{t + 2})^{2^{t - 1}}$.
\end{theorem}

\begin{proof}
Let $\rho = 2^{t - 1}$. Consider instances $I = (G, k)$ for Densest $k$-subgraph that demonstrate an integrality gap of $\alpha(n)$ for $r$ levels of the SoS Hierarchy. Let $G = (V, E)$ and here, we have $n = |V|$. Consider the elements of $E$ as sets of size $2$. We will construct an hypergraph $G' = (V, E')$ of arity $2\rho$ as follows. We set $E' = \{\displaystyle\cup_{i \le \rho} f_i \;|\; f_i \in E\}$. Note that the arity of $G'$ is at most $2\rho$ by construction. For sufficiently large $n$, since the number of edges is not bounded, we have that the arity of $G'$ is exactly $2\rho = 2^t$. We consider the instance $J = (G', k)$ on $n$ vertices.

Let $\bm{V}_S$ be the optimal SoS vectors for $I$ and let $FRAC, OPT$ be the optimum SoS relaxation value and actual optimum for $I$ respectively. So, $FRAC = \displaystyle\sum_{e \in E} \norm{\bm{V}_e}^2 \ge \alpha(n) OPT$.

We use the same SoS vectors for this new instance. Note that they are trivially a feasible solution. Let $FRAC', OPT'$ be the optimum level $r$ SoS relaxation value and actual optimum for $J$ respectively. First, observe that $OPT' \le OPT^{\rho}$. This is because, if we consider any $k$ vertices in $G'$, if the induced subgraph on these vertices of $G$ contains $l$ edges, then, by construction, the induced subgraph on these vertices of $G'$ contain at most $l^{\rho}$ edges. But we have $l \le OPT$ which implies that any $k$ vertices in $G'$ have at most $OPT^{\rho}$ edges and hence, $OPT' \le OPT^{\rho}$.

We will use the following claim which will be proved later.

\begin{claim}
For an integer $p \ge 0$, let $T = E^{2^p}$ be the set of ordered tuples of $2^p$ edges. Then, $\displaystyle\sum_{(f_1, \ldots, f_{2^p}) \in T} \norm{\bm{V}_{f_1 \cup \ldots \cup f_{2^p}}}^2 \ge FRAC^{2^p}$.
\end{claim}

Now, consider the set $T = E^{\rho}$. For each element $(f_1, \ldots, f_{\rho})$ of $T$, by construction, there is at least one hyperedge $F$ in $G'$ with $F = f_1\cup \ldots f_{\rho}$. Also, each element $F$ of $E'$ is the union of $\rho$ edges and so, can be written as $f_1 \cup \ldots \cup f_{\rho}$ for some $(f_1, \ldots, f_{\rho}) \in T$. Moreover, there are at most $(4\rho^2)^{\rho}$ such elements in $T$ for a fixed $F$. This is because each $f_i$ has at most $|F|(|F| - 1) \le (2\rho)(2\rho - 1) \le 4\rho^2$ choices. So, we have
\begin{align*}
FRAC' = \sum_{F \in E'} \norm{\bm{V}_F}^2 &\ge \frac{1}{((4\rho^2)^{\rho}} \times \displaystyle\sum_{(f_1, \ldots, f_{\rho}) \in T} \norm{\bm{V}_{f_1 \cup \ldots \cup f_{\rho}}}^2\\
&\ge \frac{FRAC^{\rho}}{4^{\rho}\rho^{2\rho}}
\end{align*}

So, we have that the integrality gap of $J$ is at least
\[\frac{FRAC'}{OPT'} \ge \frac{FRAC^{\rho}}{4^{\rho}\rho^{2\rho}OPT^{\rho}} \ge \left(\frac{\alpha(n)}{2^{t + 2}}\right)^{2^{t - 1}}\]
This completes the proof of the theorem.
\end{proof}

It remains to prove the claim.

\begin{proof}[Proof of Claim]
The proof will be by induction on $p$. When $p = 0$, we have $\displaystyle\sum_{f \in E}\norm{\bm{V}_f}^2 = FRAC$ by definition. Let $T' = E^{2^{p - 1}}$. Fix an integer $p \ge 1$. Assume $\displaystyle\sum_{(f_1, \ldots, f_{2^{p - 1}}) \in T'} \norm{\bm{V}_{f_1 \cup \ldots \cup f_{2^{p - 1}}}}^2 \ge FRAC^{2^{p - 1}}$ as the induction hypothesis and consider
\begin{align*}
\displaystyle\sum_{(f_1, \ldots, f_{2^p}) \in T} \norm{\bm{V}_{f_1 \cup \ldots \cup f_{2^p}}}^2 &= \displaystyle\sum_{(f_1, \ldots, f_{2^p}) \in T} \langle\bm{V}_{f_1 \cup \ldots \cup f_{2^p}}, \bm{V}_{f_1 \cup \ldots \cup f_{2^p}}\rangle\\
&= \displaystyle\sum_{(f_1, \ldots, f_{2^p}) \in T} \langle\bm{V}_{f_1 \cup \ldots \cup f_{2^{p - 1}}}, \bm{V}_{f_{2^{p - 1} + 1} \cup \ldots \cup f_{2^p}}\rangle\\
&= \langle\displaystyle\sum_{(f_1, \ldots, f_{2^{p - 1}}) \in T'} \bm{V}_{f_1 \cup \ldots \cup f_{2^{p - 1}}}, \displaystyle\sum_{(f_1, \ldots, f_{2^{p - 1}}) \in T'} \bm{V}_{f_1 \cup \ldots \cup f_{2^{p - 1}}}\rangle\\
&\ge \left(\frac{\langle\displaystyle\sum_{(f_1, \ldots, f_{2^{p - 1}}) \in T'} \bm{V}_{f_1 \cup \ldots \cup f_{2^{p - 1}}}, \bm{V}_{\phi}\rangle}{\norm{\bm{V}_{\phi}}^2}\right)^2\\
&= \langle\displaystyle\sum_{(f_1, \ldots, f_{2^{p - 1}}) \in T'} \bm{V}_{f_1 \cup \ldots \cup f_{2^{p - 1}}}, \bm{V}_{\phi}\rangle^2\\
&= \left(\displaystyle\sum_{(f_1, \ldots, f_{2^{p - 1}}) \in T'} \langle\bm{V}_{f_1 \cup \ldots \cup f_{2^{p - 1}}}, \bm{V}_{\phi}\rangle\right)^2\\
&= \left(\displaystyle\sum_{(f_1, \ldots, f_{2^{p - 1}}) \in T'} \langle\bm{V}_{f_1 \cup \ldots \cup f_{2^{p - 1}}}, \bm{V}_{f_1 \cup \ldots \cup f_{2^{p - 1}}}\rangle\right)^2\\
&= \left(\displaystyle\sum_{(f_1, \ldots, f_{2^{p - 1}}) \in T'} \norm{\bm{V}_{f_1 \cup \ldots \cup f_{2^{p - 1}}}}^2\right)^2\\
&\ge (FRAC^{2^{p - 1}})^2 = FRAC^{2^p}
\end{align*}
Here, the second and second last equalities follow from properties of SoS vectors; and the first inequality follows from Cauchy-Schwarz and we used the fact that $\norm{\bm{V}_{\phi}}^2 = 1$. This completes the proof of the claim.
\end{proof}

Note in particular that when $t$ is constant, we get an $\Omega(\alpha(n)^{2^{t - 1}})$ integrality gap for an instance with arity $2^t$.

Using the SoS hardness result for Densest $k$-subgraph described in the previous section and our theorem, we arrive at the following SoS hardness result for Densest $k$-subhypergraph for any arbitary arity $\rho \ge 2$ where we apply the theorem to construct hypergraphs with arity $2^{\lfloor \log \rho \rfloor}$.

\textbf{Corollary 3.13:} For any integer $\rho \ge 2$, $n^{\Omega(\epsilon)}$ levels of the SoS hierarchy has an integrality gap of at least $\Omega(n^{(2^{\lfloor \log \rho \rfloor} / 28)}) \ge \Omega(n^{\rho / 56})$ for Densest k-subhypergraph on $n$ vertices of arity $\rho$.

\subsection{Minimum p-Union}

An instance of Minimum $p$-Union is a positive integer $p$ and a collection of $m$ subsets $S_1, \ldots, S_m$ of an universe of $n$ elements. The objective is to choose exactly $p$ of these sets such that the size of their union is minimized. This problem was first studied by Chlamt\'{a}\v{c} et al.\cite{dksh} and the current best known approximation algorithm is an $O(m^{1/4})$ approximation by Chlamt\'{a}\v{c} et al.\cite{yury}.

This can be thought of as a variant of the Densest $k$-subgraph problem. The relation to Densest $k$-subgraph comes from an intermediate problem also known as the Smallest $m$-Edge Subgraph problem, where we are given a graph $G$ and an integer $m$, the objective is to choose exactly $m$ edges so that the number of vertices that are contained in these chosen edges is minimum. Intuitively, if the number of vertices in the final edge induced subgraph is small, then the subgraph should be dense. Indeed, we will exploit this intuition in our integrality gap construction. Smallest $m$-Edge Subgraph problem can be thought of as the restricted version of Minimum $p$-Union where each set has size $2$. Minimum $p$-Union can also be viewed as a variant of the Maximum $k$-coverage problem where we have the same input but the objective is to maximize the size of the union. This problem is completely understood in the sense that there is a $1 - 1/e$ approximation and Feige\cite{coverage} showed it is also tight.

This problem has an equivalent formulation in terms of bipartite graphs, known as the Small Set Bipartite Vertex Expansion (SSBVE) problem which can also be viewed as the bipartite version of the Small Set Expansion problem. In SSBVE, we are given an integer $l$ and a bipartite graph $G = (L, R, E)$ with $n$ vertices, with labelled left and right partitions $L$ and $R$. The objective is to choose exactly $l$ vertices from $L$ such that the size of the neighborhood of these $l$ vertices is minimized. The connection with Minimum $p$-Union is straightforward and comes by identifying the sets with $L$ and the universe with $R$. So, we can interchangeably work with either problem.

In the basic program for SSBVE, we have variables $x_u$ for every vertex $u$, where $x_u$ for $u \in L$ indicates whether $u$ is picked among the $l$ vertices and $x_v$ for $v \in R$ indicates whether any neighbor of $v$ is picked among the $l$ vertices. Then, $\displaystyle\sum_{u \in L} x_u = l$ since exactly $l$ vertices from $L$ have to be picked. We set $x_u \le x_v$ for all edges $(u, v)$ with $u \in L, v \in R$ so that whenever $u$ is picked, $x_v$ for all neighbors $v$ of $u$ are assigned $1$. With these constraints, it is clear that if we try to minimize $\displaystyle\sum_{v \in R} x_v$, it will also be the size of the neighborhood. The SoS relaxation for $r$ levels for SSBVE is as follows:
\begin{align*}
\text{Minimize}\qquad&\displaystyle\sum_{v \in R}\norm{\bm{V}_{\{v\}}}^2&&\\
\text{subject to}\qquad&\displaystyle\sum_{u \in L}\langle \bm{V}_{\{u\}}, \bm{V}_S \rangle = l\norm{\bm{V}_S}^2 &\forall& S \in [n]_{\le r}\\
&\langle \bm{V}_{\{u\}}, \bm{V}_S \rangle \le \langle \bm{V}_{\{v\}}, \bm{V}_S \rangle&\forall& (u, v) \in E, u\in L, v \in R, S \in [n]_{\le r}\\
&\langle \bm{V}_{S_1}, \bm{V}_{S_2}\rangle = \langle \bm{V}_{S_3}, \bm{V}_{S_4}\rangle &\forall& S_1 \cup S_2 = S_3 \cup S_4\text{ and }S_i \in [n]_{\le r}\\
&\langle \bm{V}_{S_1}, \bm{V}_{S_2}\rangle \ge 0 &\forall& S_1, S_2 \in [n]_{\le r}\\
&\norm{\bm{V}_{\phi}}^2 = 1&&
\end{align*}
%\begin{eqnarray*}
%	&\text{Maximize }\displaystyle\sum_{v \in R}\norm{\bm{V}_v}^2&\\
%	\text{subject to }&\displaystyle\sum_{u \in L}\langle \bm{V}_{\{u\}}, \bm{V}_S \rangle = l\norm{\bm{V}_S}^2 &\forall S \in [n]_{\le r}\\
%	\text{subject to }&\langle \bm{V}_{\{u\}}, \bm{V}_S \rangle \le \langle \bm{V}_{\{v\}}, \bm{V}_S \rangle&\forall (u, v) \in E, u\in L, v \in R, S \in [n]_{\le r}\\
%	&\langle \bm{V}_{S_1}, \bm{V}_{S_2}\rangle = \langle \bm{V}_{S_3}, \bm{V}_{S_4}\rangle &\forall S_1 \cup S_2 = S_3 \cup S_4\text{ and }S_i \in [n]_{\le r}\\
%	&\langle \bm{V}_{S_1}, \bm{V}_{S_2}\rangle \ge 0 &\forall S_1, S_2 \in [n]_{\le r}\\
%	& \norm{\bm{V}_{\phi}}^2 = 1&\\
%\end{eqnarray*}

Chlamt\'{a}\v{c} et al.\cite{yury} showed an integrality gap of $\tilde{\Omega}(\min(l, n / l))$ for a basic SDP relaxation of this problem. We obtain integrality gaps for the general SoS relaxation for SSBVE.

\begin{theorem}
Fix $0 < \rho < 1$. For all sufficiently large $n, q$ and integer $3 \le D \le 10$, there exist instances of SSBVE on $N = O(nq^{3D - 2 + \rho})$ vertices that demonstrate an integrality gap of $\Omega(q^{1/2 - o(1)})$ for the level $\Omega(n / (q^{5 + 6 / (D - 2) + 2\rho/(D - 2)}))$ SoS relaxation.
\end{theorem}

\begin{proof}
We will use a modification of the integrality gap instance for Densest $k$-subgraph obtained from random CSPs as was illustrated earlier.

Take a random instance $I$ of Max $K$-CSP with $m$ constraints on variables $\{x_1, \ldots, x_n\}$, alphabet $[q]$ and with optimum value of the level $r = O(\eta n)$ SoS relaxation being $m$ (perfect completeness). The parameters are as before, $K = q - 1, \Delta = 100q^{D + \rho}/K, \eta = 1/(10^8(\Delta K^D)^{2 / (D - 2)}$ and $\mathcal{C} \subseteq \mathbb{F}_q^K$ has dimension $D - 1$ and is $(D - 1)$-wise uniform.

Consider the label extended factor graph $G = H_{I, \Delta}$ and construct the instance $J = (H, l)$ of SSBVE as follows. $H$ is the bipartite graph obtained from $G$ by subdividing the edges of $G$. That is, $H = (L, R, E')$ where $L$ corresponds to the edges of $G$; $R$ corresponds to the vertices of $G$; and $E'$ contains the edge $(e, u)$ for $e \in L, u \in R$ if and only if the edge $e$ contains $u$ in $G$. Finally, set $l = \Delta mK$. We will argue that $J$ exhibits the desired integrality gap.

Suppose $G = (V, E)$ with $V = [n]$. From lemma \ref{completeness}, we have SoS vectors $\bm{V}_S$ for subsets $S$ of $V$ of size at most $r' = r / K$ that satisfy the following properties.
\begin{itemize}
	\item $\displaystyle\sum_{u \in V}\langle \bm{V}_{\{u\}}, \bm{V}_S \rangle = 2m\norm{\bm{V}_S}^2$ for all $S \in [n]_{\le r'}$
	\item $\langle \bm{V}_{S_1}, \bm{V}_{S_2}\rangle = \langle \bm{V}_{S_3}, \bm{V}_{S_4}\rangle$ for all $S_1 \cup S_2 = S_3 \cup S_4$ and $S_i \in [n]_{\le r'}$
	\item $\langle \bm{V}_{S_1}, \bm{V}_{S_2}\rangle \ge 0$ for all $S_1, S_2 \in [n]_{\le r'}$
	\item $\norm{\bm{V}_{\phi}}^2 = 1$
\end{itemize}

It is important that the vectors $\bm{V}_S$ are the same vectors as constructed in the proof of lemma \ref{completeness}. Remember that they are constructed from $\bm{W}_{(S, \alpha)}$, the SoS vectors for the level $r$ relaxation of the Max $K$-CSP instance we are reducing from, for $\alpha \in [q]^S, |S| \le r$.

We will describe level $(r' / 2 - 4)$ SoS vectors for SSBVE, $\bm{U}_S$, as follows. Consider any subset $S$ of $L \cup R$ with at most $(r'/2 - 4)$ vertices. For $T \subseteq L$, let $\mathcal{N}(T)$ denote the set of neighbors of $T$. Define $\mathcal{B}(S) = (R \cap S) \cup \mathcal{N}(L \cap S)$. Note that $\mathcal{B}(S) \subseteq R = V$. Define $\bm{U}_S = \bm{V}_{\mathcal{B}(S)}$. Note that this is well defined since $|S| \le r'/2 -  4 \Longrightarrow |\mathcal{B}(S)| \le r' - 2$ which follows since $|N(\{u\})| = 2$ for any $u \in L$.

We first prove that these vectors $\bm{U}_S$ form a feasible solution. For any $S_1, S_2 \subseteq L \cup R$ with $|S_1|, |S_2| \le r'/2 - 4$, we have $\langle \bm{U}_{S_1}, \bm{U}_{S_2}\rangle = \langle \bm{V}_{\mathcal{B}(S_1)}, \bm{V}_{\mathcal{B}(S_2)}\rangle \ge 0$. Consider $S_1, S_2, S_3, S_4 \subseteq L \cup R$ with $S_1 \cup S_2 = S_3 \cup S_4$ and $|S_1|, |S_2|, |S_3|, |S_4| \le r'/2 - 4$. If $S_1 \cup S_2 = S_3 \cup S_4 = L' \cup R'$ for $L' \subseteq L, R' \subseteq R$, then $\mathcal{B}(S_1) \cup \mathcal{B}(S_2) = R' \cup \mathcal{N}(L') = \mathcal{B}(S_3) \cup \mathcal{B}(S_4)$. So, we get $\langle \bm{U}_{S_1}, \bm{U}_{S_2}\rangle = \langle \bm{V}_{\mathcal{B}(S_1)}, \bm{V}_{\mathcal{B}(S_2)}\rangle = \langle \bm{V}_{\mathcal{B}(S_3)}, \bm{V}_{\mathcal{B}(S_4)}\rangle = \langle \bm{U}_{S_3}, \bm{U}_{S_4}\rangle$. We also have $\norm{\bm{U}_{\phi}}^2 = \norm{\bm{V}_{\phi}}^2 = 1$.

Fix any subset $S \subseteq L \cup R$ with $|S| \le r' / 2 - 4$. For any edge $(u, v)$ in $H$ with $u \in L, v \in R$, suppose $(u, w)$ with $w \neq v$ is the other unique edge in $H$, then we have $\langle \bm{U}_{\{u\}}, \bm{U}_S\rangle = \langle \bm{V}_{\{v, w\}}, \bm{V}_{\mathcal{B}(S)}\rangle = \norm{\bm{V}_{\{v, w\} \cup \mathcal{B}(S)}}^2$ and similarly, $\langle \bm{U}_{\{v\}}, \bm{U}_S\rangle = \norm{\bm{V}_{\{v\} \cup \mathcal{B}(S)}}^2$. Here, note that $|\{v, w\} \cup \mathcal{B}(S)|, |\{v\} \cup \mathcal{B}(S)| \le r'$. Using the inequality $\norm{\bm{V}_{S_2}} \le \norm{\bm{V}_{S_1}}$ for $S_1 \subseteq S_2 \subseteq [n]_{\le r'}$ (Indeed, $\norm{\bm{V}_{S_2}}^2 = \langle\bm{V}_{S_2}, \bm{V}_{S_1}\rangle \le \norm{\bm{V}_{S_2}}\cdot\norm{\bm{V}_{S_1}}$ by the Cauchy Schwarz inequality), we get $\langle \bm{U}_{\{u\}}, \bm{U}_S\rangle = \norm{\bm{V}_{\{v, w\} \cup \mathcal{B}(S)}}^2 \le \norm{\bm{V}_{\{v\} \cup \mathcal{B}(S)}}^2 = \langle \bm{U}_{\{v\}}, \bm{U}_S\rangle$ for all edges $(u, v) \in H$.

Finally, we need to show that $\displaystyle\sum_{u \in L}\langle \bm{U}_{\{u\}}, \bm{U}_S \rangle = l\norm{\bm{U}_S}^2$. We have $\displaystyle\sum_{u \in L}\langle \bm{U}_{\{u\}}, \bm{U}_S \rangle = \displaystyle\sum_{u \in L}\langle \bm{V}_{\mathcal{B}(\{u\})}, \bm{V}_{\mathcal{B}(S)} \rangle = \displaystyle\sum_{(v, w) \in E}\langle \bm{V}_{\{v, w\}}, \bm{V}_{\mathcal{B}(S)} \rangle$. Note that each edge $(v, w) \in E$ is between vertices of the form $(C_i, \alpha)$ where $i \le m, \alpha \in [q]^K, C_i(\alpha) = 1$, and $(x_j, \alpha_{x_j}, j')$ where $j \le n, \alpha_{x_j} \in [q], j' \in [\Delta]$ such that $x_j \in C_i, \alpha(x_j) = \alpha_{x_j}$. Then, by construction, $\bm{V}_{\{v, w\}} = \bm{W}_{(C_i, \alpha)}$ and this term appears $K\Delta$ times for each $(C_i, \alpha)$. Also, we have $\bm{U}_S = \bm{V}_{\mathcal{B}(S)} = \bm{W}_{(T, \beta)}$ for some $T, \beta$ with $\beta \in [q]^T, |T| \le r$. So, we get $\displaystyle\sum_{(v, w) \in E}\langle \bm{V}_{\{v, w\}}, \bm{V}_{\mathcal{B}(S)} \rangle = K\Delta\displaystyle\sum_{(C_i, \alpha) \in V} \langle \bm{W}_{(C_i, \alpha)}, \bm{W}_{(T, \beta)} \rangle$. Now, we use the fact that, for any $i \le m$, if $\mathcal{A}_i$ is the set of satisfying partial assignments $\alpha\in [q]^{C_i}$ with $C_i(\alpha) = 1$, that is $\mathcal{A}_i = \{\alpha \;|\; (C_i, \alpha) \in G\}$, then $\displaystyle \sum_{\alpha \in \mathcal{A}_i} \bm{W}_{(C_i, \alpha)} = \bm{W}_{(\phi, \phi)}$ which is true because
\begin{align*}
\norm{\displaystyle \sum_{\alpha \in \mathcal{A}_i} \bm{W}_{(C_i, \alpha)} - \bm{W}_{(\phi, \phi)}}^2 &= \langle \displaystyle \sum_{\alpha \in \mathcal{A}_i} \bm{W}_{(C_i, \alpha)} - \bm{W}_{(\phi, \phi)}, \displaystyle \sum_{\alpha \in \mathcal{A}_i} \bm{W}_{(C_i, \alpha)} - \bm{W}_{(\phi, \phi)}\rangle\\
&= \displaystyle \sum_{\alpha_1 \in \mathcal{A}_i}\sum_{\alpha_2 \in \mathcal{A}_i} \langle\bm{W}_{(C_i, \alpha_1)}, \bm{W}_{(C_i, \alpha_2)}\rangle - 2\displaystyle\sum_{\alpha \in \mathcal{A}_i} \langle \bm{W}_{(C_i, \alpha)}, \bm{W}_{(\phi, \phi)}\rangle + \norm{\bm{W}_{(\phi, \phi)}}^2\\
&= \displaystyle \sum_{\alpha \in \mathcal{A}_i} \langle\bm{W}_{(C_i, \alpha)}, \bm{W}_{(C_i, \alpha)}\rangle - 2\displaystyle\sum_{\alpha \in \mathcal{A}_i} \norm{\bm{W}_{(C_i, \alpha)}}^2 + 1\\
&= 1 - \displaystyle\sum_{\alpha \in \mathcal{A}_i} \norm{\bm{W}_{(C_i, \alpha)}}^2 = 0
\end{align*}
Here, we used the facts that $\langle\bm{W}_{(C_i, \alpha_1)}, \bm{W}_{(C_i, \alpha_2)}\rangle = 0$ for $\alpha_1 \neq \alpha_2$, $\langle \bm{W}_{(C_i, \alpha)}, \bm{W}_{(\phi, \phi)}\rangle = \norm{\bm{W}_{(C_i, \alpha)}}^2$ and since we have a perfect solution, $\displaystyle\sum_{\alpha \in \mathcal{A}_i} \norm{\bm{W}_{(C_i, \alpha)}}^2 = 1$ for all $i \le m$. So, we get
\begin{align*}
\displaystyle\sum_{u \in L}\langle \bm{U}_{\{u\}}, \bm{U}_S \rangle &= \displaystyle\sum_{(v, w) \in E}\langle \bm{V}_{\{v, w\}}, \bm{V}_{\mathcal{B}(S)} \rangle\\
&= K\Delta\displaystyle\sum_{(C_i, \alpha) \in V} \langle \bm{W}_{(C_i, \alpha)}, \bm{W}_{(T, \beta)} \rangle\\
&= K\Delta\displaystyle\sum_{i = 1}^m \langle \bm{W}_{(\phi, \phi)}, \bm{W}_{(T, \beta)} \rangle\\
&= \Delta mK \norm{\bm{W}_{(T, \beta)}}^2\\
&= l\norm{\bm{W}_{(T, \beta)}}^2 = l\norm{\bm{U}_S}^2
\end{align*}
as required.

So, we have shown that the vectors $\bm{U}_S$ form a feasible solution for the level $r'/2 - 4 = \Omega(r/K)$  SoS relaxation. The objective value of this solution is $FRAC' = \displaystyle\sum_{v \in R}\norm{\bm{U}_{\{v\}}}^2 = \displaystyle\sum_{v \in V}\norm{\bm{V}_{\{v\}}}^2 = 2m$.

Let $OPT'$ be the value of the actual optimum solution for $J$. The following claim guarantees soundness of our instance.

\begin{claim}
Fix a constant $0 < \rho < 1$. If $q \ge 10000/\rho, |\mathcal{C}| \le q^{10}$, then $OPT' \ge m\sqrt{q\rho} / (80\sqrt{\ln q})$.
\end{claim}

So, we get an integrality gap of at least $FRAC' / OPT' = \sqrt{q\rho}/(160\sqrt{\ln q}) = \Omega(q^{1/2 - o(1)})$ for the instance $J$ with $N = m|\mathcal{C}|K\Delta + m|\mathcal{C}| + nq\Delta = O(nq^{3D - 2 + 2\rho})$ vertices where the number of levels of the SoS relaxation is \[\Omega\left(\frac{r}{K}\right) = \Omega\left(\frac{n}{K(\Delta K^D)^{2/(D-2)}}\right) = \Omega\left(\frac{n}{q^{5 + 6 / (D - 2) + 2\rho/(D - 2)}}\right)\]
which proves the theorem.
\end{proof}

It remains to prove the claim.

\begin{proof}[Proof of Claim:]
Assume for the sake of contradiction that there exists a set of $l = \Delta mK$ vertices in $L$ that has a neighborhood of size $m' < m\sqrt{q\rho} / (80\sqrt{\ln q})$. Partition the set of $m'$ vertices arbitrarily into $m' / m$ subsets of size $m$, denoted $R_1, \ldots, R_{m' / m}$. The neighbors of any vertex $u$ among the chosen $l$ vertices of $L$ have their endpoints in $R_i, R_j$ for some $1 \le i \le j \le m'/m$, not necessarily distinct. So, an upper bound on $l$ is $\displaystyle\sum_{1 \le i \le j \le m'/m} E(R_i, R_j)$ where $E(R_i, R_j)$ is the number of edges (think of a pre-fixed edge orientation to avoid overcounting) with their endpoints being in $R_i, R_j$ respectively. But note that $|R_i \cup R_j| \le 2m$ and so, by Lemma \ref{pas}, we have that $|E(R_i, R_j)| \le 4000\Delta mK \ln q / (q\rho)$ for all $i, j$. Therefore, we get \[l \le \displaystyle\sum_{1 \le i \le j \le m'/m}\frac{4000\Delta mK \ln q}{q\rho} \le \left(\frac{m'}{m}\right)^2\frac{4000\Delta mK \ln q}{q\rho} < \Delta mK\]
which is a contradiction.
\end{proof}

\begin{corollary}\label{union}
For any $0 < \epsilon < 1/18$, there exists an instance of SSBVE with $N$ vertices, or equivalently, an instance of Minimum $p$-Union with $O(N)$ sets and $O(N)$ elements in the universe, that demonstrates an integrality gap of $\Omega(N^{1/18 - \epsilon})$ for the level $N^{\Omega(\epsilon)}$ SoS relaxation.
\end{corollary}

\begin{proof}
The corollary follows from the above theorem by setting $D = 3, q = N^{1/18 - \epsilon / 2}$ and $\rho = \epsilon / 1000$.
\end{proof}

\section{Pseudocalibration}

Barak et al.\cite{pseudo} developed pseudocalibration, a heuristic to construct integrality gaps for SoS relaxations in a structured manner. In various works that have appeared on SoS lower bounds since this manuscript was first published, pseudocalibration has become a fundamental technique in this field, see e.g. \cite{ghosh2020sum, mohanty2020lifting, potechin2020machinery, sparsesos, rajendran2022nonlinear}. We will describe the heuristic and show its applications to construct integrality gaps for Planted Clique and Max K-CSP.

To explain it, we first need the notion of pseudoexpectation, which presents a dual view of the SoS hierarchy that will give us more insight. This view will be very useful for constructing integrality gaps.

\subsection{Pseudoexpectations}

Let $P^{\le r}[x_1, \ldots, x_n]$ be the set of polynomials of degree at most $r$ in $\mathbb{R}[x_1, \ldots, x_n]$. A degree $2r$ pseudoexpectation operator $\tilde{E}$ is a function from $P^{\le 2r}[x_1, \ldots, x_n]$ to $\mathbb{R}$ that satisfies the following conditions.
\begin{itemize}
	\item $\tilde{E}[1] = 1$
	\item $\tilde{E}$ is linear, that is, for any two polynomials $p, q$ of degree at most $2r$, we have $\tilde{E}(\alpha p + \beta q) = \alpha\tilde{E}(p) + \beta\tilde{E}(q)$ for all $\alpha, \beta \in \mathbb{R}$.
	\item For every polynomial $p$ of degree at most $r$, $\tilde{E}[p^2] \ge 0$
\end{itemize}

We will now show that the existence of SoS vectors with some desired objective value is equivalent, up to a constant factor in the number of levels, to the existence of a pseudoexpectation operator with the same objective value. We will show this for a slightly restricted system where we do not allow inequalities but it holds in general even if we have inequalities. This duality allows us to work with pseudoexpectation operators instead of SoS vectors to construct integrality gaps.

Consider the problem $\Gamma$ of maximizing a polynomial $p(x_1, \ldots, x_n)$ over boolean variables $x_1, \ldots, x_n \in \{0, 1\}$ subject to $q_i(x_1, \ldots, x_n) = 0$ for $i = 1, 2, \ldots, m$. Since $x_i$ are boolean, assume without loss of generality that $p, q_i$ are multilinear. For all $T \subseteq [n]$, denote $\displaystyle\prod_{i \in T} x_i$ by $\bm{x}_T$ and for any multilinear polynomial $h$, for all $T \subseteq [n]$, denote the corresponding coefficient of $h$ by $h_T$, that is, $h = \displaystyle\sum_{T \subseteq [n]} h_T \bm{x}_T$. Suppose $p, q_1, \ldots, q_m$ have degree at most $r$, then $p = \displaystyle\sum_{T \in [n]_{\le r}} p_T \bm{x}_T$ and $q_i = \displaystyle\sum_{T \in [n]_{\le r}} (q_i)_T \bm{x}_T$. The SoS relaxation for $r$ levels, which we denote by $\mathcal{P}_{r}$, is the following program:
\begin{align*}
\text{Maximize}\qquad&\displaystyle\sum_{T \in [n]_{\le r}}p_T\norm{\bm{V}_T}^2&&\\
\text{subject to}\qquad&\displaystyle\sum_{T \in [n]_{\le r}}(q_i)_T\langle \bm{V}_T, \bm{V}_S\rangle = 0 & \forall& S \in [n]_{\le r}, i = 1, \ldots, m\\
&\langle \bm{V}_{S_1}, \bm{V}_{S_2}\rangle = \langle \bm{V}_{S_3}, \bm{V}_{S_4}\rangle &\forall& S_1 \cup S_2 = S_3 \cup S_4\text{ and }S_i \in [n]_{\le r}\\
&\langle \bm{V}_{S_1}, \bm{V}_{S_2}\rangle \ge 0 &\forall& S_1, S_2 \in [n]_{\le r}\\
&\norm{\bm{V}_{\phi}}^2 = 1&&
\end{align*}
%\begin{eqnarray*}
%	&\text{Maximize }\displaystyle\sum_{T \in [n]_{\le r}}p_T\norm{\bm{V}_T}^2&\\
%	\text{subject to }&\displaystyle\sum_{T \in [n]_{\le r}}(q_i)_T\langle \bm{V}_T, \bm{V}_S\rangle = 0 & \forall S \in [n]_{\le r}, i = 1, \ldots, m\\
%	&\langle \bm{V}_{S_1}, \bm{V}_{S_2}\rangle = \langle \bm{V}_{S_3}, \bm{V}_{S_4}\rangle &\forall S_1 \cup S_2 = S_3 \cup S_4 \in [n]_{\le r}\\
%	&\langle \bm{V}_{S_1}, \bm{V}_{S_2}\rangle \ge 0 &\forall S_1, S_2 \in [n]_{\le r}\\
%	& \norm{\bm{V}_{\phi}}^2 = 1&\\
%\end{eqnarray*}

Now, consider the following program which optimizes over degree $2r$ pseudoexpectation operators $\tilde{E}$, which we denote by $\mathcal{Q}_{2r}$: Let $H_i = \{h(x_1, \ldots, x_n) \;|\; q_ih \in P^{\le 2r}[x_1, \ldots, x_n]\}$.
\begin{align*}
\text{Maximize}~~&\tilde{E}[p(x_1, \ldots, x_n)]&&\\
\text{subject to}~~&\tilde{E}[q_i(x_1, \ldots, x_n)h(x_1, \ldots, x_n)] = 0\qquad \forall h \in H_i, i = 1, 2, \ldots, m&\\
&\tilde{E}[(x_i^2 - x_i)h(x_1, \ldots, x_n)] = 0 \qquad\qquad~\forall h \in P^{\le 2r - 2}[x_1, \ldots, x_n], i = 1, 2, \ldots, n&\\
&\tilde{E}\text{ is a degree }2r\text{ pseudoexpectation operator}&
\end{align*}
%\begin{eqnarray*}
%	&\text{Maximize }\tilde{E}[p(x_1, \ldots, x_n)]&\\
%	\text{subject to }&\tilde{E}[q_i(x_1, \ldots, x_n)h(x_1, \ldots, x_n)] = 0 &\forall h \in H_i, i = 1, 2, \ldots, m\\
%	& \tilde{E}[(x_i^2 - x)h(x_1, \ldots, x_n)] = 0 &\forall h \in P^{\le 2r - 2}[x_1, \ldots, x_n], i = 1, 2, \ldots, n\\
%	&\tilde{E}\text{ is a degree 2r pseudoexpectation operator}
%\end{eqnarray*}

Here, we enforce $\tilde{E}[q_ih] = 0$ for all polynomials $h$ such that $\tilde{E}[q_ih]$ is defined and also enforce $\tilde{E}[(x_i^2 - x_i)h] = 0$ for all $h$ such that $\tilde{E}[(x_i^2 - x_i)h]$ is defined. And under these constraints, we try to optimize $\tilde{E}[p]$.

\begin{theorem}
For $\Gamma$, if $\mathcal{P}_{2r}$ has a feasible solution of objective value $FRAC$, then there exists a feasible solution for $\mathcal{Q}_{2r}$ with objective value $FRAC$.
\end{theorem}

\begin{proof}
Let $\{\bm{V}_S\}_{S \in [n]_{\le 2r}}$ be the level $2r$ SoS vectors that achieve objective value $FRAC$. For any polynomial $h \in P^{\le 2r}[x_1, \ldots, x_n]$, denote by $\overline{h}$ the multilinearization of the polynomial $h$, which means $\overline{h}$ is obtained from $h$ by syntactically replacing any occurence of $x_i^k$ in any term of $h$ by $x_i$ for any $i \le n, k \ge 2$. So, using the assumption that $p, q_i$ are multilinear, we have $p_T = \overline{p}_T, (q_i)_T = (\overline{q_i})_T$. For any polynomial $h \in P^{\le 2r}[x_1, \ldots, x_n]$, define $\tilde{E}[h] = \displaystyle\sum_{T \in [n]_{\le 2r}} \overline{h}_T \langle \bm{V}_{\phi}, \bm{V}_T\rangle$.

First, observe that this operator is well defined and linear. We have $\tilde{E}[1] = \norm{\bm{V}_{\phi}}^2 = 1$. For any $h \in P^{\le 2r - 2}[x_1, \ldots, x_n]$, $\tilde{E}[(x_i^2 - x_i)h]$ is $0$ by definition of $\tilde{E}$. For any $i \le m$, to prove that $\tilde{E}[q_ih] = 0$ for all $h$ such that $q_ih \in P^{\le 2r}[x_1, \ldots, x_n]$, by linearity, it suffices to prove that $\tilde{E}[q_ih] = 0$ for all $h = \bm{x}_S$ with $deg(q_i) + deg(h) \le 2r$, but in that case, we have $\tilde{E}[q_ih] = \displaystyle\sum_{T \in [n]_{\le 2r}}(\overline{q_i})_T\tilde{E}[\bm{x}_{T \cup S}] = \displaystyle\sum_{T \in [n]_{\le 2r}}(\overline{q_i})_T\langle \bm{V}_{\phi}, \bm{V}_{T \cup S}\rangle = \displaystyle\sum_{T \in [n]_{\le 2r}}(\overline{q_i})_T\langle \bm{V}_{T}, \bm{V}_{S}\rangle = 0$. Here, note that $|T \cup S| \le 2r$ by degree conditions.

We need to prove that $\tilde{E}[h^2] \ge 0$ for all polynomials $h \in P^{\le r}[x_1, \ldots, x_n]$. We can again assume $h$ is multilinear by the definition of $\tilde{E}$. Then
\begin{align*}
\tilde{E}[h(x_1, \ldots, x_n)^2] &= \displaystyle\sum_{T_1 \subseteq [n]_{\le r}} \displaystyle\sum_{T_2 \subseteq [n]_{\le r}} h_{T_1}h_{T_2}\tilde{E}[\bm{x}_{T_1 \cup T_2}]\\
&= \displaystyle\sum_{T_1 \subseteq [n]_{\le r}} \displaystyle\sum_{T_2 \subseteq [n]_{\le r}} h_{T_1}h_{T_2}\langle \bm{V}_{\phi}, \bm{V}_{T_1 \cup T_2}\rangle\\
&= \displaystyle\sum_{T_1 \subseteq [n]_{\le r}} \displaystyle\sum_{T_2 \subseteq [n]_{\le r}} h_{T_1}h_{T_2}\langle \bm{V}_{T_1}, \bm{V}_{T_2}\rangle\\
&= \norm{\displaystyle\sum_{T \subseteq [n]_{\le r}} h_T\bm{V}_T}^2 \ge 0
\end{align*}

Finally, observe that $\tilde{E}[p(x_1 \ldots, x_n)] = \displaystyle\sum_{T \in [n]_{\le 2r}} \overline{p}_T\langle \bm{V}_{\phi}, \bm{V}_T\rangle = \displaystyle\sum_{T \in [n]_{\le 2r}} p_T\norm{\bm{V}_T}^2 = FRAC$.
\end{proof}

In particular, we get that the optimum value of $\mathcal{Q}_{2r}$ is at least the optimum value of $\mathcal{P}_{2r}$.

\begin{theorem}
For $\Gamma$, if $\mathcal{Q}_{4r}$ has a feasible solution of objective value $FRAC$, then there exists a feasible solution for $\mathcal{P}_{r}$ with objective value $FRAC$.
\end{theorem}

\begin{proof}
Let $\tilde{E}$ be the degree $4r$ pseudoexpectation operator with $\tilde{E}[p(x_1, \ldots, x_n)] = FRAC$. Consider the $n^{O(r)} \times n^{O(r)}$ matrix $M$ with rows and columns indexed by elements of $[n]_{\le r}$ such that $M_{S, T} = \tilde{E}[\bm{x}_{S \cup T}]$ for all $S, T \in [n]_{\le r}$. Clearly, $M_{S, T}$ is symmetric. We have $\tilde{E}[\bm{x}_{T_1}\bm{x}_{T_2}] = \tilde{E}[\bm{x}_{T_1\cup T_2}]$ for all $T_1, T_2 \in [n]_{\le r}$ because $\tilde{E}[x_i^2h] = \tilde{E}[x_ih]$ for all $h \in P^{\le 4r-2}[x_1,\ldots, x_n]$. So, we get that for any vector $\bm{v} \in \mathbb{R}^{[n]_{\le r}}$, we have $\bm{v}^TM\bm{v} = \displaystyle\sum_{T_1 \in [n]_{\le r}} \displaystyle\sum_{T_2 \in [n]_{\le r}} \bm{v}_{T_1}\bm{v}_{T_2}\tilde{E}[\bm{x}_{T_1 \cup T_2}] = \tilde{E}[(\displaystyle\sum_{T \in [n]_{\le r}}\bm{v}_T\bm{x}_T)^2] \ge 0$. This means that $M$ is positive semidefinite and therefore, there exist vectors $\bm{V}_S$ for $S \in [n]_{\le r}$ such that $\langle \bm{V}_S, \bm{V}_T\rangle = \tilde{E}[\bm{x}_{S \cup T}]$ for all $S, T \in [n]_{\le r}$. We will prove that these vectors give a feasible solution to $\mathcal{P}_r$ with objective value $FRAC$.

We have $\norm{\bm{V}_{\phi}}^2 = \tilde{E}[1] = 1$. And for $S_1, S_2, S_3, S_4\in [n]_{\le r}$ such that $S_1 \cup S_2 = S_3 \cup S_4$, we have $S_1 \cup S_2, S_3 \cup S_4 \in [n]_{\le 2r}$ which means $\tilde{E}[\bm{x}_{S_1 \cup S_2}], \tilde{E}[\bm{x}_{S_3 \cup S_4}], \tilde{E}[\bm{x}_{S_1 \cup S_2}^2]$ are defined and so, $\langle \bm{V}_{S_1}, \bm{V}_{S_2}\rangle = \tilde{E}[\bm{x}_{S_1 \cup S_2}] = \tilde{E}[\bm{x}_{S_3 \cup S_4}] = \langle \bm{V}_{S_3}, \bm{V}_{S_4}\rangle$. Also, $\langle \bm{V}_{S_1}, \bm{V}_{S_2}\rangle = \tilde{E}[\bm{x}_{S_1 \cup S_2}] = \tilde{E}[\bm{x}_{S_1 \cup S_2}^2] \ge 0$.

For all $i \le m$ and $S \in [n]_{\le r}$, $\displaystyle\sum_{T \in [n]_{\le r}}(q_i)_T\langle \bm{V}_T, \bm{V}_S\rangle = \displaystyle\sum_{T \in [n]_{\le r}}(q_i)_T\tilde{E}[\bm{x}_{T\cup S}] = \displaystyle\sum_{T \in [n]_{\le r}}(q_i)_T\tilde{E}[\bm{x}_T\bm{x}_S] = \tilde{E}[q_i(x_1, \ldots, x_n)\bm{x}_S] = 0$. Finally, we have the objective value $\displaystyle\sum_{T \in [n]_{\le r}}p_T\norm{\bm{V}_T}^2 = \displaystyle\sum_{T \in [n]_{\le r}}p_T\tilde{E}[\bm{x}_T] = \tilde{E}[p(x_1, \ldots, x_n)] = FRAC$.
\end{proof}

In particular, we get that the optimum value of $\mathcal{P}_{r}$ is at least the optimum value of $\mathcal{Q}_{4r}$.

\subsection{Maximum Clique}

An instance of Maximum Clique is a graph $G = (V, E)$ and the objective is to find the size of the largest clique in $G$. The basic program has boolean variables $x_u$ for $u \in V$ where $x_u$ indicates whether $u$ is in the largest clique:
\begin{eqnarray*}
	\text{Maximize }&\displaystyle\sum_{u \in V} x_u&\\
	\text{subject to }&x_ux_v = 0& \forall (u, v) \not\in E, u \neq v\\
	&x_u \in \{0, 1\}&
\end{eqnarray*}

Note that the constraint means that if $(u, v)$ is not an edge, then both $u, v$ are not picked in the final solution and vice versa. So, this program precisely solves the Maximum Clique problem.

In the previous chapter, we studied approximation guarantees of the SoS relaxation of this problem on Erd\"{o}s-R\'{e}nyi random graphs. Now, we study integrality gaps for the relaxation. The integrality gap construction by Barak et al.\cite{pseudo} are Erd\"{o}s-R\'{e}nyi random graphs $G \sim G(n, 1/2)$ which is a graph $G = (V, E)$ on $n$ vertices where for each $u \neq v$, the edge $(u, v)$ is present in $E$ with probability $1/2$. In such graphs, it can be shown that there are no cliques of size more than $2 \log n$ with high probability.

%Barak et al.\cite{pseudo} prove that the integrality gap of a level $r$ SoS relaxation is large, for any $r  = o(\log n)$, for $G \sim G_{n, 1/2}$ with high probability. So, in a sense, this shows that this average case problem is also intrisically hard that even SoS relaxations for a certain number of levels cannot give good approximations.

\begin{theorem}[\cite{pseudo}]\label{clique}
For any $r = o(\log n)$, the optimum value of the level $r$ SoS relaxation for maximum clique on $G\sim G(n, 1/2)$ is at least $k = n^{1 / 2 - O(\sqrt{r / \log n})}$ with high probability.
\end{theorem}

Since the actual optimum is $O(\log n)$, this shows that the integrality gap is large for $r = o(\log n)$ levels of SoS. On the other hand, a simple bruteforce algorithm that checks whether any $2\log n + 1$ vertices will form a clique, will run in time $n^{O(\log n)}$ and find the maximum clique for random instances with high probability.

Their proof proceeds by constructing a degree $r$ pseudoexpectation operator that witnesses this. The result would follow from the equivalence between the SoS hierarchy and the pseudoexpectation view.

Their argument proceeds in two parts, where they first use heuristics to mathematically construct the pseudoexpectation operator and then, in the second part, they prove that it satisfies the required properties. The first part is known as pseudocalibration, which we will describe here. We will skip the latter part, which is technically involved.

To be precise, for $r = o(\log n)$, we will exhibit a degree $2r$ pseudoexpectation operator $\tilde{E}$ that satisfies the following conditions with high probability when $G = (V, E)$ (assume $V = [n]$) is sampled from $G(n, 1/2)$.
\begin{itemize}
	\item $\tilde{E}$ is linear and $\tilde{E}[1] = 1$
	\item $\tilde{E}[(x_u^2 - x_u)h(x_1, \ldots, x_n)] = 0$ for all $h \in P^{\le 2r - 2}[x_1, \ldots, x_n], u = 1, \ldots, n$
	\item $\tilde{E}[x_ux_vh(x_1, \ldots, x_n)] = 0$ for all $(u, v) \not\in E, u \neq v, h \in P^{\le 2r - 2}[x_1, \ldots, x_n]$
	\item $\displaystyle\sum_{u = 1}^n \tilde{E}[x_u] = k$
	\item $\tilde{E}[h(x_1, \ldots, x_n)^2] \ge 0$ for all $h \in P^{\le r}[x_1, \ldots, x_n]$.
\end{itemize}

The idea is think of $\tilde{E}$ as a computationally bounded solver. We are trying to determine $\tilde{E}$ that will, in loose terms, think that $G(n, 1/2)$ has a clique of size $k$ for $k \gg 2\log n$. The crucial heuristic is to consider a planted version of the random graph and try to estimate the values of $\tilde{E}$ assuming that it cannot distinguish a planted version from a purely random graph. More precisely, consider the following two distributions
\begin{itemize}
	\item $G(n, 1/2)$ - A graph $G$ sampled from the Erd\"{o}s-R\'{e}nyi random graph distribution.
	\item $G(n, 1/2, k)$ - Sample a graph $G \sim G(n, 1/2)$, choose a subset of $k$ vertices uniformly at random and add all possible edges, if not already present, within this subset. We call this the planted version.
\end{itemize}
The intuition that $\tilde{E}$ is unable to distinguish these two distributions should mean in particular that for any function $f:\mathbb{R}^n \longrightarrow \mathbb{R}$ of degree at most $2r$ on the variables $x_1, \ldots, x_n$, the expected value of the pseudoexpectation of this function is the same for both distributions. That is, $\mathbb{E}_{G \sim G(n, 1/2)}\tilde{E}_G[f] = \mathbb{E}_{G \sim G(n, 1/2, k)}\tilde{E}_G[f]$ for all functions $f \in P^{\le 2r}[x_1, \ldots, x_n]$. Here, note that $\tilde{E}_G$ can depend on the graph $G$, which we emphasize by a subscript.

We take this further with the following heuristic and make a stronger assumption. Fix $f \in P^{\le 2r}[x_1, \ldots, x_n]$ and consider $\tilde{E}_G[f]$ as a function of the graph $G$. We assume that, not just the expectation but also, the correlation of $\tilde{E}_G[f]$ with any low degree function $g$ on graphs $G$ is the same for both distributions. We will describe the exact definition of low degree later. To make this formal, if we encode the edges of the graph using $\dbinom{n}{2}$ entries $G_e$ in $\{\pm 1\}$ where $G_e = 1$ means the edge $e$ is present and $G_e = -1$ means the edge $e$ is absent, we can treat $\tilde{E}[f]$ as a function from $\{\pm 1\}^{n(n-1)/2}$ to $\mathbb{R}$. Then, for all low degree functions $g:\{\pm 1\}^{n(n-1)/2}\longrightarrow\mathbb{R}$, we set the correlations to be the same for both distributions, namely \[\mathbb{E}_{G \sim G(n, 1/2)}[\tilde{E}_G[f]g(G)] = \mathbb{E}_{G \sim G(n, 1/2, k)}[\tilde{E}_G[f]g(G)]\]

Now, since $G \sim G(n, 1/2, k)$ does indeed have a $k$-clique, we heuristically assume that $\tilde{E}$ is the correct expectation on this graph, with a unique support being the indicator vector $\bm{x} \in \mathbb{R}^n$ of the planted clique (but in reality, there can be other cliques) and that $\tilde{E}_G$ only errs on $G(n, 1/2)$. Then, \[\mathbb{E}_{G \sim G(n, 1/2, k)}[\tilde{E}_G[f]g(G)] = \mathbb{E}_{(G, \bm{x}) \sim G(n, 1/2, k)}[f(\bm{x})g(G)]\] where we use the notation $(G, \bm{x}) \sim G(n, 1/2, k)$ to mean that $G \sim G(n, 1/2, k)$ and $\bm{x}$ is the indicator vector of the planted $k$-clique.

So, $\tilde{E}_G$ would ideally satisfy \[\mathbb{E}_{G \sim G(n, 1/2)}[\tilde{E}_G[f]g(G)] = \mathbb{E}_{(G, \bm{x}) \sim G(n, 1/2, k)}[f(\bm{x})g(G)]\] for all functions $f \in P^{\le 2r}[x_1, \ldots, x_n]$ and low degree $g : \{\pm 1\}^{n(n-1)/2} \longrightarrow \mathbb{R}$. From discrete Fourier analysis on boolean variables, note that $f$ of degree at most $2r$ can be written as a linear combination of the functions $\bm{x}_S:\mathbb{R}^n \longrightarrow\mathbb{R}$ for $S \in [n]_{\le 2r}$ where $\bm{x}_S(\bm{x}) = \displaystyle\prod_{i \in S} x_i$, and $g$ can be written as a linear combination of the functions $\chi_T : \{\pm 1\}^{n(n-1)/2} \longrightarrow \mathbb{R}$ for $T \subseteq [n(n-1)/2]$ where $\chi_T(G) = \displaystyle\prod_{e\in T} G_e$. So, it suffices to ensure \[\mathbb{E}_{G \sim G(n, 1/2)}[\tilde{E}_G[\bm{x}_S]\chi_T(G)] = \mathbb{E}_{(G, \bm{x}) \sim G(n, 1/2, k)}[\bm{x}_S(\bm{x})\chi_T(G)]\] for all $S \in [n]_{\le 2r}$ and low degree $T \subseteq [n(n-1)/2]$ and the condition we wish to ensure will follow from linearity of the pseudoexpectation and expectation. In fact, we make this assumption only for $S, T$ such that $|S \cup V(T)| \le \tau$ for some threshold $\tau$, where $V(T)$ is the set of vertices contained in the edges in $T$. The reason that we consider $|S \cup V(T)|$ will be clear when we compute the Fourier coefficients of $\tilde{E}_G[\bm{x}_S]$.  Barak et al. set $\tau \approx r / \epsilon$ where $k \approx n^{1/2 - \epsilon}$.

Remember that we are trying to determine $\tilde{E}_G[f]$ for $f \in P^{\le 2r}[x_1, \ldots, x_n]$ that will satisfy our constraints for graphs $G \sim G(n, 1/2)$ with high probability. Think of it as a function of $G$ and by the preceding comments, it suffices to determine $\tilde{E}_G[\bm{x}_S]$ for all $S$ of size at most $2r$. For a fixed $S$, since this is a function on graphs $G$, it has a Fourier expansion $\tilde{E}_G[\bm{x}_S] = \displaystyle\sum_{T \subseteq [n(n-1)/2]}\reallywidehat{\tilde{E}[\bm{x}_S](T)} \chi_T(G)$.

The final heuristic is to assume that $\reallywidehat{\tilde{E}[\bm{x}_S](T)} = 0$ for all subsets $T$ such that $|S \cup V(T)| > \tau$. The intuitive reason for this assumption is that the function $\tilde{E}[\bm{x}_S]$ is computed by an algorithm that runs in $n^{O(r)}$ time and hence, has to be simple upto $n^{O(r)}$ complexity. One way of interpreting this is to assume that the higher order Fourier coefficients vanish.

After we use these heuristics, we compute the remaining Fourier coefficients. For $S, T$ such that $|S| \le 2r$ and $|V(T) \cup S| \le \tau$, we have
\[\reallywidehat{\tilde{E}[\bm{x}_S](T)} = \mathbb{E}_{G \sim G(n, 1/2)}[\tilde{E}_G[\bm{x}_S]\chi_T(G)] = \mathbb{E}_{(G, \bm{x}) \sim G(n, 1/2, k)}[\bm{x}_S(\bm{x})\chi_T(G)]\]
Let $\mathcal{C} \subseteq [n]$ be the planted clique in $G$ where $(G, \bm{x}) \sim G(n, 1/2, k)$. If $\mathcal{C} \not\supseteq S$, then $\bm{x}_S(\bm{x}) = 0$ and if $\mathcal{C} \not\supseteq V(T)$, we have $\mathbb{E}_{(G, \bm{x}) \sim G(n, 1/2, k)}[\bm{x}_S(\bm{x})\chi_T(G)] = 0$ since there is an edge $e$ in $T$ that is outside the planted clique and hence, $G_e$ would be $1$ or $-1$ with probability $1/2$ each. And when $C \supseteq S \cup V(T)$, we have $\bm{x}_S(\bm{x})\chi_T(G) = 1$. So, we get
\begin{align*}
\reallywidehat{\tilde{E}[\bm{x}_S](T)} &= \mathbb{E}_{G \sim G(n, 1/2)}[\tilde{E}_G[\bm{x}_S]\chi_T(G)]\\
&= \text{Pr}[\mathcal{C}\text{ contains }S \cup V(T)]\\
&= \frac{\dbinom{n - |S \cup V(T)|}{k - |S \cup V(T)|}}{\dbinom{n}{k}}\\
&\approx \left(\frac{k}{n}\right)^{|S \cup V(T)|}
\end{align*}
So, in general, if $f(\bm{x}) = \displaystyle\sum_{S \in [n]_{\le 2r}} c_S \bm{x}_S$ is any polynomial in $P^{\le 2r}[x_1, \ldots, x_n]$, then we have \[\tilde{E}[f] = \displaystyle\sum_{S \in [n]_{\le 2r}} c_S \displaystyle\sum_{|S \cup V(T)| \le \tau, T \subseteq [n(n-1)/2]} \left(\frac{k}{n}\right)^{|S \cup V(T)|} \chi_T(G)\]
for the graph $G$.

This is the pseudoexpectation that was used to prove Theorem \ref{clique}. The polynomial constraints follow from concentration bounds and proving positivity of the operator was the main technical contribution of the paper.

\subsection{Max K-CSP}\label{pcalibration}

We now show some ingredients towards proving Theorem \ref{cspsos}, more specifically the integrality gap construction of Kothari et al.\cite{kmow} for SoS relaxations of Max $K$-CSP. The approach taken in their paper is purely combinatorial, but we will show that we can arrive at the same construction via pseudocalibration.

We will follow the terminology from section \ref{first} of this chapter. For simplicity of exposition, we will consider boolean predicates, that is, $q = 2$ with the alphabet being $\{-1, 1\}$ (instead of $\{0, 1\}$) and we will also assume $\tau = 3$, which means $\mathcal{C} \subseteq \{-1, 1\}^k$ supports a pairwise uniform distribution. For the $i$th constraint $C_i$, let the shift vector be denoted $b_i = (b_{i, 1}, \ldots, b_{i, K}) \in \{-1, 1\}^K$. So, the $i$th constraint $C_i$ on the appropriate subset of variables $x_{C_i} = (x_j)_{j \in C_i}$ is $[\text{Is }x_{C_i} \cdot b_i \in \mathcal{C}?]$, where "$\cdot$" denotes entrywise product.

Since $q = 2$, we can consider an equivalent, but simpler program than the one we constructed in Chapter $2$. We let the basic variables be $x_j \in \{-1, 1\}$. For each constraint $C_i$, we can arithmetize it into a polynomial expression $f_i(x_1, \ldots, x_n)$ of degree $K$, such that for any assignment of $x_j \in \{-1, 1\}$, the evaluation of $f_i$ on this assignment is $0$ if the respective assignment satisfies the constraint and $1$ otherwise. Indeed $f_i$ does not contain $x_j$ for $j \not\in C_i$. Since we are aiming to show perfect completeness, we can look at the feasibility problem where each constraint is perfectly satisfied. So, we wish to find $x_j$ such that $x_j^2 = 1$ for all $j \le n$ and $f_i(x_1, \ldots, x_n) = 0$ for all $i \le m$.

%Recall that from the previous chapter, each variable in the basic program is of the form $y_{i, \alpha_i}$ which indicates whether variable $x_i$ has been assigned $\alpha_i \in \{-1, +1\}$. For each constraint $C_i$, we can arithmetize it into a polynomial expression on the $y_{i, \alpha_i}$ for all $i \in C_i$, of degree at most $K$, that indicates whether the respective assignment satisfies the constraint.

By the equivalence shown between the SoS hierarchy and pseudoexpectation operators, it suffices to obtain a degree $2r = \zeta\eta n / 3$ pseudoexpectation operator $\tilde{E}$ such that the following conditions are satisfied with high probability for a random Max $K$-CSP instance $I$ as per definition \ref{rand}.
%(Here, we take random instances but actually, instances that satisfy the Plausibility assumption are enough).
\begin{itemize}
	\item $\tilde{E}$ is linear and $\tilde{E}[1] = 1$
	\item $\tilde{E}[(x_j^2 - 1)h(x_1, \ldots, x_n)] = 0$ for all $h \in P^{\le 2r - 2}[x_1, \ldots, x_n], j = 1, \ldots, n$
	\item $\tilde{E}[f_i(x_1, \ldots, x_n)h(x_1, \ldots, x_n)] = 0$ for $h \in P^{\le 2r - K}[x_1, \ldots, x_n], i = 1, \ldots, m$
	\item $\tilde{E}[h(x_1, \ldots, x_n)^2] \ge 0$ for all $h \in P^{\le r}[x_1, \ldots, x_n]$.
\end{itemize}
%\begin{itemize}
%	\item $\tilde{E}$ is linear and $\tilde{E}[1] = 1$
%	\item $\tilde{E}[(y_{i, \alpha_i}^2 - y_{i, \alpha_i})h(x)] = 0$ for all $h \in P^{\le 2r - 2}[x_1, \ldots, x_n], \alpha_i \in \{-1, 1\}, i = 1, \ldots, n$
%	\item $\tilde{E}[(y_{i, 1} + y_{i, -1} - 1)h(x)] = 0$ for all $h \in P^{\le 2r - 1}[x_1, \ldots, x_n], i = 1, \ldots, n$
%	\item $\tilde{E}[y_{i, 1}y_{i, -1} h(x)] = 0$ for all $h \in P^{\le 2r - 2}[x_1, \ldots, x_n], i = 1, \ldots, n$
%	\item $\tilde{E}[h^2] \ge 0$ for all $h \in P^{\le r}[x_1, \ldots, x_n]$.
%\end{itemize}

Now, we will fix the structure of the instance, that is, the factor graph. So, we know the variables involved in each clause $C_i$, let $C_i = (x_{t_{i, 1}}, \ldots, x_{t_{i, K}})$. Our only degrees of freedom will be the shift vectors $b_i$ for $i \le m$. Similar to the case of Maximum Clique, we will assume that $\tilde{E}$ cannot distinguish the following two distributions.
\begin{itemize}
	\item $\mu_r$ - For each clause $C_i$, sample $b_{i, 1}, \ldots, b_{i, K}$ from $\{-1, 1\}$ independently and uniformly.
	\item $\mu_p$ - Sample a global assignment $(y_1, \ldots, y_n) \in \{-1, 1\}^n$ uniformly at random. Then, independently for each clause $C_i$, sample $(z_{i, 1}, \ldots, z_{i, K})$ from $C \subseteq \mathbb{F}_2^K$ uniformly and set $b_{i, j} = y_{t_{i, j}}z_{i, j}$ for all $j = 1, 2, \dots, K$.
\end{itemize}

The intuitive reason we chose $\mu_p$ like that is because we would want some distribution very similar to $\mu_r$ so that distinguishing is hard, yet it should have some globally satisfying assignment for our pseudocalibration heuristic to work. Now, if $\tilde{E}$ is unable to distinguish $\mu_r$ from $\mu_p$, then, for any $f:\mathbb{R}^n \longrightarrow \mathbb{R}$ of degree at most $2r$ over $x_1, \ldots, x_n$, the expected value of the pseudoexpectations over $b_{i, j}$ should be the same, that is, $\mathbb{E}_{b \sim \mu_r}\tilde{E}_b[f] = \mathbb{E}_{b \sim \mu_p}\tilde{E}_b[f]$. Also, for a fixed $f \in P^{\le 2r}[x_1, \ldots, x_n]$, consider $\tilde{E}_b[f]$ as a function of the $b_{i, j}$. We assume that, since $\tilde{E}$ is unable to distinguish $\mu_r$ from $\mu_p$, the correlation of $\tilde{E}_b[f]$ with any low degree function $g$ of the $b_{i, j}$ is the same in both distributions, that is, \[\mathbb{E}_{b \sim \mu_r}[\tilde{E}_b[f]g(b)] = \mathbb{E}_{b \sim \mu_p}[\tilde{E}_b[f]g(b)]\]

When $b \in \mu_p$, there is an actual satisfying assignment $(y_1, \ldots, y_n)$. In that case, we assume that $\tilde{E}$ is the correct expectation, with a unique support being this assignment. Then,
\[\mathbb{E}_{b \sim \mu_p}[\tilde{E}_b[f]g(b)] = \mathbb{E}_{(b, y, z) \sim \mu_p}[f(y)g(b)]\]
where we use the notation $(b, y, z) \sim \mu_p$ to mean that, when we sampled $b$ from $\mu_p$, the global assignment is $(y_1, \ldots, y_n)$ and for each clause $C_i$, the sampled element from $C$ is $(z_{i, 1}, \ldots, z_{i, K})$.

So, we want $\tilde{E}$ to satisfy
\[\mathbb{E}_{b \sim \mu_r}[\tilde{E}_b[f]g(b)] = \mathbb{E}_{(b, y, z) \sim \mu_p}[f(y)g(b)]\]

We can think of $g$ as a function from $\{-1, 1\}^{mK}$ to $\mathbb{R}$ since there are $mK$ different $b_{i, j}$s. From discrete Fourier analysis, it is enough to satisfy this equation for all $f = \bm{x}_S:\mathbb{R}^n\longrightarrow \mathbb{R}$ for $S \in [n]_{\le 2r}$ where $\bm{x}_S(x_1, \ldots, x_n) = \displaystyle\prod_{i \in S} x_i$ and $g = \chi_T:\{-1, 1\}^{mK}\longrightarrow\mathbb{R}$ for some $T \subseteq \{(i, j) \;|\; i \le m, j \le K\} = [m]\times [K]$ where $\chi_T(b) = \displaystyle\prod_{(i, j) \subseteq T} b_{i, j}$. The assumption becomes
\[\mathbb{E}_{b \sim \mu_r}[\tilde{E}_b[\bm{x}_S]\chi_T(b)] = \mathbb{E}_{(b, y, z) \sim \mu_p}[\bm{x}_S(y_1, \ldots, y_n)\chi_T(b)]\]

Recall that we are trying to determine $\tilde{E}_b[f]$ for $f \in P^{\le 2r}[x_1, \ldots, x_n]$. For a fixed $S$ of size at most $2r$, we have the Fourier expansion $\tilde{E}_b[\bm{x}_S] = \displaystyle\sum_{T \subseteq [m]\times [K]}\reallywidehat{\tilde{E}[\bm{x}_S](T)} \chi_T(b)$. Let's compute the Fourier coefficients. For $S, T$ such that $S \in [n]_{\le 2r}$, $T \subseteq [m]\times [K]$, we have
\begin{align*}
\reallywidehat{\tilde{E}[\bm{x}_S](T)} &= \mathbb{E}_{b \sim \mu_r}[\tilde{E}_b[\bm{x}_S]\chi_T(b)]\\
&= \mathbb{E}_{(b, y, z) \sim \mu_p}[\bm{x}_S(y_1, \ldots, y_n)\chi_T(b)]\\
&= \mathbb{E}_{(b, y, z) \sim \mu_p}[\displaystyle\prod_{i \in S} y_i \displaystyle\prod_{(i, j) \in T} (y_{t_{i, j}}z_{i, j})]
\end{align*}

Note that any subset $T \subseteq [m]\times [K]$ can be thought of to be a collection of edges of the factor graph $G_I$. So, $T$ corresponds to an unique edge induced subgraph $H_T$ of $G_I$. If $H_T$ contains any variable vertex $x_j$ of odd degree outside $S$, then observe that the expectation above becomes $0$ because $y_j$ would occur an odd number of times in right hand side and it is chosen uniformly from $\{-1, 1\}$. Similarly, if any constraint vertex $C_i$ in $H_T$ has degree at most $2$, then the expectation above becomes $0$, since $C$ is pairwise independent and the choice of $z_{i, j}$  for this $i$ is independent of the other terms in the product. So, the only nonzero Fourier coefficients correspond to subgraphs $H_T$ of $G_I$ such that every constraint vertex in $H_T$ has degree at least $3$ and every variable vertex of $H_T$ with odd degree in $H_T$ is inside $S$.

The approach taken in \cite{kmow} was a bit more direct. They view the pseudoexpectations using the idea of local distributions. It is known that if $\tilde{E}$ is a degree $2r$ pseudoexpectation, then for any subset of $r$ variables $x_j$, there exists an actual probability distribution on them, whose true expectation matches $\tilde{E}$. Motivated by prior work by Razborov et al.\cite{raz}, Bennabas et al.\cite{bgmt}, etc., if $S$ denotes a set of variables $x_j$, Kothari et al. consider a larger set containing both variables and constraints, called the closure of $S$. Then, they define $\tilde{E}[\bm{x}_S]$ to be the actual expectation of a locally satisfying assignment on the closure. Under some assumptions, this locally satisfying assignment can be shown to exist. The closure of $S$ is defined to be the union of all subgraphs $H$ of $G_I$ for which all the constraint vertices have degree at least $3$, every leaf vertex of $H$ is inside $S$ and the number of constraint vertices in $H$ is at most $\eta n$.

Different definitions of closure of a set of variables have been studied before but one of the main contributions of $\cite{kmow}$ was this new definition of closure. Their motivation was to define it in such a way that it contains all the variables and constraints that, loosely speaking, affect the set $S$. Here, we show that this definition is also motivated by our computation of the Fourier coefficients. In particular, out of all the Fourier coefficients that are nonzero, we consider only the Fourier coefficients $T$ for which all the constraint vertices of the subgraph $H_T$ have degree at least $3$, every variable vertex of $H_T$ which has degree $1$ (a leaf variable vertex of $H_T$) is inside $S$ and the number of constraint vertices in $H_T$ is at most $\eta n$. We set all the remaining Fourier coefficients to $0$.

Similar to the case of Maximum Clique, the hardest part of proving that this construction works is proving positivity of the pseudoexpectation operator, which we will not cover here.

\subsection{Other applications}

In follow-up works, we use similar techniques (with various generalizations and modifications) to obtain high degree SoS lower bounds for the following problems
\begin{itemize}
    \item Sherrington-Kirkpatrick Hamiltonian \cite{ghosh2020sum}: An important problem from statistical physics that aims to certify the ground energy of a spin system.
    \item Tensor PCA \cite{potechin2020machinery}: A useful statistical and computational technique that exploits higher order moments of the data, used for various problems such as latent variable modeling, topic modeling, community detection, etc., see e.g. \cite{anandkumar2014tensor, kivva2021learning, hsu2012spectral, anandkumar2014guaranteed, richard2014statistical} and references therein.
    \item Sparse PCA \cite{potechin2020machinery}: A fundamental primitive in machine learning that extracts principal components from data, with applications in a diverse range of fields, e.g. \cite{wang2012online, naikal2011informative, tan2014classification, allen2011sparse}.
    \item Planted Slightly Denser Subgraph \cite{potechin2020machinery}: A problem closely related to the Maximum Clique problem that is studied in this thesis. This problem can be thought of as a harder version of the standard Densest $k$-subgraph problem (hence easier to show lower bounds for).
    \item Sparse Independent Set \cite{sparsesos}: A fundamental graph theoretic problem whose SoS lower bounds are fascinating because sparse random graphs are studied (as opposed to most works that study dense models) and a modified pseudocalibration heuristic is used in the work.
\end{itemize}

For a more up to date summary of pseudocalibration and its connections to other fields, see \cite{rajendran2022nonlinear}.

%% file: chap4.tex
\chapter{Future Work}

As we saw, hierarchies form a unified approach to optimization problems. It is natural to consider Sum of Squares relaxations for other problems of interest and prove approximation guarantees as well as tight integrality gaps.

\subsection{Approximability}

Guruswami and Sinop\cite{gs} round SoS solutions and get good approximation guarantees for low threshold-rank graphs. We also know that SoS achieves good approximation for other classes of graphs such as $K_r$-minor free graphs but the analysis proceeds differently. Essentially, it is because these graphs structurally admit decompositions into graphs of bounded diameter, see for instance \cite{lapchilau}. Naturally, it would be interesting to unify the above two results and identify a larger class of graphs that preferably contains the above two classes, for which the natural SoS relaxation provably gives a good approximation.

For the Densest $k$-subgraph problem, we have an approximation guarantee of $n^{1/4 + \epsilon}$ for $O(1 / \epsilon)$ levels of the Lov\'{a}sz-Schrijver hierarchy\cite{dks} for a graph on $n$ vertices and hence, the Sum of Squares Hierarchy also gives the same guarantee. The analysis of this algorithm roughly proceeds by considering subgraphs of small size with a special structure known as caterpillar graphs and arguing that dense graphs have lots of them. The motivation for this algorithm comes from the algorithm for the related problem of distinguishing a random graph from a planted graph, where we simply count the number of caterpillar subgraphs. It is an open problem to simplify their analysis by trying to understand exactly which polynomial that SoS considers in making the distinction, if one exists. This would also make it possible to generalize this idea to analyze the SoS relaxation of the Densest $k$-subhypergraph problem.

\subsection{Inapproximability}

On the lower bound front, the best known integrality gap for the polynomial level SoS relaxation for Densest $k$-subgraph is $n^{1/14 - \epsilon}$(\cite{bcv}, \cite{pasin}). It is an open problem to improve this gap possibly using a different construction and it is plausible that the actual integrality gap is $n^{1/4}$, which would also be tight. It is known that the level $\Omega(\log n / \log \log n)$ Sherali-Adams relaxation for this problem has an integrality gap of $\tilde{\Omega}(n^{1/4})$\cite{bcv} which provides extra evidence to the truth of this gap.

For the Densest $k$-subhypergraph problem of arity $\rho$, the integrality gaps we obtained seem far from optimal and we conjecture that the actual integrality gap is $\Omega(n^{(\rho - 1) / 4})$. As remarked earlier, we also do not know approximation guarantees for the SoS relaxation for this problem. In particular, the currently known analysis for Densest $k$-subgraph\cite{dks} does not seem to easily extend for hypergraphs.

For Minimum $p$-Union, it was shown in \cite{yury} that the level $\Omega(\epsilon \log m / \log \log m)$ Sherali-Adams relaxation has an integrality gap of $O(m^{1/4 - \epsilon})$. And they also proved that, assuming the hypergraph extension of a conjecture known as "Dense versus Random", we can obtain a $m^{1/4}$ hardness of approximation. So, a natural first step would be to prove this lower bound for the restricted Sum of Squares hierarchy without any assumptions. Since this problem can be thought of as a more general version of Densest $k$-subgraph, it should plausibly be easier to prove lower bounds. Also, in our integrality gap for Minimum $p$-Union, our construction and proof can be modified to work for Smallest $m$-Edge subgraph, which is a restricted version of Minimum $p$-Union. So, it seems that we could further utilize the flexibility of the problem's input to our advantage.

It is possible to apply pseudocalibration to systematically construct integrality gaps for the SoS relaxations of these problems, but it is still not clear how to analyze them. Pseudocalibration was employed by Chlamt\'{a}\v{c} and Manurangsi\cite{logdensity} to obtain Sherali-Adams integrality gaps for $\tilde{\Omega}(\log n)$ levels for Densest $k$-subgraph, Smallest $m$-Edge subgraph, their hypergraph variants and Minimum $p$-Union, all through a common framework.

%% file: biblio.tex
%% This defines the bibliography file (main.bib) and the bibliography style.
%% If you want to create a bibliography file by hand, change the contents of
%% this file to a `thebibliography' environment.  For more information 
%% see section 4.3 of the LaTeX manual.
\begin{singlespace}
\bibliography{main}
\bibliographystyle{alpha}
\end{singlespace}

%% file: main.bbl
\newcommand{\etalchar}[1]{$^{#1}$}
\begin{thebibliography}{KMOW17}

\bibitem[AGH{\etalchar{+}}14]{anandkumar2014tensor}
Animashree Anandkumar, Rong Ge, Daniel Hsu, Sham~M Kakade, and Matus Telgarsky.
\newblock Tensor decompositions for learning latent variable models.
\newblock {\em Journal of machine learning research}, 15:2773--2832, 2014.

\bibitem[AGJ14]{anandkumar2014guaranteed}
Animashree Anandkumar, Rong Ge, and Majid Janzamin.
\newblock Guaranteed non-orthogonal tensor decomposition via alternating
  rank-$1 $ updates.
\newblock {\em arXiv preprint arXiv:1402.5180}, 2014.

\bibitem[AL17]{lapchilau}
Vedat~Levi Alev and Lap~Chi Lau.
\newblock Approximating unique games using low diameter graph decomposition.
\newblock In {\em Approximation, randomization, and combinatorial optimization.
  {A}lgorithms and techniques}, volume~81 of {\em LIPIcs. Leibniz Int. Proc.
  Inform.}, pages Art. No. 18, 15. Schloss Dagstuhl. Leibniz-Zent. Inform.,
  Wadern, 2017.

\bibitem[AMS11]{allen2011sparse}
Genevera~I Allen and Mirjana Maleti{\'c}-Savati{\'c}.
\newblock Sparse non-negative generalized pca with applications to
  metabolomics.
\newblock {\em Bioinformatics}, 27(21):3029--3035, 2011.

\bibitem[AOW15]{aow}
Sarah~R. Allen, Ryan O'Donnell, and David Witmer.
\newblock How to refute a random {CSP}.
\newblock In {\em 2015 {IEEE} 56th {A}nnual {S}ymposium on {F}oundations of
  {C}omputer {S}cience---{FOCS} 2015}, pages 689--708. IEEE Computer Soc., Los
  Alamitos, CA, 2015.

\bibitem[BCC{\etalchar{+}}10]{dks}
Aditya Bhaskara, Moses Charikar, Eden Chlamt\'a\v{c}, Uriel Feige, and
  Aravindan Vijayaraghavan.
\newblock Detecting high log-densities---an {$O(n^{1/4})$} approximation for
  densest {$k$}-subgraph.
\newblock In {\em S{TOC}'10---{P}roceedings of the 2010 {ACM} {I}nternational
  {S}ymposium on {T}heory of {C}omputing}, pages 201--210. ACM, New York, 2010.

\bibitem[BCG{\etalchar{+}}12]{bcv}
Aditya Bhaskara, Moses Charikar, Venkatesan Guruswami, Aravindan
  Vijayaraghavan, and Yuan Zhou.
\newblock Polynomial integrality gaps for strong {SDP} relaxations of
  {$\sf{Densest}$} {$k$}-{$\sf{subgraph}$}.
\newblock In {\em Proceedings of the {T}wenty-{T}hird {A}nnual {ACM}-{SIAM}
  {S}ymposium on {D}iscrete {A}lgorithms}, pages 388--405. ACM, New York, 2012.

\bibitem[BGMT12]{bgmt}
Siavosh Benabbas, Konstantinos Georgiou, Avner Magen, and Madhur Tulsiani.
\newblock S{DP} gaps from pairwise independence.
\newblock {\em Theory Comput.}, 8:269--289, 2012.

\bibitem[Bha97]{bhatia}
Rajendra Bhatia.
\newblock {\em Matrix analysis}, volume 169 of {\em Graduate Texts in
  Mathematics}.
\newblock Springer-Verlag, New York, 1997.

\bibitem[BHK{\etalchar{+}}16]{pseudo}
Boaz Barak, Samuel~B. Hopkins, Jonathan Kelner, Pravesh Kothari, Ankur Moitra,
  and Aaron Potechin.
\newblock A nearly tight sum-of-squares lower bound for the planted clique
  problem.
\newblock In {\em 57th {A}nnual {IEEE} {S}ymposium on {F}oundations of
  {C}omputer {S}cience---{FOCS} 2016}, pages 428--437. IEEE Computer Soc., Los
  Alamitos, CA, 2016.

\bibitem[BRS11]{brs}
Boaz Barak, Prasad Raghavendra, and David Steurer.
\newblock Rounding semidefinite programming hierarchies via global correlation.
\newblock In {\em 2011 {IEEE} 52nd {A}nnual {S}ymposium on {F}oundations of
  {C}omputer {S}cience---{FOCS} 2011}, pages 472--481. IEEE Computer Soc., Los
  Alamitos, CA, 2011.

\bibitem[CDK{\etalchar{+}}16]{dksh}
Eden Chlamt\'a\v{c}, Michael Dinitz, Christian Konrad, Guy Kortsarz, and George
  Rabanca.
\newblock The densest {$k$}-subhypergraph problem.
\newblock In {\em Approximation, randomization, and combinatorial optimization.
  {A}lgorithms and techniques}, volume~60 of {\em LIPIcs. Leibniz Int. Proc.
  Inform.}, pages Art. No. 6, 19. Schloss Dagstuhl. Leibniz-Zent. Inform.,
  Wadern, 2016.

\bibitem[CDM17]{yury}
Eden Chlamt\'a\v{c}, Michael Dinitz, and Yury Makarychev.
\newblock Minimizing the union: {T}ight approximations for small set bipartite
  vertex expansion.
\newblock In {\em Proceedings of the {T}wenty-{E}ighth {A}nnual {ACM}-{SIAM}
  {S}ymposium on {D}iscrete {A}lgorithms}, pages 881--899. SIAM, Philadelphia,
  PA, 2017.

\bibitem[CM18]{logdensity}
Eden Chlamt\'a\v{c} and Pasin Manurangsi.
\newblock Sherali-{A}dams integrality gaps matching the log-density threshold.
\newblock Unpublished Manuscript, 2018.

\bibitem[Fei98]{coverage}
Uriel Feige.
\newblock A threshold of {$\ln n$} for approximating set cover.
\newblock {\em J. ACM}, 45(4):634--652, 1998.

\bibitem[FK81]{furedi}
Z.~F\"uredi and J.~Koml\'os.
\newblock The eigenvalues of random symmetric matrices.
\newblock {\em Combinatorica}, 1(3):233--241, 1981.

\bibitem[FK03]{probable}
Uriel Feige and Robert Krauthgamer.
\newblock The probable value of the {L}ov\'asz-{S}chrijver relaxations for
  maximum independent set.
\newblock {\em SIAM J. Comput.}, 32(2):345--370, 2003.

\bibitem[FS02]{feige}
Uriel Feige and Gideon Schechtman.
\newblock On the optimality of the random hyperplane rounding technique for
  {MAX} {CUT}.
\newblock {\em Random Structures Algorithms}, 20(3):403--440, 2002.
\newblock Probabilistic methods in combinatorial optimization.

\bibitem[GJJ{\etalchar{+}}20]{ghosh2020sum}
Mrinalkanti Ghosh, Fernando~Granha Jeronimo, Chris Jones, Aaron Potechin, and
  Goutham Rajendran.
\newblock Sum-of-squares lower bounds for sherrington-kirkpatrick via planted
  affine planes.
\newblock In {\em 2020 IEEE 61st Annual Symposium on Foundations of Computer
  Science (FOCS)}, pages 954--965. IEEE, 2020.

\bibitem[GLS88]{oracle}
Martin Gr\"otschel, L\'aszl\'o Lov\'asz, and Alexander Schrijver.
\newblock {\em Geometric algorithms and combinatorial optimization}, volume~2
  of {\em Algorithms and Combinatorics: Study and Research Texts}.
\newblock Springer-Verlag, Berlin, 1988.

\bibitem[GS11]{gs}
Venkatesan Guruswami and Ali~Kemal Sinop.
\newblock Lasserre hierarchy, higher eigenvalues, and approximation schemes for
  graphs partitioning and quadratic integer programming with {PSD} objectives
  (extended abstract).
\newblock In {\em 2011 {IEEE} 52nd {A}nnual {S}ymposium on {F}oundations of
  {C}omputer {S}cience---{FOCS} 2011}, pages 482--491. IEEE Computer Soc., Los
  Alamitos, CA, 2011.

\bibitem[GS12]{gspre}
Venkatesan Guruswami and Ali~Kemal Sinop.
\newblock Optimal column-based low-rank matrix reconstruction.
\newblock In {\em Proceedings of the {T}wenty-{T}hird {A}nnual {ACM}-{SIAM}
  {S}ymposium on {D}iscrete {A}lgorithms}, pages 1207--1214. ACM, New York,
  2012.

\bibitem[GT14]{gt}
Shayan~Oveis Gharan and Luca Trevisan.
\newblock Partitioning into expanders.
\newblock In {\em Proceedings of the {T}wenty-{F}ifth {A}nnual {ACM}-{SIAM}
  {S}ymposium on {D}iscrete {A}lgorithms}, pages 1256--1266. ACM, New York,
  2014.

\bibitem[GW95]{gw}
Michel~X. Goemans and David~P. Williamson.
\newblock Improved approximation algorithms for maximum cut and satisfiability
  problems using semidefinite programming.
\newblock {\em J. Assoc. Comput. Mach.}, 42(6):1115--1145, 1995.

\bibitem[H{\aa}s96]{hastad}
Johan H{\aa}stad.
\newblock Clique is hard to approximate within {$n^{1-\epsilon}$}.
\newblock In {\em 37th {A}nnual {S}ymposium on {F}oundations of {C}omputer
  {S}cience ({B}urlington, {VT}, 1996)}, pages 627--636. IEEE Comput. Soc.
  Press, Los Alamitos, CA, 1996.

\bibitem[HKZ12]{hsu2012spectral}
Daniel Hsu, Sham~M Kakade, and Tong Zhang.
\newblock A spectral algorithm for learning hidden markov models.
\newblock {\em Journal of Computer and System Sciences}, 78(5):1460--1480,
  2012.

\bibitem[JPR{\etalchar{+}}21]{sparsesos}
Chris Jones, Aaron Potechin, Goutham Rajendran, Madhur Tulsiani, and Jeff Xu.
\newblock Sum-of-squares lower bounds for sparse independent set.
\newblock {\em IEEE 62nd Annual Symposium on Foundations of Computer Science
  (FOCS)}, 2021.

\bibitem[Juh82]{juhasz}
Ferenc Juh{\'a}sz.
\newblock The asymptotic behaviour of {L}ov{\'a}sz' {$\vartheta$} function for
  random graphs.
\newblock {\em Combinatorica}, 2(2):153--155, Jun 1982.

\bibitem[KKMO07]{kkmo}
Subhash Khot, Guy Kindler, Elchanan Mossel, and Ryan O'Donnell.
\newblock Optimal inapproximability results for {MAX}-{CUT} and other
  2-variable {CSP}s?
\newblock {\em SIAM J. Comput.}, 37(1):319--357, 2007.

\bibitem[KMOW17]{kmow}
Pravesh~K. Kothari, Ryuhei Mori, Ryan O'Donnell, and David Witmer.
\newblock Sum of squares lower bounds for refuting any {CSP}.
\newblock In {\em S{TOC}'17---{P}roceedings of the 49th {A}nnual {ACM} {SIGACT}
  {S}ymposium on {T}heory of {C}omputing}, pages 132--145. ACM, New York, 2017.

\bibitem[KOS17]{globaltolocal}
Pravesh Kothari, Ryan O'Donnell, and Tselil Schramm.
\newblock {SoS} lower bounds for hard constraints: Think global, act local.
\newblock Unpublished Manuscript, 2017.

\bibitem[KP06]{ponnuswami}
Subhash Khot and Ashok~Kumar Ponnuswami.
\newblock Better inapproximability results for {M}ax{C}lique, chromatic number
  and {M}in-3{L}in-{D}eletion.
\newblock In {\em Automata, languages and programming. {P}art {I}}, volume 4051
  of {\em Lecture Notes in Comput. Sci.}, pages 226--237. Springer, Berlin,
  2006.

\bibitem[KRRA21]{kivva2021learning}
Bohdan Kivva, Goutham Rajendran, Pradeep Ravikumar, and Bryon Aragam.
\newblock Learning latent causal graphs via mixture oracles.
\newblock {\em Advances in Neural Information Processing Systems}, 34, 2021.

\bibitem[KV02]{kriv}
Michael Krivelevich and Van~H. Vu.
\newblock Approximating the independence number and the chromatic number in
  expected polynomial time.
\newblock {\em J. Comb. Optim.}, 6(2):143--155, 2002.

\bibitem[Las01]{las}
Jean~B. Lasserre.
\newblock Global optimization with polynomials and the problem of moments.
\newblock {\em SIAM J. Optim.}, 11(3):796--817, 2000/01.

\bibitem[Lov79]{lov}
L\'aszl\'o Lov\'asz.
\newblock On the {S}hannon capacity of a graph.
\newblock {\em IEEE Trans. Inform. Theory}, 25(1):1--7, 1979.

\bibitem[Lov09]{geomrep}
L{\'a}szl{\'o} Lov{\'a}sz.
\newblock Geometric representations of graphs.
\newblock 2009.

\bibitem[LS91]{ls}
L.~Lov\'asz and A.~Schrijver.
\newblock Cones of matrices and set-functions and {$0$}-{$1$} optimization.
\newblock {\em SIAM J. Optim.}, 1(2):166--190, 1991.

\bibitem[Man15]{pasin}
Pasin Manurangsi.
\newblock On approximating projection games.
\newblock Master's thesis, Massachusetts Institute of Technology, 2015.

\bibitem[MRX20]{mohanty2020lifting}
Sidhanth Mohanty, Prasad Raghavendra, and Jeff Xu.
\newblock Lifting sum-of-squares lower bounds: degree-2 to degree-4.
\newblock In {\em Proceedings of the 52nd Annual ACM SIGACT Symposium on Theory
  of Computing}, pages 840--853, 2020.

\bibitem[Nes00]{nesterov}
Yurii Nesterov.
\newblock Squared functional systems and optimization problems.
\newblock In {\em High performance optimization}, volume~33 of {\em Appl.
  Optim.}, pages 405--440. Kluwer Acad. Publ., Dordrecht, 2000.

\bibitem[NYS11]{naikal2011informative}
Nikhil Naikal, Allen~Y Yang, and S~Shankar Sastry.
\newblock Informative feature selection for object recognition via sparse pca.
\newblock In {\em 2011 International Conference on Computer Vision}, pages
  818--825. IEEE, 2011.

\bibitem[Par03]{parr}
Pablo~A. Parrilo.
\newblock Semidefinite programming relaxations for semialgebraic problems.
\newblock {\em Math. Program.}, 96(2, Ser. B):293--320, 2003.
\newblock Algebraic and geometric methods in discrete optimization.

\bibitem[PR20]{potechin2020machinery}
Aaron Potechin and Goutham Rajendran.
\newblock Machinery for proving sum-of-squares lower bounds on certification
  problems.
\newblock {\em arXiv preprint arXiv:2011.04253}, 2020.

\bibitem[Raj22]{rajendran2022nonlinear}
Goutham Rajendran.
\newblock {\em Nonlinear {R}andom {M}atrices and {A}pplications to the {S}um of
  {S}quares {H}ierarchy}.
\newblock PhD thesis, University of Chicago, 2022.

\bibitem[Raz98]{raz}
Alexander~A. Razborov.
\newblock Lower bounds for the polynomial calculus.
\newblock {\em Comput. Complexity}, 7(4):291--324, 1998.

\bibitem[RM14]{richard2014statistical}
Emile Richard and Andrea Montanari.
\newblock A statistical model for tensor pca.
\newblock In {\em Advances in Neural Information Processing Systems}, pages
  2897--2905, 2014.

\bibitem[RRS17]{rrs}
Prasad Raghavendra, Satish Rao, and Tselil Schramm.
\newblock Strongly refuting random {CSP}s below the spectral threshold.
\newblock In {\em S{TOC}'17---{P}roceedings of the 49th {A}nnual {ACM} {SIGACT}
  {S}ymposium on {T}heory of {C}omputing}, pages 121--131. ACM, New York, 2017.

\bibitem[SA90]{sa}
Hanif~D. Sherali and Warren~P. Adams.
\newblock A hierarchy of relaxations between the continuous and convex hull
  representations for zero-one programming problems.
\newblock {\em SIAM J. Discrete Math.}, 3(3):411--430, 1990.

\bibitem[Sho87]{shor}
Naum~Zuselevich Shor.
\newblock An approach to obtaining global extremums in polynomial mathematical
  programming problems.
\newblock {\em Cybernetics}, 23(5):695--700, 1987.

\bibitem[TPW14]{tan2014classification}
Kean~Ming Tan, Ashley Petersen, and Daniela Witten.
\newblock Classification of rna-seq data.
\newblock In {\em Statistical analysis of next generation sequencing data},
  pages 219--246. Springer, 2014.

\bibitem[Tul09]{tul}
Madhur Tulsiani.
\newblock C{SP} gaps and reductions in the {L}asserre hierarchy [extended
  abstract].
\newblock In {\em S{TOC}'09---{P}roceedings of the 2009 {ACM} {I}nternational
  {S}ymposium on {T}heory of {C}omputing}, pages 303--312. ACM, New York, 2009.

\bibitem[WLY12]{wang2012online}
Dong Wang, Huchuan Lu, and Ming-Hsuan Yang.
\newblock Online object tracking with sparse prototypes.
\newblock {\em IEEE transactions on image processing}, 22(1):314--325, 2012.

\end{thebibliography}
